\documentclass[11pt,a4paper]{article}

\usepackage{amsmath,amsthm,amsfonts,amssymb}
\usepackage[a4paper,margin=2cm]{geometry}
\usepackage{colonequals}
\usepackage{graphicx,tikz}
\usepackage{ifthen}
\usepackage{authblk}

\newboolean{ElectronicVersion}
\setboolean{ElectronicVersion}{true}

\ifthenelse{\boolean{ElectronicVersion}}{
    \usepackage[a4paper=true,pdftex,bookmarks,pagebackref,
	plainpages=false, 
        pdfpagelabels=true 
        ]{hyperref}}{}

\makeatletter
\newtheorem*{rep@theorem}{\rep@title}
\newcommand{\newreptheorem}[2]{%
\newenvironment{rep#1}[1]{%
 \def\rep@title{#2 \ref*{##1}}%
 \begin{rep@theorem}}%
 {\end{rep@theorem}}}
\makeatother

\usepackage{hyperref}
\hypersetup{
    bookmarksnumbered=true, 
    unicode=false, 
    pdfstartview={FitH}, 
    pdftitle={Quantum Subroutine Composition}, 
    pdfauthor={Stacey Jeffery}, 
    pdfsubject={}, 
    pdfcreator={}, 
    pdfproducer={}, 
    pdfkeywords={}, 
    pdfnewwindow=true, 
    colorlinks=true, 
    linkcolor=blue, 
    citecolor=blue, 
    filecolor=blue, 
    urlcolor=blue 
}

\newcommand{\eq}[1]{\hyperref[eq:#1]{(\ref*{eq:#1})}}
\renewcommand{\sec}[1]{\hyperref[sec:#1]{Section~\ref*{sec:#1}}}
\newcommand{\thm}[1]{\hyperref[thm:#1]{Theorem~\ref*{thm:#1}}}
\newcommand{\lem}[1]{\hyperref[lem:#1]{Lemma~\ref*{lem:#1}}}
\newcommand{\cor}[1]{\hyperref[cor:#1]{Corollary~\ref*{cor:#1}}}
\newcommand{\app}[1]{\hyperref[app:#1]{Appendix~\ref*{app:#1}}}
\newcommand{\tabl}[1]{\hyperref[tab:#1]{Table~\ref*{tab:#1}}}
\newcommand{\defin}[1]{\hyperref[def:#1]{Definition~\ref*{def:#1}}}
\newcommand{\fig}[1]{\hyperref[fig:#1]{Figure~\ref*{fig:#1}}}
\newcommand{\clm}[1]{\hyperref[clm:#1]{Claim~\ref*{clm:#1}}}
\newcommand{\conj}[1]{\hyperref[conj:#1]{Conjecture~\ref*{conj:#1}}}
\newcommand{\rem}[1]{\hyperref[rem:#1]{Remark~\ref*{rem:#1}}}

\newcommand{\thmthm}[2]{\hyperref[thm:#1]{Theorem~\ref*{thm:#1}} and~\hyperref[thm:#2]{\ref*{thm:#2}}}
\newcommand{\lemlem}[2]{\hyperref[lem:#1]{Lemma~\ref*{lem:#1}} and~\hyperref[lem:#2]{\ref*{lem:#2}}}

\newtheorem{theorem}{Theorem}[section]
\newtheorem{lemma}[theorem]{Lemma}

\newtheorem{corollary}[theorem]{Corollary}

\newtheorem{claim}[theorem]{Claim}

\newtheorem{definition}[theorem]{Definition}

\def\ket#1{{\lvert}#1\rangle}
\def\bra#1{{\langle}#1\rvert}

\def\braket#1#2{{{\langle}#1\vert}#2\rangle}

\def\norm#1{\left\| #1 \right\|}

\newcommand{\eps}{\varepsilon}

\def\w{{\sf w}}
\def\q{{\sf q}}
\def\M{{\sf M}}

\title{Quantum Subroutine Composition}
\author[{}]{Stacey Jeffery\thanks{This work is supported by ERC STG grant 101040624-ASC-Q, NWO Klein project number OCENW.Klein.061, and ARO contract no W911NF2010327. SJ is a CIFAR Fellow in the Quantum Information Science Program.}}
\affil[{}]{CWI \& QuSoft, the Netherlands}

\begin{document}

\maketitle

\begin{abstract}
An important tool in algorithm design is the ability to build algorithms from other algorithms that run as subroutines. In the case of quantum algorithms, a subroutine may be called on a superposition of different inputs, which complicates things. For example, a classical algorithm that calls a subroutine ${\sf Q}$ times, where the average probability of querying the subroutine on input $i$ is $\bar{p}_i$, and the cost of the subroutine on input $i$ is a random variable $T_i$, incurs expected cost ${\sf Q}\sum_i\bar{p}_i \mathbb{E}[T_i]$ from all subroutine queries. While this statement is obvious for classical algorithms, for quantum algorithms, it is much less so, since naively, if we run a quantum subroutine on a superposition of inputs, we need to wait for all branches of the superposition to terminate before we can apply the next operation. We nonetheless show an analogous quantum statement (*): If $\bar{q}_i$ is the average \emph{query weight on $i$} over all queries, the cost from all quantum subroutine queries is ${\sf Q}\sum_i\bar{q}_i \mathbb{E}[T_i]$. Here the query weight on $i$ for a particular query is the probability of measuring $i$ in the input register if we were to measure right before the query. 

We prove this result using the technique of multidimensional quantum walks, recently introduced in [Jeffery, Zur 2022]. We present a more general version of their quantum walk edge composition result, which yields \emph{variable-time quantum walks}, generalizing variable-time quantum search, by, for example, replacing the update cost with $\sqrt{\sum_{u,v}\pi_u P_{u,v} \mathbb{E}[T_{u,v}^2]}$, where $T_{u,v}$ is the cost to move from vertex $u$ to vertex $v$. The same technique that allows us to compose quantum subroutines in quantum walks can also be used to compose in any quantum algorithm, which is how we prove (*). 
\end{abstract}

\section{Introduction}\label{sec:intro}

Classical algorithms make extensive use of subroutines, which allow an algorithm designer to build up an algorithm in modular blocks. Analysis of such algorithms can also be done in a simple, modular way. Suppose an outer classical algorithm makes use of some subroutine at ${\sf Q}$ steps of the algorithm, and suppose the subroutine's stopping time on input $i$ is some random variable $T_i$. Suppose $\bar{p}_i$ is the average probability, over all ${\sf Q}$ queries to the subroutine, that the subroutine is called on input $i$ -- so ${\sf Q}\bar{p}_i$ is the
expected number of times the subroutine is called on input $i$, and $\sum_i\bar{p}_i=1$. Let ${\sf L}$ denote the total number of \emph{other} operations used by the outer algorithm. Then the expected running time of the classical algorithm is:
\begin{equation}\label{eq:intro-classical-alg}
{\sf L}+{\sf Q}\sum_i \bar{p}_i \mathbb{E}[T_i].
\end{equation}
We can get a bounded-error algorithm with this asymptotic running time by simply stopping after 10 times the expected number of steps, if the algorithm has not already halted. At that point we can output an arbitrary (likely wrong) answer if the algorithm hasn't already halted, but this happens with probability at most some small constant. 

While subroutines are also quite useful in quantum algorithms, it is not at all obvious how to do the above classical analysis in a quantum setting. Suppose an outer quantum algorithm makes use of some subroutine. In general, the subroutine will be called on a superposition of inputs $i$, so if some branches of the superposition are finished early, it is not obvious how to exploit this in the overall complexity. Generally, if ${\sf T}_{\max}$ is the maximum time we must wait for the quantum subroutine to terminate, the total complexity of the quantum algorithm is:
\begin{equation}\label{eq:intro-naive}
O({\sf L}+{\sf Q}\cdot {\sf T}_{\max}),
\end{equation}
where again, ${\sf Q}$ is the number of queries to the subroutine, and ${\sf L}$ is the number of other operations made by the outer algorithm. Our main result is to give a quantum analogue of \eq{intro-classical-alg} for the special case where each subroutine outputs a bit. Given a quantum algorithm that decides $f$ using ${\sf L}$ basic operations, as well as ${\sf Q}$ queries to a subroutine that 
runs in time $T_i$ on input $i$ to compute some bit\footnote{We can think of the bit $g_i$ as being $g_i(x)$ for some function $g_i$, and $x$ the full input to the composed algorithm.} $g_i$, there is a bounded error quantum algorithm that computes $f(g)$ with complexity (neglecting log factors):
$${\sf L}+{\sf Q}\sum_i\bar{q}_i\mathbb{E}[T_i],$$
where $\bar{q}_i$ is the \emph{average query weight on $i$} -- the squared norm on $\ket{i}$ averaged over all queries, so in particular, $\sum_{i}\bar{q}_i=1$. We describe the result in more detail in \sec{intro-alg}. While such a result is obvious in classical algorithms, it is much less so in quantum algorithms. In fact, if the outer algorithm and subroutine are zero-error algorithms, and we want to compose them to get a zero-error algorithm with this \emph{expected} running time, while again obvious in the classical case, this is \emph{not} possible in general for quantum algorithms~\cite{buhrman2003zeroerror}. 

If $\mathbb{E}[T_i]=\mu$ is a known constant (in $i$, meaning $\mu$ does not depend on $i$, though $T_i$ could), then this result is much less interesting: we can always stop the subroutine after $10\mu$ steps, and by Markov's inequality, this introduces at most $1/10$ additional probability of error (which we can reduce through logarithmic repetitions). However, if the values $\mathbb{E}[T_i]$ vary significantly in $i$, in contrast to the classical case, it is not obvious (to the best of our knowledge) that this result should hold for quantum algorithms, and special cases of this result have been the subject of significant effort. 

For example, consider evaluating an unbalanced formula that is the OR of $n$ ANDs, of arity $k_1,\dots,k_n$:
$$f(x^{(1)},\dots,x^{(n)}) = \mathrm{OR}_n(\mathrm{AND}_{k_1}(x^{(1)}),\dots,  \mathrm{AND}_{k_n}(x^{(n)})),$$
where $x^{(j)}\in\{0,1\}^{k_j}$. Naively, using nested quantum search, one could search for a $j$ such that AND$_{k_j}(x^{(j)})=1$, where 
AND$_{k_j}(x^{(j)})$ can be evaluated in $\sqrt{k_j}$ steps, for total cost \mbox{$\sqrt{n}\cdot\max_{j\in [n]}\sqrt{k_j}$} (we neglect log factors in this paragraph), but this may be far from optimal. However, using highly non-trivial techniques, it has been shown how to evaluate any such formula in $\sqrt{n}\sqrt{\frac{1}{n}\sum_{j=1}^nk_j}$ quantum time
or \emph{any} AND-OR formula on $N$ bits in $\sqrt{N}$ time~\cite{childs2007NANDformula,reichardt2009GameTree}. 

This non-trivial upper bound for OR of ANDs can be seen as a special case of a technique called \emph{variable-time quantum search}. 
Consider a setting in which a quantum subroutine runs in time ${\sf T}_i$ on input $i$, and one wants to decide if there exists an $i\in [n]$ on which the subroutine outputs 1. Classically, it is easy to see that this can be decided in cost $\sum_{i\in [n]}{\sf T}_i$. Naive use of quantum search gives an upper bound of $\sqrt{n}{\sf T}_{\max}$, where ${\sf T}_{\max}=\max_{i\in [n]}{\sf T}_i$. Ambainis showed, in a highly non-trivial way, an upper bound of $\widetilde{O}\left(\sqrt{\sum_{i=1}^n{\sf T}_i^2}\right)$~\cite{ambainis2010VTSearch}. This result assumes that if the $i$-th subroutine consists of unitaries $U_1^i,\dots,U_{{\sf T}_{\max}}^i$, the operator $\sum_{t=1}^{{\sf T}_{\max}}\ket{t}\bra{t}\otimes U_t^i$ can be implemented in ``unit cost''. We will make this assumption as well, throughout this paper, and use ``unit cost'' to describe a cost we are willing to accept as a multiplicative factor on all complexities (for example, $O(1)$  or polylogarithmic in some natural variable). This assumption is only reasonable if the unitaries each have unit cost. It does not hold for strict gate complexity when $U_1^i,\dots,U_{{\sf T}_{\max}}^i$ are arbitrary gates (local unitaries), but it holds, for example, in the \emph{fully quantum} QRAM model if $U_1^i,\dots,U_{{\sf T}_{\max}}^i$ are stored as a list of gates in QCRAM\footnote{Classical memory to which a quantum computer has read-only superposition access, sometimes referred to as ``QRAM'' in previous literature.}, or 
if they satisfy certain uniformity conditions (see~\cite[Section 2.2]{jeffery2022kDist}). 

In~\cite{ambainis2010VTAA}, Ambainis considers a slightly different setting, where a quantum subroutine runs in time $T$, where $T$ is a random variable on $\{1,\dots,{\sf T}_{\max}\}$ -- that is, the algorithm may stop at different times with different probabilities, outputting either a 0 or 1 in some answer register whenever it stops.  Ambainis shows how to amplify the 1 part of the computation in complexity (neglecting log factors):
$${\sf T}_{\max}+\sqrt{\frac{\mathbb{E}[T^2]}{\eps}}$$
where $\eps$ is a lower bound on the probability that the algorithm outputs 1 (assuming this is non-zero).

We consider a combination of Ambainis' two settings, where a subroutine computes some bit $g_i$ for any input $i$, and the time it takes is a random variable $T_i$ on $\{1,\dots,{\sf T}_{\max}\}$. When the expected stopping times $\mathbb{E}[T_i]$ are unknown, and $\eps$ is a lower bound on the probability in uniform input $i$ that the subroutine outputs 1, we obtain (neglecting log factors)
\begin{equation}\label{eq:intro-VT1}
\sqrt{\frac{1}{n\eps}\sum_{i\in[n]}\mathbb{E}[T_i^2]}.
\end{equation}
When the expected stopping times $\mathbb{E}[T_i]$ are known, we can replace $\mathbb{E}[T_i^2]$ with $\mathbb{E}[T_i]^2$, which is better when $T_i$ has large variance, and improve, but not remove, the polylogarithmic dependence on ${\sf T}_{\max}$.

We also show two alternative variable-time search results that may be better in certain settings. 
Our variable-time search results are a special case of a new technique: variable-time quantum walk algorithms. Loosely speaking, if the cost of moving from a vertex $u$ to a vertex $v$ in a random walk on a graph $G$ is some random variable $T_{u,v}$ on $\{1,\dots,{\sf T}_{\max}\}$, previous quantum walk search algorithms would have incurred a multiplicative overhead on the complexity of ${\sf T}_{\max}$. We show how to improve this overhead to $\sqrt{\sum_{\{u,v\}\in E(G)}\pi(u)P_{u,v}\mathbb{E}[T_{u,v}^2]}$, where $P$ is the transition matrix of the random walk, and $\pi$ is its stationary distribution. We show a similar variable-time dependence on the \emph{checking cost}.
We describe our variable-time quantum search and variable-time quantum walk results in \sec{intro-walks}. 

\paragraph{Techniques:} Our results build on multidimensional quantum walks, recently described in~\cite{jeffery2022kDist}. Given a collection of states, each of which can be efficiently prepared, and therefore reflected around, we can define an \emph{overlap graph}, with one node for every state, and a pair of nodes adjacent if and only if the states overlap. If this graph is bipartite, with bipartition $V_{\cal A},V_{\cal B}$, then reflections around the states in $V_{\cal A}$ (resp. $V_{\cal B}$) commute, and so we can also efficiently reflect around the span of all states in $V_{\cal A}$, and around the span of all states in $V_{\cal B}$. Then sufficiently precise phase estimation of the product of these two reflections can be used to reflect around the span of \emph{all} states. When the states are \emph{quantum walk states} for some graph $G$: $\ket{\psi_\star^G(u)}=\sum_{v:\{u,v\}\in E(G)}\sqrt{P_{u,v}}\ket{\{u,v\}}$, then $\ket{\psi_\star^G(u)}$ and $\ket{\psi_\star^G(v)}$ overlap if and only if $\{u,v\}\in G$, so the overlap graph of these states is exactly $G$. Quantum walk search algorithms, like those we describe shortly in \sec{intro-walks}, use analysis of the properties of the graph $G$ to analyze what precision of phase estimation is needed to reflect around the span of all states. 

Ref.~\cite{jeffery2022kDist} extends this idea to a set of \emph{spaces} that can be efficiently reflected around. If $\{\Psi_u\}_{u\in V}$ is a set of subspaces of some inner product space $H$, we can define an overlap graph on $V$, with an edge between $u$ and $v$ whenever $\Psi_u$ and $\Psi_v$ are non-orthogonal. Just as in the case when each space is one-dimensional, if the overlap graph is bipartite, a bipartition gives rise to two efficiently implementable reflections, whose product can be used to reflect around the span of all the spaces using sufficiently precise phase estimation. The properties of the overlap graph can be used to bound the required precision. 

Ref.~\cite{jeffery2022kDist} use this idea in two ways. The first is the technique of \emph{alternative neighbourhoods}, which is orthogonal to the ideas of this paper: we do not deal with this technique here for simplicity, but there is no reason it could not be used simultaneously with our techniques. This work builds on the second technique, called \emph{edge composition}\footnote{which is similar to the concurrent, independent technique of \emph{graph composition} of span programs~\cite{cornelissen2023thesis}.}. Suppose there is a subroutine that implements a step of the quantum walk: for now imagine some reversible version of moving from $u$, to some neighbour $v$, in cost ${\sf T}_{u,v}$. Then implementing the reflection around the states $\{\ket{\psi_\star^G(u)}\}_{u\in V_{\cal A}}$ costs $\max_{u,v}{\sf T}_{u,v}$. The edge composition technique uses the random walk and the subroutine to define, for each $\{u,v\}\in E$, a sequence of ${\sf T}_{u,v}$ spaces whose overlap graph is similar to $G$, but each edge is replaced by a path of length ${\sf T}_{u,v}$ (see \fig{intro}). Our quantum walk composition described in \sec{intro-walks} is similar, except that we allow the subroutine to have a variable, unknown stopping time, and then we describe spaces whose overlap graph looks like $G$, but in place of an edge $\{u,v\}$, there is a ladder-like gadget (see \fig{intro}), where we can imagine getting from $u$ to $v$ by going up one side (representing computation) and that at various ``rungs'' with various probabilities, crossing to the other side of the ladder and uncomputing. By increasing or decreasing the ``weights'' ($\{\alpha_t\}_t$ in \thm{graph-fwk}) as the ladder goes up, we can get various variable-time complexities without knowing the total lengths of the ladders. 

Our result for general algorithmic composition (\sec{intro-alg}) uses the same ideas, but within an algorithm rather than a graph. This can be done by defining a sequence of spaces from the outer algorithm whose overlap graph is a line of length ${\sf L}$, and then plugging in ladder-like gadgets like those in \fig{intro} wherever there is a subroutine query (see \fig{alg-fwk-ortho}).

\begin{figure}
\centering
\begin{tikzpicture}

\node at (0,0) {\begin{tikzpicture}
\draw (0,0)--(2,0)--(4,1); \filldraw (.75,0) circle (.05); \filldraw (1.25,0) circle (.05);
\draw(0,2)--(2,2)--(4,1); \filldraw (.6,2) circle (.05); \filldraw (.8,2) circle (.05); \filldraw (1,2) circle (.05); \filldraw (1.2,2) circle (.05); \filldraw (1.4,2) circle (.05); 
\draw (0,0)--(2,2); \filldraw (.75,.75) circle (.05); \filldraw (1.25,1.25) circle (.05);
\filldraw (2.75,1.625) circle (.05); \filldraw (3.25,1.375) circle (.05);
\filldraw (2.6,.3) circle (.05); \filldraw (3,.5) circle (.05); \filldraw (3.4,.7) circle (.05);

\node[circle,draw, thick, fill=white] at (0,0) {$u_1$};

\node[circle,draw, thick, fill=white] at (0,2) {$u_2$};

\node[circle,draw, thick, fill=white] at (2,0) {$u_3$};

\node[circle,draw, thick, fill=white] at (2,2) {$u_4$};

\node[circle,draw, thick, fill=white] at (4,1) {$u_5$};
\end{tikzpicture}};

\node at (8,0) {\begin{tikzpicture}

\node at (1,3) {\begin{tikzpicture}
\draw (6,-2)--(6.5,-1.75)--(6.5,-.09)--(7,-.09)--(7,-1.75)--(7.5,-2);
	\filldraw (6.5,-.09) circle (.05);	\filldraw (6.75,-.09) circle (.05);				\filldraw (7,-.09) circle (.05);
	\filldraw (6.5,-.42) circle (.05);	\filldraw (6.75,-.42) circle (.05);	\draw (6.5,-.42)--(7,-.42);			\filldraw (7,-.42) circle (.05);
	\filldraw (6.5,-.75) circle (.05);	\filldraw (6.75,-.75) circle (.05);	\draw (6.5,-.75)--(7,-.75);			\filldraw (7,-.75) circle (.05);
	\filldraw (6.5,-1.08) circle (.05);	\filldraw (6.75,-1.08) circle (.05);	\draw (6.5,-1.08)--(7,-1.08);	\filldraw (7,-1.08) circle (.05);
	\filldraw (6.5,-1.417) circle (.05);	\filldraw (6.75,-1.417) circle (.05);	\draw (6.5,-1.417)--(7,-1.417);	\filldraw (7,-1.417) circle (.05);
	\filldraw (6.5,-1.75) circle (.05);
				\filldraw (7,-1.75) circle (.05);
\end{tikzpicture}};

\node[rotate=45] at (.75,1.25) {\begin{tikzpicture}
\draw (5.5,-1.75)--(6.5,-1.75)--(6.5,-1.08)--(7,-1.08)--(7,-1.75)--(8,-1.75);
	\filldraw (6.5,-1.08) circle (.05);	\filldraw (6.75,-1.08) circle (.05);	\draw (6.5,-1.08)--(7,-1.08);	\filldraw (7,-1.08) circle (.05);
	\filldraw (6.5,-1.417) circle (.05);	\filldraw (6.75,-1.417) circle (.05);	\draw (6.5,-1.417)--(7,-1.417);	\filldraw (7,-1.417) circle (.05);
	\filldraw (6.5,-1.75) circle (.05);
				\filldraw (7,-1.75) circle (.05);
\end{tikzpicture}};

\node[rotate=180] at (1,-.6) {\begin{tikzpicture}
\draw (6,-2)--(6.5,-1.75)--(6.5,-1.08)--(7,-1.08)--(7,-1.75)--(7.5,-2);
	\filldraw (6.5,-1.08) circle (.05);	\filldraw (6.75,-1.08) circle (.05);	\draw (6.5,-1.08)--(7,-1.08);	\filldraw (7,-1.08) circle (.05);
	\filldraw (6.5,-1.417) circle (.05);	\filldraw (6.75,-1.417) circle (.05);	\draw (6.5,-1.417)--(7,-1.417);	\filldraw (7,-1.417) circle (.05);
	\filldraw (6.5,-1.75) circle (.05);
				\filldraw (7,-1.75) circle (.05);
\end{tikzpicture}};

\node[rotate=-30] at (3.35,2.15) {\begin{tikzpicture}
\draw (6,-2)--(6.5,-1.75)--(6.5,-1.08)--(7,-1.08)--(7,-1.75)--(7.5,-2);
	\filldraw[white] (6.5,-.75) circle (.05);	\filldraw[white] (6.75,-.75) circle (.05);				\filldraw[white] (7,-.75) circle (.05);
	\filldraw (6.5,-1.08) circle (.05);	\filldraw (6.75,-1.08) circle (.05);	\draw (6.5,-1.08)--(7,-1.08);	\filldraw (7,-1.08) circle (.05);
	\filldraw (6.5,-1.417) circle (.05);	\filldraw (6.75,-1.417) circle (.05);	\draw (6.5,-1.417)--(7,-1.417);	\filldraw (7,-1.417) circle (.05);
	\filldraw (6.5,-1.75) circle (.05);
				\filldraw (7,-1.75) circle (.05);
\end{tikzpicture}};

\node[rotate=210] at (3.35,-.15) {\begin{tikzpicture}
\draw (6,-2)--(6.5,-1.75)--(6.5,-.75)--(7,-.75)--(7,-1.75)--(7.5,-2);
	\filldraw (6.5,-.75) circle (.05);	\filldraw (6.75,-.75) circle (.05);				\filldraw (7,-.75) circle (.05);
	\filldraw (6.5,-1.08) circle (.05);	\filldraw (6.75,-1.08) circle (.05);	\draw (6.5,-1.08)--(7,-1.08);	\filldraw (7,-1.08) circle (.05);
	\filldraw (6.5,-1.417) circle (.05);	\filldraw (6.75,-1.417) circle (.05);	\draw (6.5,-1.417)--(7,-1.417);	\filldraw (7,-1.417) circle (.05);
	\filldraw (6.5,-1.75) circle (.05);
				\filldraw (7,-1.75) circle (.05);
\end{tikzpicture}};

\node[circle,draw, thick, fill=white] at (0,0) {$u_1$};

\node[circle,draw, thick, fill=white] at (0,2) {$u_2$};

\node[circle,draw, thick, fill=white] at (2,0) {$u_3$};

\node[circle,draw, thick, fill=white] at (2,2) {$u_4$};

\node[circle,draw, thick, fill=white] at (4,1) {$u_5$};
\end{tikzpicture}};

\end{tikzpicture}
\caption{Instead of using a $T_{u,v}$ step quantum algorithm for the transition from $u$ to $v$ to build a ``bridge'' from $u$ to $v$ that functions like a path of length $T_{u,v}$ (left), we use a variable-time quantum algorithm to put a ladder-like gadget between $u$ and $v$ (right). The rungs of the ladder correspond to the steps of the quantum algorithm, and intuitively, we should think of the weight of each rung as corresponding to the probability that the algorithm terminates at that step.}\label{fig:intro}
\end{figure}
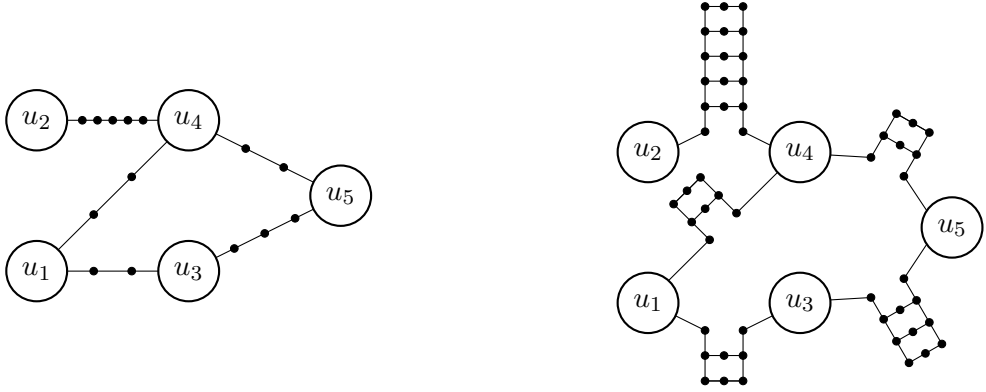

\paragraph{Comparison with Previous and Concurrent Work:} Aside from the special cases of formula evaluation and variable-time search already mentioned, previous work in composing quantum algorithms in non-trivial ways has mainly been restricted to the model of query complexity. \emph{Span programs} can be composed in various nice ways~\cite{reichardt2011spanformulas,reichardt2009unbalanced}, for example, to evaluate formulas in various gate sets with a quantum query complexity that amortizes nicely -- a fact exploited in \cite{childs2022Divide} (concurrent with this work) to analyze the quantum query complexity of divide-and-conquer algorithms. It is also rather simple to get a version of Ambainis' variable-time search result that applies only to query complexity using span program composition.
Ref.~\cite{cornelissen2020SpanProgramTime} showed how to extend span program composition to time complexity, also reproducing Ambainis' variable-time search in a time-efficient way using span program composition. The main idea of~\cite{cornelissen2020SpanProgramTime} allows span programs, also known as \emph{dual adversary solutions}, to maintain time complexity structure where it is given, whereas previously span programs were only used to study query complexity.
The techniques of~\cite{cornelissen2020SpanProgramTime} are the basis for characterizing time complexty in~\cite{jeffery2022kDist}, and now in this work as well. 

In a concurrent, independent work, Belovs and Yolcu defined the notion of quantum Las Vegas query complexity, and in particular, showed a version of our main theorem (\thm{alg-composition-informal}) that applies only to query complexity~\cite{belovs2023LasVegas}. While their results do not apply to time complexity, they work with a significantly more general class of subroutines. While we assume subroutines compute a single bit (which easily extends to any classical function), their results work for subroutines that enact arbitrary state conversion, or even non-unitary maps. Their techniques are quite novel, and distinct from ours. 

\paragraph{Subsequent Work:} In subsequent work, our algorithmic composition results have been extended to quantum subroutines whose behaviour is described by an arbitrary unitary, rather than subroutines that are assumed to compute a single bit~\cite{belovs2024Taming}. This improved composition is obtained by a new tool \emph{transducers}, that generalizes multidimensional quantum walks, using the ideas from this work to extend the constructions in~\cite{belovs2023LasVegas} to capture time complexity. For subroutines that compute deterministic functions with bounded error -- i.e.~the setting where majority voting can be used to amplify the success probability -- \cite{belovs2024Taming} also shows how to get composition results like ours, but without the log factor overhead that would come from such majority voting.

\subsection{Quantum Walk Composition}\label{sec:intro-walks}

Quantum walk search algorithms were first introduced by Szegedy~\cite{szegedy2004QMarkovChainSearch} who showed the following. Let $P$ be the transition matrix of a random walk on a graph $G$ with stationary distribution $\pi$, and $M\subset V(G)$ a \emph{marked set}, both of which may implicitly depend on some input $x$ -- we write, for example $M_x$ when we want to make this dependence explicit. Let ${\cal H}$ be a known bound such that whenever $M\neq\emptyset$, we are promised that the hitting time to $M$ -- the expected number of steps needed by a classical random walk starting in the distribution $\pi$ to reach some $u\in M$ -- is at most ${\cal H}$. Then there is a quantum algorithm that decides, with bounded error, if $M=\emptyset$, which starts by generating an initial state $\sum_u\sqrt{\pi(u)}\ket{u}$, and then makes $\sqrt{\cal H}$ calls to a subroutine whose main components are: (1) for any $u\in V(G)$, generating the quantum walk state $\sum_v\sqrt{P_{u,v}}\ket{u,v}$, which can be seen as the quantum analogue of taking a random walk step (and are essentially the states $\ket{\psi_\star^G(u)}$ mentioned above); (2) for any $u\in V(G)$, checking if $u\in M$. In early work, the most important parameter in the complexity was $\sqrt{\cal H}$ -- a generic quadratic speedup over the analogous classical algorithm -- and it was assumed that other operations, such as generating the quantum walk states had unit cost. In later applications, these operations often had non-trivial costs, impacting the complexity in different ways. Generally ${\sf S}$ is used to denote the cost of generating $\sum_u\sqrt{\pi(u)}\ket{u}$, ${\sf U}$ the cost of generating the quantum walk states, and ${\sf C}$ the cost of checking if $u\in M$. In our work, we will subdivide the task of generating quantum walk states. We suppose (as is common in classical random walk algorithms) that for each $u\in V(G)$, there is a set of labels $L(u)$, for example, $L(u)=[d_u]$, and a mapping $f_u:L(u)\rightarrow V$ such that $f_u(i)$ is the $i$-th neighbour of $u$. Then the following two subroutines can be used in place of generating quantum walk states (this is not entirely obvious): 
\begin{enumerate}
\item Sampling: For each $u\in V$, generate $\sum_{i\in L(u)}\sqrt{P_{u,f_u(i)}}\ket{i}$. We assume this can be done in unit cost.
\item Transitions: For $\{u,v\}\in E(G)$ with $f_u(i)=v$ and $f_v(j)=u$, map $\ket{u,i}\mapsto \ket{v,j}$, in cost $T_{u,v}$, which is a random variable on $[{\sf T}_{\max}]$.
\end{enumerate}
In all applications that we know of, sampling has at most polylogarithmic cost, and it is the transition that may be expensive.\footnote{For example, in algorithms for triangle finding there is often a quantum walk step where the sampling step consists of selecting a new vertex of the input graph (not the graph being walked on) to add to a stored set of vertices, but to complete the transition, it is necessary to find all neighbours of the new vertex in the already stored set.}  We can also suppose the checking cost is variable, with random variable $C_u$ on $[{\sf C}_{\max}]$ representing the cost of checking if $u$ is marked. With this notation in place, the complexity of a Szegedy walk is
\begin{equation}\label{eq:szegedy-fwk}
{\sf S}+\sqrt{\cal H}({\sf T}_{\max}+{\sf C}_{\max}).
\end{equation}

Belovs generalized Szegedy's framework to the \emph{electric network framework}~\cite{belovs2013ElectricWalks} (published in~\cite{belovs2013TimeEfficientQW3Distintness}), in which the walk may start in any distribution $\sigma$. If ${\sf S}_\sigma$ is the cost of generating $\ket{\sigma}=\sum_{u}\sqrt{\sigma(u)}\ket{u}$ (so ${\sf S}={\sf S}_{\pi}$ above), and ${\cal C}_{\sigma,\M}$ is an upper bound on a quantity ${\cal C}_{\sigma,M}(G)$, to be described shortly, then Belovs exhibited a quantum algorithm for deciding if $M=\emptyset$ with complexity 
\begin{equation}\label{eq:intro-elec}
{\sf S}_{\sigma}+\sqrt{{\cal C}_{\sigma,\M}}({\sf T}_{\max}+{\sf C}_{\max}).
\end{equation}
Let us say more about this quantity ${\cal C}_{\sigma,M}(G)$. If ${\cal W}(G)$ is the \emph{total weight} of $G$ (the number of edges when $G$ is unweighted) and ${\cal R}_{\sigma,M}(G)$ is the \emph{effective resistance} between $\sigma$ and $M$ (see \defin{resistance}), then ${\cal C}_{\sigma,M}(G)=2{\cal W}(G){\cal R}_{\sigma,M}(G)$. This is probably not a very helpful definition, so we mention some special cases that give some intuition. 
If $\sigma=\pi$ is the stationary distribution, then ${\cal C}_{\pi,M}(G)$ is just the hitting time from $\pi$ to $M$, and we recover the Szegedy framework. If $\sigma$ is supported on a single vertex $s$, then ${\cal C}_{s,M}(G)$ is the \emph{commute time} or expected number of steps starting from $s$, to get to any vertex in $M$ and then return to $s$~\cite{chandra1996ElectricalResAndCommute} (see~\cite{apers2019UnifiedFrameworkQWSearch} for a proof of the case when $|M|>1$, which was previously folklore). 

Another framework, incomparable to Szegedy's, is the MNRS framework~\cite{magniez2006SearchQuantumWalk}, in which the walk must start in the stationary distribution $\pi$, but the cost is:
\begin{equation}\label{eq:MNRS-fwk}
{\sf S}_{\pi}+\frac{1}{\sqrt{\eps}}\left(\frac{1}{\sqrt{\delta}}{\sf T}_{\max}+{\sf C}_{\max}\right),
\end{equation}
where $\delta$ is the \emph{spectral gap} of $P$, and $\eps$ is a lower bound on $\pi(M)=\sum_{u\in M}\pi(u)$. This framework has different dependence on the cost of transitions and the cost of checking, since the checking subroutine is run less frequently.  Since the hitting time is between $\frac{1}{\eps}$ and $\frac{1}{\eps\delta}$, this may be better or worse than a Szegedy walk, depending on the relative values of ${\sf T}_{\max}$ and ${\sf C}_{\max}$. 
In~\cite{apers2019UnifiedFrameworkQWSearch}, the electric network framework was extended to be able to not only detect if $M\neq \emptyset$, but to \emph{find} an element of $M$ as well, and this extension was shown to also include the MNRS framework as a special case -- simulating an MNRS walk with a walk whose hitting time is $1/\eps$, and update cost is ${\sf T}_{\max}/\sqrt{\delta}$. In this paper, we do not consider finding a marked vertex, but only detecting if there exists a marked vertex.

Finally, a result of Dohotaru and H{\o}yer~\cite{dohotaru2017controlledQAmp} shows how to achieve an optimal number of both checking and update steps in the special case where $|M|\leq 1$, obtaining complexity:
\begin{equation}\label{eq:intro-HD}
{\sf S}_{\pi}+\sqrt{\cal H}{\sf T}_{\max}+\frac{1}{\sqrt{\eps}}{\sf C}_{\max}.
\end{equation}

Previously, quantum walk applications have used various tricks to handle the case when the transitions can have different complexities. The MNRS framework is already one means of handling the case when the checking cost is much larger than the update cost. Nested quantum walks~\cite{jeffery2012NestedQW,childs2013arXivTimeEfficientQW3Distintness}, where the checking or update is implemented by another quantum walk subroutine, can be seen as one large quantum walk with multiple types of update operations with different costs. Variable-time search has also been used \emph{within} the checking subroutine of a quantum walk~\cite{legall2014triangle}. 
The most common method for handling variability in the transition or checking cost has been through the use of tail bounds, but this only works when the cost is highly concentrated around the average. 
In the special case of MNRS quantum walks on \emph{Johnson graphs}, where subroutines are restricted to being in the framework of \emph{extended learning graphs} and the measure of complexity is query complexity rather than time complexity,~\cite{carette2019triangleLG} showed how to replace ${\sf T}_{\max}$ and ${\sf C}_{\max}$ with $\ell_2$-averages, also getting an average in the setup cost ${\sf S}_\pi$, in a way that is specific to Johnson graphs. While Johnson graphs are quite a specific case, almost all known MNRS quantum walk algorithms work on Johnson graphs. 

The recent work of~\cite{jeffery2022kDist}, on whose techniques our results are heavily based, achieves an improvement on the electric network framework (and by extension, its special cases) as follows. Suppose the transition algorithm, rather than having random stopping times, stops after a known time ${\sf T}_{u,v}\leq {\sf T}_{\max}$ on input $(u,i)$, where $f_u(i)=v$. The value ${\sf T}_{u,v}$ must be computable from $(u,i)$, but also from $(v,j)$, where $f_v(j)=u$. Suppose checking can be done in unit cost. Then there is a quantum algorithm that decides if $M=\emptyset$ in complexity 
\begin{equation}\label{eq:intro-JZ}
{\sf S}_\sigma+\sqrt{{\cal C}_{\sigma,\M}^{\sf T}},
\end{equation}
where ${\cal C}_{\sigma,\M}^{\sf T}$ is an upper bound on ${\cal C}_{\sigma,M}(G^{\sf T})$, where $G^{\sf T}$ is the graph $G$ where every edge $\{u,v\}$ has been replaced by a path of length ${\sf T}_{u,v}$. 
This is already a huge improvement over having worst case ${\sf T}_{\max}$ complexity. The results presented in this section are based on a generalization of the results in~\cite{jeffery2022kDist}, and a careful analysis of their implications.

\paragraph{Variable-time Quantum Walks:} We now state our results. For simplicity, we begin by stating our results in the case that checking has unit cost. This is without loss of generality, since we can always add extra edges to the graph from a vertex $u$ to a new vertex that includes a bit indicating if $u\in M$, although this change to the graph changes parameters such as hitting time. In that case, we can simply replace ${\sf T}_{\max}$ with a weighted $\ell_2$ average, weighted by $\pi(u) P_{u,v}$, which is the proportion of time spent traversing edge $(u,v)$ in an infinitely long random walk starting from any distribution, to get complexity (see \cor{variable-time-walk}):
\begin{equation}\label{eq:intro-VT-walk-I}
{\sf S}_\sigma+\sqrt{{\cal C}_{\sigma,\M}}\sqrt{\sum_{(u,v)\in\overrightarrow{E}(G)}\pi(u)P_{u,v}\mathbb{E}[T_{u,v}^2]}\log^{1.5}{\sf T}_{\max}.
\end{equation}
Above, $\overrightarrow{E}(G)$ is the set of all edges, assigned some arbitrary orientation (see \defin{network}). 
If we further assume that the values $\mathbb{E}[T_{u,v}]$ are computable in some strong sense (see \cor{variable-time-walk}), then there is a quantum algorithm that decides if $M=\emptyset$ with bounded error in complexity:
\begin{equation}\label{eq:intro-VT-walk-I-known}
{\sf S}_\sigma+\sqrt{{\cal C}_{\sigma,\M}}\sqrt{\sum_{(u,v)\in\overrightarrow{E}(G)}\pi(u)P_{u,v}\mathbb{E}[T_{u,v}]^2}\log{\sf T}_{\max}.
\end{equation}
Aside from the log factor improvement, \eq{intro-VT-walk-I-known} is better when the variances of $T_{u,v}$ are large, as $\mathrm{Var}(T_{u,v}) = \mathbb{E}[T_{u,v}^2]-\mathbb{E}[T_{u,v}]^2\geq 0$. 
The expressions in \eq{intro-VT-walk-I} and \eq{intro-VT-walk-I-known} should be compared with the electric network framework, \eq{intro-elec} (or \eq{szegedy-fwk} in the special case when $\sigma=\pi$, and so ${\cal C}_{\sigma,\M}={\cal H}$), and the graph composition in \eq{intro-JZ}, although comparison with \eq{intro-JZ} is more difficult. These results follow from a more general statement (\thm{graph-fwk}) of which a version of \eq{intro-JZ} where the transition subroutine may have variable stopping times, and the values $\mathbb{E}[T_{u,v}]$ need not be known in advance, also follows as a special case. As we discuss further below, in some cases, the expressions in \eq{intro-VT-walk-I} and \eq{intro-VT-walk-I-known} are better, whereas in some cases the expression in \eq{intro-JZ} is better, but analysis of \eq{intro-VT-walk-I} and \eq{intro-VT-walk-I-known} may be easier, since we can work with the original graph $G$, and need not analyze a modified graph $G^{\sf T}$.

As stated above, the checking cost can be made part of the update cost by adding a new edge to each vertex $u$, whose other endpoint encodes $u$ and a bit indicating if $u\in M$, and whose transition cost is $C_u$. For comparison with previous work, we work out what impact this might have on the complexity. For comparison with the MNRS framework, we restrict our attention to initial distribution $\sigma=\pi$. Then we can decide if $M=\emptyset$ in complexity (see \cor{variable-time-MNRS}):
\begin{equation}
{\sf S}_\pi+\frac{1}{\sqrt{\eps}}\left(\frac{1}{\sqrt{\delta}}\sqrt{\sum_{(u,v)\in\overrightarrow{E}}\pi(u) P_{u,v}\mathbb{E}[T_{u,v}^2]}+\sqrt{\sum_{u\in V(G)}\pi(u)\mathbb{E}[C_u^2]}\right)\log^{.5}\frac{1}{\pi_{\min}}\log^{1.5}\left({\sf T}_{\max}{\sf C}_{\max}\right),
\end{equation}
where $\pi_{\min}=\min_u\pi(u)$.
Compared to the complexity achieved by the MNRS framework in \eq{MNRS-fwk}, ignoring log factors, we have replaced ${\sf T}_{\max}$ and ${\sf C}_{\max}$ with weighted $\ell_2$-averages. This also generalizes the results of \cite{carette2019triangleLG} from Johnson graphs to general graphs, learning graphs to arbitrary subroutines, and query complexity to time complexity.
As before, if the values $\mathbb{E}[C_u]$ and $\mathbb{E}[T_{u,v}]$ are computable in some strong sense, then we can replace $\mathbb{E}[T_{u,v}^2]$ with $\mathbb{E}[T_{u,v}]^2$ and $\mathbb{E}[C_u^2]$ with $\mathbb{E}[C_u]^2$, and shave off a $\sqrt{\log}$ factor. 
For more general initial distributions, see \cor{variable-time-walk}.

In the simpler case where $|M|\leq 1$, we can once again consider a more general starting distribution $\sigma$. In that case, we get a variable-time analogue of the results of Dohotaru and H{\o}yer in \eq{intro-HD}, but even more general. Specifically, let $\tau$ be \emph{any} distribution on $V(G)$ such that we are promised that if $M=\{m\}\neq \emptyset$, then $\tau(m)\geq \eps$. Then we can decide if $M=\emptyset$ in complexity (see \cor{variable-time-walk}):
$${\sf S}_{\sigma}+\left(\sqrt{{\cal C}_{\sigma,\M}}\sqrt{\sum_{(u,v)\in\overrightarrow{E}(G)}\pi(u)P_{u,v}\mathbb{E}[T_{u,v}^2]}
+\frac{1}{\sqrt{\eps}}\sum_{u\in V(G)}\tau(u)\mathbb{E}[{C}_u^2]\right)\log^{1.5}\left({\sf T}_{\max}{\sf C}_{\max}\right).$$
When $\sigma=\tau=\pi$, so ${\cal C}_{\sigma,\M}={\cal H}$, we get a variable-time version of the Dohotaru-H{\o}yer result in \eq{intro-HD}.  As before, if the values $\mathbb{E}[C_u]$ and $\mathbb{E}[T_{u,v}]$ are computable in some strong sense, then we can replace $\mathbb{E}[T_{u,v}^2]$ with $\mathbb{E}[T_{u,v}]^2$ and $\mathbb{E}[C_u^2]$ with $\mathbb{E}[C_u]^2$, and shave off a $\sqrt{\log}$ factor.

\paragraph{Variable-time Search:} As a special case of our variable-time quantum walk results, we recover something similar to Ambainis variable-time search result, but allowing the stopping time on input $i$ to be a random variable $T_i$, similar to the setting Ambainis considers separately in variable-time amplitude amplification. We restate our result (mentioned already in \eq{intro-VT1}) with slightly more detail. Let $M\subset [n]$ be some marked set, and let $\pi$ be any distribution on $[n]$ such that if $M\neq\emptyset$, $\pi(M)\geq \eps$ for some known bound $\eps$. Leting $T_i$ be the stopping time on input $i\in [n]$ of a subroutine that checks if $i\in M$, we can decide if $M=\emptyset$ with bounded error in the following complexities, assuming the values $\mathbb{E}[T_i]$ are unknown (first expression), or known in a certain strong sense (see \cor{variable-time}):
$$\mbox{(1) unknown:} \sqrt{\frac{1}{\eps}\sum_{i\in [n]}\pi(i)\mathbb{E}[T_i^2]}\log^{1.5}{\sf T}_{\max}
\qquad
\mbox{known:} \sqrt{\frac{1}{\eps}\sum_{i\in [n]}\pi(i)\mathbb{E}[T_i]^2}\log{\sf T}_{\max}.
$$
The case when the values $\mathbb{E}[T_i]$ are known (computable from $i$) follows from Ambainis' variable-time search result, since in that case, on input $i$ we can simply stop after ${\sf T}_i= 10\mathbb{E}[T_i]$ steps (repeating $O(\log n)$ times and taking a majority to reduce the error). In the case when these values are unknown, this would probably also follow from a combination of variable-time search and variable-time amplitude amplification. What is perhaps more interesting is that we obtain alternative complexities for variable-time search, by using different parameters. In the case of unknown expected stopping times, we can also choose either of the following two complexities (neglecting log factors). 
$$\mbox{(2) }\quad \sqrt{\frac{\sum_{i\in [n]}\pi(i)\mathbb{E}[T_i]}{\min_{x:M_x\neq\emptyset}\sum_{i\in M}\frac{\pi(i)}{\mathbb{E}[T_i]}}}
\quad\mbox{ or }\quad\mbox{ (3) }\quad
\frac{1}{\sqrt{\min_{x:M_x\neq\emptyset}\sum_{i\in M}\frac{\pi(i)}{\mathbb{E}[T_i^2]}}}.
$$
Above, we minimize over all allowed (non-empty) marked sets, assuming the marked set depends on some implicit input. \tabl{variable-time-search} shows the three complexities we can obtain, and argues that there are settings where each of the three is the smallest complexity. 
We note that we can achieve analogous complexities to each of (1), (2) and (3) for more general quantum walk algorithms. The generalization of (1) is the results we presented above. The generalization of (2) is essentially what was shown in~\cite{jeffery2022kDist} (\eq{intro-JZ}). We have not explored the generalization of (3). Other settings of the parameters in our main quantum walk theorem, \thm{graph-fwk}, could lead to further alternatives, although it is unclear if these would be useful. 

\subsection{Quantum Algorithm Composition}\label{sec:intro-alg}

We use a similar technique to compose an arbitrary quantum query algorithm with a variable-time subroutine. That is, fix functions $f:\{0,1\}^n\rightarrow\{0,1\}$ and $\{g_i:\{0,1\}^m\rightarrow\{0,1\}\}_{i\in [n]}$, and define $f\circ g:\{0,1\}^m\rightarrow\{0,1\}^n$ by $f\circ g(x)=f(g(x))=f(g_1(x),\dots,g_n(x))$. Suppose we have a quantum algorithm that decides $f$ with bounded error in ${\sf L}$ time steps and ${\sf Q}$ queries to the input to $f$, and a quantum algorithm that decides $g(i,x)=g_i(x)$ in ${\sf T}_{\max}$ steps with some sufficiently small error. Then by composing these algorithms in a naive way, we get a bounded error quantum algorithm for $f\circ g$ with complexity:
$$\widetilde{O}({\sf L}+{\sf Q}\cdot{\sf T}_{\max}).$$
As in the case of quantum walk algorithms, we show how we can do significantly better if the inner algorithm's running time varies significantly in the input $i$, as well as its internal randomness. We show the following (see \thm{alg-composition}):
\begin{theorem}[Informal]\label{thm:alg-composition-informal}
Let $\bar\q_i$ be the average (over all queries made by the outer algorithm) squared amplitude on querying index $i\in [n]$ (see \eq{alg-comp-bar-q_i}). Let $\epsilon_i$ be the error of the inner subroutine on input $i$, and $T_i$ the stopping time of the inner subroutine on input $i$, which is a random variable. 
Let ${\sf T}_{\mathrm{avg}}$ be an upper bound such that:
$$\sum_{i\in [n]}\bar\q_i\mathbb{E}[T_i]\leq {\sf T}_{\mathrm{avg}}$$
and suppose the subroutine's errors satisfy the following condition:
$$\epsilon_{\mathrm{avg}}:=\sum_{i\in [n]}\bar\q_i \epsilon_i \leq \frac{1}{{\sf Q}({\sf L}+{\sf Q}\cdot{\sf T}_{\mathrm{avg}})}.$$
Then there is a quantum algorithm that computes $f\circ g$ with bounded error in complexity
$$\widetilde{O}\left({\sf L}+{\sf Q}\cdot{\sf T}_{\mathrm{avg}}\right).$$
\end{theorem}

We compare this with the naive upper bound of ${\sf L}+{\sf Q}\cdot {\sf T}_{\max}$. Since ${\sf T}_{\max}$ is the \emph{maximum} running time of any subroutine, our new running time can be significantly better when $\mathbb{E}[T_i]$ varies widely over different $i$. If the values $\mathbb{E}[T_i]$ are \emph{known} (i.e. computable from $i$), we can always truncate the inner algorithm, on input $i$, at some time ${\sf T}_i = 10\mathbb{E}[T_i]$, and then by Markov's inequality, the probability that the algorithm has output correctly is not impacted by more than a constant. In that case, \thm{alg-composition-informal} does better than naive composition if and only if the values $\mathbb{E}[T_i]$ vary.   

We remark that our error dependence is not optimal. If $\epsilon_i=\epsilon$ for all $i$, then a subroutine called ${\sf Q}$ times needs $\epsilon\leq 1/{\sf Q}^2$ to achieve overall bounded error.  
In the usual way of calling a subroutine, there is no reason that the error bound on the subroutine queries should scale with anything other than the number of times the subroutine is called. In particular, it need not scale with the number of \emph{other} steps of the algorithm (${\sf L}$) or the number of steps taken by a given call of the subroutine, since $\epsilon$ is already totalled over all these steps. There is also no reason to think that the error should scale with these things when we combine a subroutine and outer algorithm in the way we do in the proof of \thm{alg-composition}, so we suspect there is something here that can be improved. In practice it does not make much difference, since the error of a bounded error subroutine can be made as small as any $\epsilon$ at a multiplication cost of $\log\frac{1}{\epsilon}$, so the overhead is at most logarithmic. 
Still, it would be nice to understand, conceptually, how different errors add up, and at the moment it seems there is something missing from the picture. 
We leave this for future work.

\subsection{Outline}

The rest of this paper is organized as follows. In \sec{prelim}, we give preliminaries on graph theory and random walks (\sec{prelim-graphs}) as well as phase estimation algorithms, on which all of our algorithms are based (\sec{phase-estimation}). In \sec{variable-time}, we introduce our model of variable-time quantum algorithms, which is a slight generalization of the models of \cite{ambainis2010VTSearch} and \cite{ambainis2010VTAA}. We define certain states and subspaces from such algorithms, and prove a number of properties of these that will be used in our composition results. In \sec{graph-composition}, we state and prove our results for composing variable-time subroutines in quantum walks, and in \sec{alg-composition}, we prove \thm{alg-composition-informal}, by showing how to compose a variable-time subroutine into an arbitrary quantum algorithm.

\section{Preliminaries}\label{sec:prelim}

\subsection{Graph Theory}\label{sec:prelim-graphs}

In this section, we define graph theoretic concepts and notation, used in our results on quantum walks. 

\begin{definition}[Network]\label{def:network}
A \emph{network} is a weighted graph $G$ with an (undirected) edge set $E(G)$, vertex set $V(G)$, and some weight function $\w:E(G)\rightarrow\mathbb{R}_{>0}$. Since edges are undirected, we can equivalently describe the edges by some set $\overrightarrow{E}(G)$ such that for all $\{u,v\}\in E(G)$, exactly one of $(u,v)$ or $(v,u)$ is in $\overrightarrow{E}(G)$. The choice of edge directions is arbitrary. Then we can view the weights as a function $\w:\overrightarrow{E}(G)\rightarrow\mathbb{R}_{> 0}$, and for all $(u,v)\in \overrightarrow{E}(G)$, define $\w_{v,u}=\w_{u,v}$. For convenience, we will define $\w_{u,v}=0$ for every pair of vertices such that $\{u,v\}\not\in E(G)$. 
The \emph{total weight} of $G$ is 
$${\cal W}(G):=\sum_{e\in \overrightarrow{E}(G)}\w_e.$$ 
\end{definition}

\noindent For an implicit network $G$, and $u\in V(G)$, we will let $\Gamma(u)$ denote the \emph{neighbourhood} of $u$:
$$\Gamma(u):=\{v\in V(G):\{u,v\}\in E(G)\}.$$
We use the following notation for \emph{the out- and in-neighbourhoods} of $u\in V(G)$:
\begin{equation}
\begin{split}
\Gamma^+(u) &:= \{v\in\Gamma(u):(u,v)\in \overrightarrow{E}(G)\}\\
\Gamma^-(u) &:= \{v\in\Gamma(u):(v,u)\in \overrightarrow{E}(G)\}.
\end{split}\label{eq:neighbourhoods}
\end{equation}

\paragraph{Random Walks:} A \emph{Markov process} on a finite set $V$ is specified by a stochastic \emph{transition matrix} $P\in \mathbb{R}^{V\times V}$, where we interpret $P_{u,v}$ as the probability of moving to the state $v$ when currently in the state $u$. We can define a Markov process $P$ from a weighted graph $G$ by letting:
$$P_{u,v} = \frac{\w_{u,v}}{\w_u},
\mbox{ where }
\w_u:=\sum_{v\in\Gamma(u)}\w_{u,v}.$$
Next, define 
$$\forall u\in V(G),\; \pi(u):=\frac{\w_u}{2{\cal W}(G)}.$$
Then it is easy to check that $\sum_{u\in V(G)}\pi(u)=1$. Furthermore, if $G$ is connected, $\pi$ is the unique left-1-eigenvector of $P$, called the \emph{stationary distribution}. We note that for any $u,v\in V(G)$,
\begin{equation}\label{eq:detailed-balance}
\pi(u) P_{u,v} = \frac{\w_u}{2{\cal W}(G)}\frac{\w_{u,v}}{\w_u} = \frac{\w_{u,v}}{2{\cal W}(G)}=\frac{\w_{v,u}}{2{\cal W}(G)}
= \pi(v) P_{v,u}.
\end{equation}
When $P$ satisfies this condition, called \emph{detailed balance}, we say that $P$ is \emph{reversible}. We have just seen that any random walk on an undirected weighted graph is reversible, and it is also true that any reversible Markov process can be modelled as a random walk on an undirected weighted graph, using edge weights $\w_{u,v} = \pi(u) P_{u,v}$.

Finally, we let $\delta$ denote the \emph{spectral gap} of $P$, which is the smallest non-zero eigenvalue of $I-P$. For intuition, it is useful to know that $\frac{1}{\delta}$ is within a $\log|V(G)|$ factor of the mixing time of $P$ (see, for example,~\cite{levin2017MarkovChainsMixingTimes}).

\paragraph{Accessing $G$:} In computations involving a (classical) random walk on a graph $G$, it is usually assumed that for any $u\in V(G)$, it is possible to sample a neighbour $v\in\Gamma(u)$ according to the distribution given by the $u$-th row of $P$.
It is standard to assume this is broken into two steps: (1)~sampling some $i\in [d_u]$, where $d_u:=|\Gamma(u)|$ is the degree of $u$, and (2)~computing the $i$-th neighbour of $u$. That is, we assume that for each $u\in V(G)$, there is an efficiently computable function $f_u:[d_u]\rightarrow V(G)$ such that $\mathrm{im}(f_u)=\Gamma(u)$, and we call $f_u(i)$ the \emph{$i$-th neighbour of $u$}. In the quantum case (see \defin{QW-access} below), we assume that sampling~(1) can be done coherently, and we use a reversible version of the map $(u,i)\mapsto f_u(i)$. We will also find it convenient to suppose the indices $i$ of the neighbours of $u$ come from some more general set $L(u)$, which may equal $[d_u]$, or some other convenient set, which we call the \emph{edge labels of $u$}. It is possible to have $|L(u)|>|\Gamma(u)|=d_u$, meaning that some elements of $L(u)$ do not label an edge adjacent to $u$ (these labels should be sampled with probability 0). 
We assume we have a partition of $L(u)$ into disjoint $L^+(u)$ and $L^-(u)$ such that:
\begin{equation*}
\begin{split}
L^+(u) &\supseteq \{i\in L(u): (u,f_u(i))\in\overrightarrow{E}(G)\} = \{i\in L(u):f_u(i)\in \Gamma^+(u)\}\\
L^-(u) &\supseteq \{i\in L(u): (f_u(i),u)\in\overrightarrow{E}(G)\} = \{i\in L(u):f_u(i)\in \Gamma^-(u)\}.
\end{split}
\end{equation*}
Note that for any $(u,v)\in\overrightarrow{E}(G)$, with $i=f_u^{-1}(v)$ and $j=f_v^{-1}(u)$, any of $(u,v)$, $(v,u)$, $(u,i)$, or $(v,j)$ fully specify the edge. Thus, it will be convenient to denote the weight of the edge using any of the alternatives:
$$\w_{u,v}=\w_{v,u}=\w_{u,i}=\w_{v,j}.$$
For any $i\in L(u)$, we set $\w_{u,i}=0$ if and only if $\{u,f_u(i)\}\not\in E(G)$.

\begin{definition}[Quantum Walk access to $G$]\label{def:QW-access}
For each $u\in V(G)$, let $L(u)=L^+(u)\cup L^-(u)$ be some finite set of \emph{edge labels}, and $f_u:L(u)\rightarrow V(G)$ a function such that $\Gamma(u)\subseteq \mathrm{im}(f_u)$. 
A quantum algorithm has \emph{quantum walk access} to $G$ if it has access to the following subroutines:
\begin{itemize}
\item A subroutine that generates quantum samples from $L(u)$ by implementing a unitary $U_\star$ in cost ${\sf A}_\star$ that acts as:
\begin{equation*}
U_\star\ket{u,0} \propto \sum_{i\in L^+(u)}\sqrt{\w_{u,i}}\ket{u,i}-\sum_{i\in L^-(u)}\sqrt{\w_{u,i}}\ket{u,i} =: \ket{\psi_\star^G(u)}.
\end{equation*}
\item A subroutine that implements the \emph{transition map}
\begin{equation*}
\ket{u,i}\mapsto \ket{v,j}
\end{equation*}
(possibly with some error) where $i=f^{-1}_u(v)$ and $j=f^{-1}_v(u)$, with costs $\{{\sf T}_{u,i}={\sf T}_{u,v}\}_{(u,v)\in\overrightarrow{E}(G)}$. 

\item Query access to the total vertex weights $\w_u=\sum_{v\in\Gamma(u)}\w_{u,v}$. 
\end{itemize}
We call $\{{\sf T}_e\}_{e\in\overrightarrow{E}(G)}$ the set of \emph{transition costs} and ${\sf A}_\star$ the \emph{cost of generating the star states}. 
\end{definition}

\paragraph{Flows and Resistances:} Electric networks of resistors can be modelled as weighted graphs, where the weights represent \emph{conductances} (so their inverses represent \emph{resistances}). This beautiful connection was expounded by Doyle and Snell~\cite{doyle1984RandomWalksAndElectriNetw}, or see~\cite{levin2017MarkovChainsMixingTimes} for a modern exposition. This connection inspires the following definitions. 

\begin{definition}[Flow, Circulation]\label{def:flow}
A \emph{flow} on a network $G$ is a real-valued function $\theta:\overrightarrow{E}(G)\rightarrow\mathbb{R}$, extended to edges in both directions by $\theta(u,v)=-\theta(v,u)$ for all $(u,v)\in\overrightarrow{E}(G)$. 
For any flow $\theta$ on $G$, and vertex $u\in V(G)$, we define $\theta(u)=\sum_{v\in \Gamma(u)}\theta(u,v)$ as the flow coming out of $u$. If $\theta(u)=0$, we say flow is conserved at $u$. If flow is conserved at every vertex, we call $\theta$ a \emph{circulation}. 
If $\theta(u)>0$, we call $u$ a \emph{source}, and if $\theta(u)<0$ we call $u$ a \emph{sink}. The set of sources and sinks is called the \emph{boundary} of $\theta$.
A flow with unique source $s$ and unique sink $t$ is called an \emph{$st$-flow}. If additionally $\theta(s)=1$, we call $\theta$ a \emph{unit} $st$-flow. 
The \emph{energy} of $\theta$ is 
$${\cal E}(\theta):=\sum_{(u,v)\in\overrightarrow{E}(G)}\frac{\theta(u,v)^2}{\w_{u,v}}.$$
\end{definition}

\begin{definition}[Effective Resistance]\label{def:resistance}
Let $\sigma$ and $\tau$ be distributions on $V(G)$. We define the \emph{effective resistance from $\sigma$ to $\tau$}:
$${\cal R}_{\sigma,\tau}(G):= \min\{{\cal E}(\theta): 
	\forall u \in V(G), \theta(u) = \sigma(u)-\tau(u)\}.$$
When $\sigma$ and $\tau$ have disjoint support, this is the minimum energy of a flow whose sources are exactly $\mathrm{supp}(\sigma)$, with $\theta(u)=\sigma(u)$, and whose sinks are exactly $\mathrm{supp}(\tau)$, with $\theta(u)=-\tau(u)$. In the special case when $\mathrm{supp}(\sigma)=\{s\}$ and $\mathrm{supp}(\tau)=\{t\}$, then $\theta$ ranges over all unit $st$-flows, and we denote ${\cal R}_{\sigma,\tau}(G)$ by ${\cal R}_{s,t}(G)$. Finally, for any $M\subseteq V(G)\setminus\mathrm{supp}(\sigma)$, we define:
$${\cal R}_{\sigma,M}(G):= \min\{{\cal R}_{\sigma,\tau}(G): \mathrm{supp}(\tau)\subseteq M\}.$$
That is, ${\cal R}_{\sigma,M}(G)$ is the minimum energy of a flow from $\sigma$ to $M$ (where any distribution on $M$ is allowed). 
\end{definition}
\noindent There is always a unique $\theta$ from $\sigma$ to $\tau$ (or $\sigma$ to $M$) that achieves energy ${\cal R}_{\sigma,\tau}(G)$ (or ${\cal R}_{\sigma,M}(G)$): If two $\sigma$-$\tau$ flows achieve energy ${\cal E}$, then there is some affine combination of them that achieves energy strictly smaller than ${\cal E}$.

\noindent These concepts are related to random walks as follows.
\begin{theorem}[\cite{chandra1996ElectricalResAndCommute}]
For any $s,t\in V(G)$, let ${\cal C}_{s,t}(G)$ be the \emph{commute time} from $s$ to $t$, which is the expected number of steps taken in a random walk starting at $s$, before the walker reaches $t$, and then returns to $s$. Then ${\cal C}_{s,t}(G) = 2{\cal W}(G){\cal R}_{s,t}(G)$. 
\end{theorem}

\begin{theorem}[Folklore, or see~\cite{apers2019UnifiedFrameworkQWSearch}]
For any set $M\subset V(G)$ and $s\in V(G)$, let ${\cal C}_{s,M}(G)$ be the \emph{commute time} from $s$ to $M$, which is the expected number of steps taken in a random walk starting at $s$, before the walker reaches $M$, and then returns to $s$. Then ${\cal C}_{s,M}(G) = 2{\cal W}(G){\cal R}_{s,M}(G)$. 
\end{theorem}

\begin{theorem}[\cite{belovs2013ElectricWalks}]
For any set $M\subset V(G)$,
let ${\cal H}_{\pi,M}(G)$ be the \emph{hitting time} from the stationary distribution $\pi$ to $M$, which is the expected number of steps taken in a random walk starting in a vertex sampled according to $\pi$, before the walker reaches $M$. Then ${\cal H}_{\pi,M}(G) = 2{\cal W}(G){\cal R}_{\pi,M}(G)$. 
\end{theorem}

\noindent Motivated by the three theorems above, we define 
\begin{equation}\label{eq:commute-time}
{\cal C}_{\sigma,\tau}(G):=2{\cal W}(G){\cal R}_{\sigma,\tau}(G)
\mbox{ and }
{\cal C}_{\sigma,M}(G):=2{\cal W}(G){\cal R}_{\sigma,M}(G).
\end{equation}

\subsection{Phase Estimation Algorithms}\label{sec:phase-estimation}

We will use the precise notion of a \emph{phase estimation algorithm} from \cite{jeffery2022kDist}, which is based on the phase estimation technique of~\cite{kitaev1996PhaseEst}. We refer the reader to \cite[Section~3.1]{jeffery2022kDist} for intuition about why such an algorithm works as stated here.

\begin{definition}[Parameters of a Phase Estimation Algorithm]\label{def:phase-est-alg}
For an implicit input $x\in\{0,1\}^*$, fix a finite-dimensional complex inner product space $H$, a unit vector $\ket{\psi_0}\in H$, and sets of vectors
$\Psi^{\cal A},\Psi^{\cal B}\subset H$. We further assume that $\ket{\psi_0}$ is orthogonal to every vector in $\Psi^{\cal B}$.  Let $\Pi_{\cal A}$ be the orthogonal projector onto ${\cal A}=\mathrm{span}\{\Psi^{\cal A}\}$, and similarly for $\Pi_{\cal B}$.
\end{definition}
Let 
$U_{\cal AB}=(2\Pi_{\cal A}-I)(2\Pi_{\cal B}-I)$.
The algorithm defined by $(H,\ket{\psi_0},\Psi^{\cal A},\Psi^{\cal B})$ performs phase estimation of $U_{\cal AB}$ on initial state $\ket{\psi_0}$, to sufficient precision that by measuring the phase register and checking if the output is 0, we can distinguish between a \emph{negative case} and a \emph{positive case}.

\begin{definition}[Negative Witness]\label{def:neg-witness}
A $\delta$-\emph{negative witness} for $(H,\ket{\psi_0},\Psi^{\cal A},\Psi^{\cal B})$ is a pair of vectors $\ket{w_{\cal A}},\ket{w_{\cal B}}\in H$ such that $\norm{(I-\Pi_{\cal A})\ket{w_{\cal A}}}^2\leq \delta$, $\norm{(I-\Pi_{\cal B})\ket{w_{\cal B}}}^2\leq \delta$, and $\ket{\psi_0}=\ket{w_{\cal A}}+\ket{w_{\cal B}}$. 
\end{definition}

\begin{definition}[Positive Witness]\label{def:pos-witness}
A $\delta$-\emph{positive witness} for $(H,\ket{\psi_0},\Psi^{\cal A},\Psi^{\cal B})$ is a vector $\ket{w}\in H$ such that $\braket{\psi_0}{w}\neq 0$ and $\ket{w}$ is almost orthogonal to all $\ket{\psi}\in \Psi^{\cal A}\cup \Psi^{\cal B}$, in the sense that $\norm{\Pi_{\cal A}\ket{w}}^2\leq \delta\norm{\ket{w}}^2$ and $\norm{\Pi_{\cal B}\ket{w}}^2\leq \delta\norm{\ket{w}}^2$.\footnote{We note that for technical reasons, positive witness error is defined multiplicatively (relative error), whereas negative witness error is defined additively.} 
\end{definition}

\begin{theorem}[\cite{jeffery2022kDist}]\label{thm:phase-est-fwk}
Fix $(H,\ket{\psi_0},\Psi^{\cal A},\Psi^{\cal B})$ as in \defin{phase-est-alg}.
Suppose we can generate the state $\ket{\psi_0}$ in cost ${\sf S}$, and implement $U_{\cal AB}=(2\Pi_{\cal A}-I)(2\Pi_{\cal B}-I)$ in cost ${\sf A}$.

\noindent Let $c_+\in [1,50]$ be some constant, and let ${\cal C}_-\geq 1$ be a positive real number that may scale with $|x|$. Let $\delta$ and $\delta'$ be positive real parameters such that 
\begin{equation*}
\delta\leq \frac{1}{(8c_+)^{3}\pi^8{{\cal C}_-}}\quad\mbox{ and }\quad\delta'\leq \frac{3}{4}\frac{1}{\pi^4c_+}.
\end{equation*}
Suppose we are guaranteed that exactly one of the following holds:
\begin{description}
\item[Positive Condition:] There is a $\delta$-positive witness $\ket{w}$ s.t.~$\frac{|\braket{w}{\psi_0}|^2}{\norm{\ket{w}}^2}\geq \frac{1}{c_+}$. 
\item[Negative Condition:] There is a $\delta'$-negative witness $\ket{w_{\cal A}},\ket{w_{\cal B}}$ s.t.~$\norm{\ket{w_{\cal A}}}^2 \leq {\cal C}_-$.
\end{description}
Then there is a quantum algorithm that distinguishes these two cases with bounded error in cost
$$O\left({\sf S}+\sqrt{{\cal C}_-}{\sf A}\right).$$
\end{theorem}

\section{Variable-Time Subroutines}\label{sec:variable-time}

We formally define a \emph{variable-time subroutine}, which is an algorithm whose runtime is a random variable, which may also depend on some input $i$, 
assumed to be called in the context of some outer algorithm.
Formally, a variable-time subroutine is a sequence of unitaries $U_1,\dots,U_{{\sf T}}$, for some ${\sf T}={\sf T}_{\max}$, acting on the space
$$H_{\cal I}\otimes H_{\cal A}\otimes H_{\cal Z} = \mathrm{span}\{\ket{i}\ket{a}\ket{z}:i\in{\cal I},a\in{\cal A},z\in{\cal Z}\}$$
 for finite sets ${\cal I}$, representing inputs to the subroutine; ${\cal A}$, representing answers the subroutine may output; and ${\cal Z}$, representing states of the algorithm's workspace. Both ${\cal A}$ and ${\cal Z}$ can be assumed to be initialized to $\ket{0}$. 
The unitaries are controlled on $H_{\cal I}$, so we can express $U_t=\sum_{i\in {\cal I}}\ket{i}\bra{i}\otimes U_t^i$ for some unitaries $U_t^i$. In other words, the input register is read only. Following~\cite{ambainis2010VTAA}, we make the following assumptions.
\begin{enumerate}
\item There are subspaces 
$$\{0\}=H_0\subseteq H_1\subseteq\dots \subseteq H_{{\sf T}}= H_{\cal Z}$$
such that $U_t$ leaves $H_{t-1}$ invariant, so any part of the algorithm in $H_{t-1}$ after $U_{t-1}$ has been applied, does not get changed any more. Letting $\Pi_{\leq t-1}$ denote the orthogonal projector onto $H_{t-1}$, we can express this mathematically as:
\begin{equation}\label{eq:Pi-t-U}
U_t\Pi_{\leq t-1}=\Pi_{\leq t-1}.
\end{equation}
For any $t$, we assume we can measure a bit indicating if we are in $H_t$ in unit cost. \footnote{For example, imagine having ${\sf T}$ single-qubit registers, $F_1,\dots,F_{\sf T}$, where we set a 1 in $F_t$ after $U_t$ has been applied to indicate that the algorithm is ready to halt, and then all subsequent unitaries are controlled on no 1s yet being set.}
\item We can implement $\sum_{t=1}^{{\sf T}}\ket{t}\bra{t}\otimes U_t$ in unit cost. 
\end{enumerate}

We say that $t$ is a {potential stopping time} if $\overline{H}_t:=H_t\cap H_{t-1}^\bot\neq\{0\}$, so $H_{\cal Z}:=\overline{H}_1\oplus\dots\oplus \overline{H}_{\sf T}$. 
We can imagine an algorithm that, for every potential stopping time $t$, after applying $U_t$, measures to see if the algorithm is in the space $H_t$, and if so, the algorithm is done and the answer register can also be measured, and if not, the algorithm continues. Since the later unitaries leave $H_t$ invariant, these measurements do not affect the computation. 

Let $\{\ket{z}:z\in{\cal Z}_t\}$ be an orthonormal basis for $\overline{H}_t$, for some disjoint sets ${\cal Z}=\bigcup_t{\cal Z}_t$. Note that ${\cal Z}_t\neq \emptyset$ precisely when $t$ is a potential stopping time. Since we assume we can measure if the algorithm is in each space, we assume we can generate each of these bases. 

For $t\in [{\sf T}]$, we let $\bar{p}_i(t)$ be the probability of stopping at time $t$ (i.e. measuring ``1'' right after $U_t$ is applied) on input $i$. Then 
\begin{equation}
p_i(t)=\bar{p}_i(0)+\dots+\bar{p}_i(t-1)\label{eq:p-i-t}
\end{equation}
is the probability that the algorithm halts some time before $U_t$ is applied. We let $T_i$ be the random variable such that $\Pr[T_i=t]=\bar{p}_i(t)$, so $T_i$ is the stopping time on input $i$. 

Suppose the algorithm is meant to compute some function $g:{\cal I}\rightarrow{\cal A}$. Then if we let $\Pi_t$ be the orthogonal projector onto $\overline{H}_t$, and $\Xi_a$ the orthogonal projector onto $a$ in the answer register, the probability that the algorithm outputs an incorrect answer on input $i$, given that it stops at time $t$, is:
\begin{equation}
\epsilon_i^t:= \frac{1}{\bar{p}_i(t)}\norm{((I-\Xi_{g(i)})\otimes \Pi_t)U_t\dots U_1\ket{i,0,0}}^2.
\end{equation}

\paragraph{Reversible Variable-Time Subroutines:} So far this is just a generalization of Ambainis' notion of a variable stopping time algorithm in~\cite{ambainis2010VTAA} to allowing the output to be more than just a bit. We will extend it to what we call \emph{Reversible Variable-Time Subroutines}. Building on the above notation, suppose what we are \emph{actually} interested in is computing $\ket{i}\mapsto A\ket{i}$, for some isometry $A$. A simple example we will consider (\sec{alg-composition}) is $A\ket{i}=(-1)^{g(i)}\ket{i}$ for some $g:[n]\rightarrow\{0,1\}$. In the other example we will see, in \sec{graph-composition}, $A\ket{u,i}=\ket{v,j}$ for vertices $(u,v)\in \overrightarrow{E}(G)$, with $f_u(i)=v$ and $f_v(j)=u$ (see \defin{QW-access}). 
Whatever the case, we will suppose that computing $A$ can be reduced to computing some auxiliary information $g(i)\in {\cal A}$, for which we have a variable-time subroutine. To this end, we say that a variable-time subroutine \emph{reversibly computes $A$} if it computes $g(i)$, and satisfies the following additional assumptions:
\begin{enumerate}
\item For $t\in\{1,\dots,{\sf T}\}$, letting
$$\tilde{U}_t=\sum_{i\in {\cal I}} A\ket{i}\bra{i}A^\dagger\otimes U_t^i,$$
the unitary $\sum_{t=1}^{{\sf T}}\ket{t}\bra{t}\otimes\tilde{U}_t$ can be implemented in unit cost. 
\item There exists a unitary $A'$ of the following form that can be implemented in unit cost:
$$A'=\sum_{a\in {\cal A}}\ket{a}\bra{a}\otimes A_a,$$
and such that $A_{g(i)}=A$. 
\item For all $i,i'\in{\cal I}$ such that $i\neq i'$, and all $a\in {\cal A}$, $\bra{i'}A^\dagger A_a\ket{i}=0$.
\end{enumerate}
If $A$ were easy to implement, then we would not need a subroutine. Instead, we are assuming (condition 2) that there is an easy to implement $A'$ that computes $A$ given the auxiliary information $g(i)$ computed by the variable-time subroutine. Given a reversible variable-time subroutine for computing $A$, we can implement the map $\ket{i}\mapsto A\ket{i}$ by running the variable-time subroutine $U_1,\dots,U_{{\sf T}}$ until the algorithm halts at some step $T_i$, applying $A'$, and then uncomputing by running $\tilde{U}_{T_i},\dots,\tilde{U}_1$ (condition 1). 
The last condition is not strictly necessary, but the total error of algorithms that use the subroutine (see \sec{graph-composition} and \sec{alg-composition}) is different if we omit it. The condition basically says that while the subroutine might output the wrong answer, it will not output an answer that could interfere with the correct part of another branch of the superposition running the algorithm on a different input. 

For example, suppose $A\ket{i}\ket{0}=\ket{i}\ket{g(i)}$, so $A$ just adds information to the state -- not in an information theoretic sense, but in a computational sense. In that case, it is still possible to control on $i$ to reverse the computation after we have applied $A$ (conditions 1 and 2), and condition 3 is satisfied, since $\braket{i',g(i')}{i,a}=0$ whenever $i\neq i'$, regardless of $a$. Similarly, if $g(i)$ is a single bit, and $A\ket{i}=(-1)^{g(i)}\ket{i}$, then a variable-time subroutine for $g$ is trivially reversible, by setting $\tilde{U}_t = U_t$ for all $t$, and $A_a\ket{i}=(-1)^a\ket{i}$.

\subsection{Transition States and Algorithm States}

We define several states associated with a variable stopping time algorithm, and prove some properties that will be useful later. 
We first define some states called \emph{transition states} from the subroutine, and sort them into two sets $\Psi_0$ and $\Psi_1$, such that each set is pairwise orthogonal. These sets could be used to define a phase estimation algorithm as described in \sec{phase-estimation}. There would be little point in turning an algorithm into a phase estimation algorithm -- if we already have an algorithm, we do not need another one -- but we will later (in \sec{graph-composition} and \sec{alg-composition}) combine these states with some other ones to get a more complicated phase estimation algorithm that uses the variable stopping time subroutine as a building block. 

\begin{definition}\label{def:transition-states}
Fix some set of positive weights $\{\alpha_t\}_{t=1}^{{\sf T}}$ such that $\alpha_0=1$. 
The \emph{transition states} of a reversible variable stopping time subroutine are defined as follows. 
Let ${\cal Z}_{>t} = \bigcup_{t'=1}^{t-1}{\cal Z}_{t'}$.
The \emph{forward transitions states} are defined:
\begin{multline*}
\forall i\in {\cal I}, t\in\{0,\dots,{\sf T}-1\}\\
\Psi_t^{i,\rightarrow}:=\left\{\ket{\psi_{a,z,t}^{i,\rightarrow}} := \ket{\rightarrow}\ket{i}\left( \sqrt{\alpha_t}\ket{a,z}\ket{t} - \sqrt{\alpha_{t+1}}U_{t+1}^i\ket{a,z}\ket{t+1} \right): a\in {\cal A},z\in {\cal Z}_{>t}\right\}.
\end{multline*}
The \emph{backward transition states} are defined:
\begin{multline*}
\forall i\in {\cal I}, t\in\{0,\dots,{\sf T}-1\}\\
\Psi_t^{i,\leftarrow}:=\left\{ \ket{\psi_{a,z,t}^{i,\leftarrow}} := \ket{\leftarrow}A\ket{i}\left( \sqrt{\alpha_t}\ket{a,z}\ket{t} - \sqrt{\alpha_{t+1}}U_{t+1}^i\ket{a,z}\ket{t+1}\right) : a\in {\cal A},z\in {\cal Z}_{>t}  \right\}.
\end{multline*}
The \emph{reversal states} are defined:
\begin{equation*}
\forall i\in {\cal I}, t\in\{1,\dots,{\sf T}\}
\qquad\Psi_t^{i,\leftrightarrow}:=\left\{ \ket{\psi_{a,z,t}^{i,\leftrightarrow}} := \sqrt{\alpha_t}\left(\ket{\rightarrow}\ket{i}-\ket{\leftarrow}A_a\ket{i}\right) \ket{a,z}\ket{t}: a\in{\cal A},z\in {\cal Z}_t\right\}.
\end{equation*}
We finally define $\Psi_0^{i,\leftrightarrow}=\Psi_{\sf T}^{i,\rightarrow}=\Psi_{\sf T}^{i,\leftarrow}=\emptyset$, and:
\begin{align*}
\Psi_0&=\bigcup_{i\in {\cal I}}\bigcup_{t=0: t\;\mathrm{even}}^{{\sf T}}\left(\Psi_t^{i,\rightarrow}\cup\Psi_t^{i,\leftarrow}\cup\Psi_{t}^{i,\leftrightarrow}\right)
\quad\mbox{and}\quad
\Psi_1=\bigcup_{i\in {\cal I}}\bigcup_{t=0: t\;\mathrm{odd}}^{{\sf T}}\left(\Psi_t^{i,\rightarrow}\cup\Psi_t^{i,\leftarrow}\cup\Psi_{t}^{i,\leftrightarrow}\right).
\end{align*}
\end{definition}

\begin{claim}\label{clm:variable-orthog}
For $b\in\{0,1\}$, the states of $\Psi_b$ are pairwise orthogonal. 
\end{claim}
\noindent To convince oneself of \clm{variable-orthog}, one should notice that the states in each set $\Psi_{t}^{i,d}$ are pairwise orthogonal, and the only overlaps between sets are those shown in the graph in \fig{variable-overlap-graph}. Then  since $\Psi_0$ and $\Psi_1$ are formed from a bipartition of this graph, they are each pairwise orthogonal.

\begin{figure}
\centering
\begin{tikzpicture}
\draw (0,3)--(12,3);
\draw(0,0)--(12,0);

\draw (0,3)--(0,0);
\node[rectangle, rounded corners, draw, thick, fill=white] at (0,3) {$\Psi_{0}^{i,\rightarrow}$};
\node[rectangle, rounded corners, draw, thick, fill=white] at (0,1.5) {$\Psi_{1}^{i,\leftrightarrow}$};
\node[rectangle, rounded corners, draw, thick, fill=white] at (0,0) {$\Psi_{0}^{i,\leftarrow}$};

\draw (2,3)--(2,0);
\node[rectangle, rounded corners, draw, thick, fill=white] at (2,3) {$\Psi_{1}^{i,\rightarrow}$};
\node[rectangle, rounded corners, draw, thick, fill=white] at (2,1.5) {$\Psi_{2}^{i,\leftrightarrow}$};
\node[rectangle, rounded corners, draw, thick, fill=white] at (2,0) {$\Psi_{1}^{i,\leftarrow}$};

\draw (4,3)--(4,0);
\node[rectangle, rounded corners, draw, thick, fill=white] at (4,3) {$\Psi_{2}^{i,\rightarrow}$};
\node[rectangle, rounded corners, draw, thick, fill=white] at (4,1.5) {$\Psi_{3}^{i,\leftrightarrow}$};
\node[rectangle, rounded corners, draw, thick, fill=white] at (4,0) {$\Psi_{2}^{i,\leftarrow}$};

\node[fill=white] at (6,3) {$\dots$};
\node[fill=white] at (6,1.5) {$\dots$};
\node[fill=white] at (6,0) {$\dots$};

\draw (8,3)--(8,0);
\node[rectangle, rounded corners, draw, thick, fill=white] at (8,3) {$\Psi_{t}^{i,\rightarrow}$};
\node[rectangle, rounded corners, draw, thick, fill=white] at (8,1.5) {$\Psi_{t+1}^{i,\leftrightarrow}$};
\node[rectangle, rounded corners, draw, thick, fill=white] at (8,0) {$\Psi_{t}^{i,\leftarrow}$};

\node[fill=white] at (10,3) {$\dots$};
\node[fill=white] at (10,1.5) {$\dots$};
\node[fill=white] at (10,0) {$\dots$};

\draw (12,3)--(12,0);
\node[rectangle, rounded corners, draw, thick, fill=white] at (12,3) {$\Psi_{{\sf T}-1}^{i,\rightarrow}$};
\node[rectangle, rounded corners, draw, thick, fill=white] at (12,1.5) {$\Psi_{\sf T}^{i,\leftrightarrow}$};
\node[rectangle, rounded corners, draw, thick, fill=white] at (12,0) {$\Psi_{{\sf T}-1}^{i,\leftarrow}$};

\end{tikzpicture}
\caption{The overlap graph of the sets of states defined in \defin{algorithm-states} for some fixed $i$. Each node represents a set of states that are pairwise orthogonal. Two nodes share an edge if and only if their sets contain overlapping states. For different values of $i$, all states are orthogonal. One can imagine an algorithm starting in the state $\ket{\rightarrow}\ket{i}\ket{0,0}\ket{0}$, which only overlaps $\Psi_{0}^{i,\rightarrow}$. The state of the algorithm moves through the graph until some part of it is of the form $\ket{\leftarrow}A_a\ket{i}\ket{a,z}\ket{t}$, hopefully for $a=g(i)$, then on that part, it uncomputes to move back down the other side of the ladder to a state that only overlaps $\Psi_0^{i,\leftarrow}$. The length of the algorithm's path depends on which ``rung'' $t$ of the ladder it uses to move from the $\rightarrow$ to the $\leftarrow$ part, which will depend on the stopping probabilities at various steps.}\label{fig:variable-overlap-graph}
\end{figure}
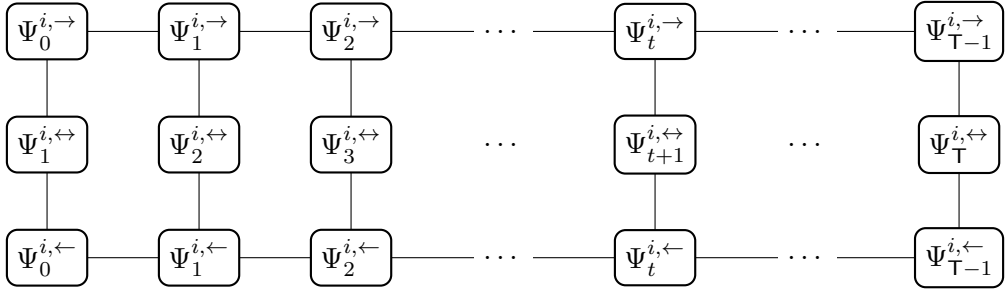

\begin{lemma}\label{lem:variable-time-unitary}
For $b\in \{0,1\}$, let $\Pi_{\Psi_b}$ denote the orthogonal projector onto the span of $\Psi_b$. Then $(2\Pi_{\Psi_b}-I)$ can be implemented in $O(\log {\sf T})$ complexity.
\end{lemma}
\begin{proof}
The basic idea is that for $b\in \{0,1\}$, we can implement a unitary $U_b$ that acts, for all $i\in {\cal I}$, $a\in\{0,1\}$, and $t$ of parity $b$, as:
\begin{enumerate}
\item $\ket{\rightarrow}\ket{i}\ket{a,z,t} \overset{U_b}{\mapsto} \propto \ket{\psi_{a,z,t}^{i,\rightarrow}}$ for all $z\in {\cal Z}_{>t}$
\item $\ket{\leftarrow}A\ket{i}\ket{a,z,t} \overset{U_b}{\mapsto} \propto \ket{\psi_{a,z,t}^{i,\leftarrow}}$ for all $z\in {\cal Z}_{>t}$
\item $\ket{\leftrightarrow}\ket{i}\ket{a,z,t} \overset{U_b}{\mapsto} \propto \ket{\psi_{a,z,t}^{i,\leftrightarrow}}$ for all $z\in {\cal Z}_{t}$
\end{enumerate}
using, respectively:
\begin{enumerate}
\item The ability to implement $\sum_{t=1}^{\sf T}\ket{t}\bra{t}\otimes U_t$ in unit cost.
\item The ability to implement $\sum_{t=1}^{\sf T}\ket{t}\bra{t}\otimes\tilde{U}_t$ in unit cost.
\item The ability to implement $\sum_{a\in{\cal A}}\ket{a}\bra{a}\otimes A_a$ in unit cost. 
\end{enumerate}
We implement $U_b$ in two parts: the first is conditioned on $z\in {\cal Z}_{>t}$, which will ensure (1) and (2); the second is conditioned on $z\in {\cal Z}_t$, which will ensure condition (3). First, conditioned on $z\in {\cal Z}_{>t}$, do the following. Append an auxiliary qubit initialized to $\ket{0}$. Controlled on $t$, rotate this qubit to $\propto\sqrt{\alpha_t}\ket{0}-\sqrt{\alpha_{t+1}}\ket{1}$. When $z\in {\cal Z}_{>t}$, we only care about the cases where the first qubit is $\ket{\rightarrow}$ or $\ket{\leftarrow}$, and in these cases we now have:
\begin{align*}
(\rightarrow)\mbox{-case: }&\sqrt{\alpha_t}\ket{0}\ket{\rightarrow}\ket{i}\ket{a,z,t} - \sqrt{\alpha_{t+1}}\ket{1}\ket{\rightarrow}\ket{i}\ket{a,z,t},\\
(\leftarrow)\mbox{-case: }&\sqrt{\alpha_t}\ket{0}\ket{\leftarrow}A\ket{i}\ket{a,z,t} - \sqrt{\alpha_{t+1}}\ket{1}\ket{\leftarrow}A\ket{i}\ket{a,z,t}.
\end{align*}
Controlled on the auxiliary qubit being $\ket{1}$, increment the last register, and then (still controlled on $\ket{1}$) controlled on $\ket{\rightarrow}$, apply $\sum_{t=1}^{\sf T}\ket{t}\bra{t}\otimes U_t$, and controlled on $\ket{\leftarrow}$, apply $\sum_{t=1}^{\sf T}\ket{t}\bra{t}\otimes \tilde{U}_t$ to get:
\begin{align*}
(\rightarrow)\mbox{-case: }&\sqrt{\alpha_t}\ket{0}\ket{\rightarrow}\ket{i}\ket{a,z,t} - \sqrt{\alpha_{t+1}}\ket{1}\ket{\rightarrow}\ket{i}U_{t+1}\ket{a,z}\ket{t+1},\\
(\leftarrow)\mbox{-case: }&\sqrt{\alpha_t}\ket{0}\ket{\leftarrow}A\ket{i}\ket{a,z,t} - \sqrt{\alpha_{t+1}}\ket{1}\ket{\leftarrow}A\ket{i}U_{t+1}\ket{a,z}\ket{t+1}.
\end{align*}
Since we only care about the action of $U_b$ when $t$ has parity $b$, we can uncompute the auxiliary using the parity of the last register. 

Next, conditioned on $z\in {\cal Z}_t$, do the following. Map $\ket{\leftrightarrow}$ to 
$(\ket{\rightarrow}-\ket{\leftarrow})/\sqrt{2}$, and then conditioned on $\ket{\leftarrow}$, apply $\sum_a\ket{a}\bra{a}\otimes A_a$ to get:
\begin{equation*}
\frac{1}{\sqrt{2}}(\ket{\rightarrow}\ket{i}\ket{a,z,t} - \ket{\leftarrow}A_a\ket{i}\ket{a,z,t}) \propto \ket{\psi_{a,z,t}^{i,\leftrightarrow}}.
\end{equation*}

The map $U_b$ can be combined with a reflection $R_b$ that checks, for any state $\ket{d}\ket{i}\ket{a,z}\ket{t}$: $t$ has parity $b$; if $d\in\{\leftarrow,\rightarrow\}$, then $z\in {\cal Z}_{>t}$ and $t<{\sf T}$; and if $d=\leftrightarrow$, then $z\in {\cal Z}_t$ and $t>0$. Then $(2\Pi_{\Psi_b}-I) = U_bR_bU_b^\dagger$. 
\end{proof}

We now define some states which are essentially positive and negative witnesses for a phase estimation algorithm defined by $\Psi_0$ and $\Psi_1$. When we combine $\Psi_0$ and $\Psi_1$ with other states to get more complex phase estimation algorithms in \sec{graph-composition} and \sec{alg-composition}, these will be building blocks for the positive and negative witnesses for those phase estimation algorithms, so we prove several claims analyzing these states here.

\begin{definition}[Algorithm States]\label{def:algorithm-states}
Fix a reversible variable-time subroutine $U_1,\dots,U_{\sf T},A$ and 
some set of positive weights $\{\alpha_t\}_t$.
Let $\Pi_{\geq t}:=\Pi_{t}+\dots+\Pi_{{\sf T}}$ be the orthogonal projector onto $H_{t-1}^\bot = \overline{H}_{t}\oplus \dots\oplus \overline{H}_{{\sf T}}$. Define the following \emph{algorithm states} in $H_{\cal A}\otimes H_{\cal Z}$, for all $i\in {\cal I}$: 
\begin{align*}
\ket{w^{0}(i)} &= \ket{0,0}\\
\forall t\in [{\sf T}],\; \ket{w^t(i)} &= U_t^i\Pi_{\geq t}\ket{w^{t-1}(i)}.
\end{align*}
The \emph{positive history state} of the algorithm on input $i\in {\cal I}$ is defined:
\begin{equation*}
\ket{w_+(i)} = (\ket{\rightarrow}\ket{i}+\ket{\leftarrow}A\ket{i})\sum_{t=0}^{{\sf T}}\frac{1}{\sqrt{\alpha_t}}\ket{w^t(i)}\ket{t}.
\end{equation*}
The \emph{negative history state} of the algorithm on input $i\in {\cal I}$ is defined:
\begin{equation*}
\ket{w_-(i)} = (\ket{\rightarrow}\ket{i}-\ket{\leftarrow}A\ket{i})\sum_{t=0}^{{\sf T}}{\sqrt{\alpha_t}}(-1)^t\ket{w^t(i)}\ket{t}.
\end{equation*}
\end{definition}

We have defined the algorithm states in this way because $\Pi_{\geq t}\ket{w^{t-1}(i)}$ is the part of the algorithm that has not yet output before we apply $U_t^i$. The following claim makes clear that $\ket{w^t(i)}$ consists of two orthogonal parts: the part in $\overline{H}_t$, which is output at this step, and the part in $H_t^\bot = \overline{H}_{t+1}\oplus\dots\oplus \overline{H}_{{\sf T}}$, which has yet to halt.

\begin{claim}\label{clm:alg-state-forms}
For all $i\in {\cal I}$ and $t\in \{0,\dots,{\sf T}\}$, 
$$\ket{w^t(i)}=\Pi_{\geq t}U_t^i\Pi_{\geq t-1}U_{t-1}^i\dots \Pi_{\geq 1}U_1^i\ket{0,0}
=\Pi_{\geq t}U_t^i\dots U_1^i\ket{0,0}.$$
\end{claim}
\begin{proof}
First, note that since $U_t^i\Pi_{\leq t-1} = \Pi_{\leq t-1}$ (by \eq{Pi-t-U}), we have:
\begin{equation}\label{eq:U-t-Pi-commute}
U_t^i\Pi_{\geq t} = U_t^i(I-\Pi_{\leq t-1}) = U_t^i - \Pi_{\leq t-1} = (I-\Pi_{\leq t-1})U_t^i = \Pi_{\geq t}U_t^i,
\end{equation}
so the first equality follows easily from the definition of $\ket{w^t(i)}$. 

We prove the second equality by induction. For the base step, we have:
$\ket{w^0(i)}=\ket{0,0}=\Pi_{\geq 0}\ket{0,0}.$
For the induction step, assuming  $\ket{w^{t-1}(i)}=\Pi_{\geq t-1}U_{t-1}^i\dots U_1^i\ket{0,0}$, we have:
\begin{align*}
\ket{w^t(i)}=U_t\Pi_{\geq t}\ket{w^{t-1}(i)}
&=  U_t\Pi_{\geq t} \Pi_{\geq t-1}U_{t-1}^i\dots U_1^i\ket{0,0}
=U_t\Pi_{\geq t}U_{t-1}^i\dots U_1^i\ket{0,0}\\
&=\Pi_{\geq t}U_tU_{t-1}^i\dots U_1^i\ket{0,0}.\qedhere
\end{align*}
\end{proof}

\begin{corollary}
Let $i\in {\cal I}$. Recall the definitions of $\bar{p}_i(t)$ and $p_i(t)$ (\eq{p-i-t}). Then we have:
\begin{equation*}
\norm{\Pi_t\ket{w^t(i)}}^2 = \bar{p}_i(t)
\mbox{ and }
\norm{\ket{w^t(i)}}^2 = 1-p_i(t).
\end{equation*}
\end{corollary}
\begin{proof}
By \clm{alg-state-forms}, we have
$\Pi_t\ket{w^t(i)} = \Pi_tU_t^i\dots U_1^i\ket{0,0}$,
which is the part of the state after applying $U_t$ that did not already output, but does output in the measurement after $U_t$, so 
$\bar{p}_i(t)=\norm{\Pi_t\ket{w^t(i)}}^2.$
For the second part, we have 
$$\norm{\ket{w^t(i)}}^2 = \norm{U_t^i\dots U_1^i\ket{0,0}}^2-\norm{\Pi_{<t}U_t^i\dots U_1^i\ket{0,0}}^2=1-p_i(t),$$
since $\Pi_{<t}U_t^i\dots U_1^i\ket{0,0}$ is the part of the state that has already halted before $U_t$ was applied.
\end{proof}

\begin{corollary}\label{cor:pos-witness-complexity}
For all $i\in {\cal I}$, 
\begin{align*}
\norm{\ket{w_-(i)}}^2 &= 2\sum_{t=0}^{{\sf T}}\alpha_t(1-p_i(t)) = 2\mathbb{E}\left[\sum_{t=0}^{T_i}\alpha_t\right]\\
\mathrm{ and }\;
\norm{\ket{w_+(i)}}^2 &= 2\sum_{t=0}^{{\sf T}}\frac{1}{\alpha_t}(1-p_i(t)) =  2\mathbb{E}\left[\sum_{t=0}^{T_i}\frac{1}{\alpha_t}\right].
\end{align*}
\end{corollary}
\begin{proof}
By \eq{p-i-t} and the fact that $\sum_{t=1}^{\sf T}\bar{p}_i(t)=1$, and letting $\bar{p}_i(0)=0$, we have:
\begin{align*}
\norm{\ket{w_-(i)}}^2 = 
\sum_{t=0}^{\sf T}\alpha_t(1-p_i(t)) &= \sum_{t=0}^{\sf T}\alpha_t\sum_{t'=t}^{\sf T}\bar{p}_i(t')
=\sum_{t'=0}^{{\sf T}}\bar{p}_i(t')\sum_{t=0}^{t'}\alpha_t=\mathbb{E}\left[\sum_{t=0}^{T_i}\alpha_t\right].
\end{align*}
A similar computation applies in the positive case.
\end{proof}

\begin{claim}[Positive Witness]\label{clm:ortho-alg-trans}
Assume ${\sf T}$ is odd.
For all $i$, $\ket{w_+(i)}$ is orthogonal to all forward and backward transition states, and all reversal states $\ket{\psi_{a,z,t}^{i',\leftrightarrow}}$ such that $i'\neq i$ or $a'\neq g(i)$ (see \defin{transition-states}). Moreover, if $\Pi_{\Psi_b}$ is the orthogonal projector onto $\mathrm{span}\{\Psi_b\}$, then 
\begin{align*}
\norm{\Pi_{\Psi_0}\ket{w_+(i)}}^2 &\leq 2\sum_{\substack{t=0\\\mathrm{even}}}^{{\sf T}-1}\frac{1}{\alpha_t}\bar{p}_i(t)\epsilon_i^t\leq 2\mathbb{E}\left[\frac{\epsilon_i^{T_i}}{\alpha_{T_i}}\right]
\text{and}\; \norm{\Pi_{\Psi_1}\ket{w_+(i)}}^2  \leq 2\sum_{\substack{t=1\\\mathrm{odd}}}^{{\sf T}}\frac{1}{\alpha_t}\bar{p}_i(t)\epsilon_i^t\leq 2\mathbb{E}\left[\frac{\epsilon_i^{T_i}}{\alpha_{T_i}}\right].
\end{align*}
\end{claim}
\begin{proof}
We start with the \textbf{forward transition states}. We always have $\braket{w_+(i)}{\psi_{a,z,t}^{i',\rightarrow}}=0$ when $i\neq i'$. For all $z\in {\cal Z}_{>t}$, we have:
\begin{align*}
\braket{\psi_{a,z,t}^{i,\rightarrow}}{w_+(i)} = \braket{a,z}{w^t(i)}-\bra{a,z}(U_{t+1}^i)^\dagger\ket{w^{t+1}(i)}
= \braket{a,z}{w^t(i)}-\bra{a,z}(U_{t+1}^i)^\dagger U_{t+1}^i\Pi_{\geq t+1}\ket{w^{t}(i)} = 0,
\end{align*}
since $\bra{a,z}\Pi_{\geq t}=\bra{a,z}$ for all $z\in {\cal Z}_{>t}$. 
A nearly identical argument holds for the \textbf{backward transition states}. 

Next consider the \textbf{reversal states}. We have
\begin{align}
\braket{\psi_{a,z,t}^{i,\leftrightarrow}}{w_+(i)} = \left(1-\bra{i}A_a^\dagger A\ket{i} \right)\braket{a,z}{w^t(i)}.\label{eq:ortho-alg-trans}
\end{align}
When $a=g(i)$, so $A_a=A_{g(i)}=A$, this is 0, and otherwise we have $|\braket{\psi_{a,z,t}^{i,\leftrightarrow}}{w_+(i)}|^2\leq 4|\braket{a,z}{w^t(i)}|^2$. Thus:
\begin{align*}
\Pi_{\Psi_0}\ket{w_+(i)} &= \sum_{\substack{t=0\\ \mathrm{even}}}^{{\sf T}-1}\sum_{\substack{a\in{\cal A}\setminus\{g(i)\}\\ z\in{\cal Z}_t}}\frac{\ket{\psi_{a,z,t}^{i,\leftrightarrow}}\bra{\psi_{a,z,t}^{i,\leftrightarrow}} }{\norm{\ket{\psi_{a,z,t}^{i,\leftrightarrow}}}^2}\ket{w_+(i)}\\
\norm{\Pi_{\Psi_0}\ket{w_+(i)}}^2 &= \sum_{\substack{t=0\\ \mathrm{even}}}^{{\sf T}-1}\sum_{\substack{a\in{\cal A}\setminus\{g(i)\}\\ z\in{\cal Z}_t}}\frac{|\braket{\psi_{a,z,t}^{i,\leftrightarrow}}{w_+(i)}|^2 }{\norm{\ket{\psi_{a,z,t}^{i,\leftrightarrow}}}^2}
\leq \sum_{\substack{t=0\\ \mathrm{even}}}^{{\sf T}-1}\sum_{\substack{a\in{\cal A}\setminus\{g(i)\}\\ z\in{\cal Z}_t}}\frac{4|\braket{a,z}{w^t(i)}|^2}{2\alpha_t}.
\end{align*}
We have, for any $t$,
$$\sum_{\substack{a\in{\cal A}\setminus\{g(i)\}\\ z\in{\cal Z}_t}}|\braket{a,z}{w^t(i)}|^2 = \bar{p}_i(t)\epsilon_i^t$$
is the probability that the algorithm outputs at time $t$, and the answer is incorrect (not equal to $g(i)$). Thus:
\begin{align*}
\norm{\Pi_{\Psi_0}\ket{w_+(i)}}^2 &\leq 2\sum_{\substack{t=0\\ \mathrm{even}}}^{{\sf T}-1}\bar{p}_i(t)\frac{1}{\alpha_t}\epsilon_i^t
\leq 2\sum_{t=0}^{{\sf T}}\bar{p}_i(t)\frac{1}{\alpha_t}\epsilon_i^t
=2\mathbb{E}\left[\frac{\epsilon_i^{T_i}}{\alpha_{T_i}}\right],
\end{align*}
and similarly for $\Psi_1$. 
\end{proof}

\begin{claim}[Negative witness]\label{clm:neg-witness-error}
Assume ${\sf T}$ is odd.
For $i\in {\cal I}$, let $\ket{{w}_-(i)}$ be the negative history state defined in \defin{algorithm-states}. Then, letting $\Pi_{\Psi_b}$ be the orthogonal projector onto $\mathrm{span}\{\Psi_b\}$, we have:
\begin{align*}
\norm{(I-\Pi_{\Psi_0})\ket{{w}_-(i)}}^2 &\leq 2\sum_{\substack{t=0\\ \mathrm{even}}}^{{\sf T}-1}{\alpha_t}\bar{p}_i(t)\epsilon^t_i\leq 2\mathbb{E}\left[\alpha_{T_i}\epsilon_i^{T_i}\right] \\
\text{and}\; \norm{(I-\Pi_{\Psi_1})\left(\ket{{w}_-(i)} - (\ket{\rightarrow}\ket{i}-\ket{\leftarrow}A\ket{i})\ket{0,0}\ket{0}\right)}^2 &\leq 2\sum_{\substack{t=1\\ \mathrm{odd}}}^{{\sf T}}{\alpha_t}\bar{p}_i(t)\epsilon^t_i\leq 2\mathbb{E}\left[\alpha_{T_i}\epsilon_i^{T_i}\right].
\end{align*}
\end{claim}
\begin{proof}
Recall from \defin{algorithm-states} that 
\begin{align*}
\ket{w_-(i)} &= \left(\ket{\rightarrow}\ket{i} - \ket{\leftarrow}A\ket{i}\right)\sum_{t=0}^{{\sf T}}\sqrt{\alpha_t}(-1)^t\ket{w^t(i)}\ket{t},
\end{align*}
where $\ket{w^t(i)}=\Pi_{\geq t}\ket{w^t(i)}$ by \clm{alg-state-forms}. 
We have:
\begin{align*}
&\sum_{t=0}^{{\sf T}}\sqrt{\alpha_t}(-1)^t\ket{w^t(i)}\ket{t}
={} \sum_{\substack{t=0\\ \mathrm{even}}}^{{\sf T}-1}\sqrt{\alpha_t}\ket{w^t(i)}\ket{t}
- \sum_{\substack{t=1\\ \mathrm{odd}}}^{{\sf T}}\sqrt{\alpha_t}\ket{w^t(i)}\ket{t}\\
={}& \sum_{\substack{t=0\\ \mathrm{even}}}^{{\sf T}-1}\sqrt{\alpha_t}\Pi_t\ket{w^t(i)}\ket{t}
+\sum_{\substack{t=0\\ \mathrm{even}}}^{{\sf T}-1}\sqrt{\alpha_t}\Pi_{\geq t+1}\ket{w^t(i)}\ket{t}
- \sum_{\substack{t=1\\ \mathrm{odd}}}^{{\sf T}}\sqrt{\alpha_{t}}\Pi_{\geq t}U_{t}^i\ket{w^{t-1}(i)}\ket{t}\\
={}& \sum_{\substack{t=0\\ \mathrm{even}}}^{{\sf T}-1}\sqrt{\alpha_t}\Pi_{t}\ket{w^t(i)}\ket{t}
+\sum_{\substack{t=0\\ \mathrm{even}}}^{{\sf T}-1}\left(\sqrt{\alpha_t}\Pi_{\geq t+1}\ket{w^t(i)}\ket{t}
- \sqrt{\alpha_{t+1}}U_{t+1}^i\Pi_{\geq t+1}\ket{w^{t}(i)}\ket{t+1}\right) & \mbox{by \eq{U-t-Pi-commute}}.
\end{align*}
Then since $\Pi_{\geq t+1}\ket{w^t(i)}$ is supported on states $\ket{a,z}$ such that $z\in {\cal Z}_{>t}$, we have:  
 \begin{equation*}
\left(\ket{\rightarrow}\ket{i} - \ket{\leftarrow}A\ket{i}\right)\sum_{\substack{t=0\\ \mathrm{even}}}^{{\sf T}-1}\left(\sqrt{\alpha_t}\Pi_{\geq t+1}\ket{w^t(i)}\ket{t}
- \sqrt{\alpha_{t+1}}U_{t+1}^i\Pi_{\geq t+1}\ket{w^{t}(i)}\ket{t+1}\right)
\in\mathrm{span}\{\Psi_0\}.
\end{equation*}
Let $\Xi_a$ be the orthogonal projector onto states with $a$ in the answer register. We can thus see that:
\begin{equation}
\begin{split}
(I-\Pi_{\Psi_0})\ket{w_-(i)} 
&=(I-\Pi_{\Psi_0})\left(\ket{\rightarrow}\ket{i} - \ket{\leftarrow}A\ket{i}\right)\sum_{\substack{t=0\\ \mathrm{even}}}^{{\sf T}-1}\sqrt{\alpha_t}\Pi_{t}\ket{w^t(i)}\ket{t}\\
&\qquad -(I-\Pi_{\Psi_0})\underbrace{\left(\ket{\rightarrow}\ket{i} - \ket{\leftarrow}A_{g(i)}\ket{i}\right)\sum_{\substack{t=0\\ \mathrm{even}}}^{{\sf T}-1}\sqrt{\alpha_t}\Xi_{g(i)}\Pi_{t}\ket{w^t(i)}\ket{t}}_{\in \Psi_0}\\
&= (I-\Pi_{\Psi_0})\left(\ket{\rightarrow}\ket{i} - \ket{\leftarrow}A_{g(i)}\ket{i}\right)\sum_{\substack{t=0\\ \mathrm{even}}}^{{\sf T}-1}\sqrt{\alpha_t}\left(I-\Xi_{g(i)} \right)\Pi_t\ket{w^t(i)}\ket{t},
\end{split}\label{eq:variable-neg-error}
\end{equation}
since $A\ket{i}=A_{g(i)}\ket{i}$.  Thus:
\begin{align*}
\norm{(I-\Pi_{\Psi_0})\ket{w_-(i)}}^2 &\leq \norm{\left(\ket{\rightarrow}\ket{i} - \ket{\leftarrow}A_{g(i)}\ket{i}\right)\sum_{\substack{t=0\\ \mathrm{even}}}^{{\sf T}-1}\sqrt{\alpha_t}\left(I-\Xi_{g(i)} \right)\Pi_t\ket{w^t(i)}\ket{t}}^2\\
&= 2\sum_{\substack{t=0\\ \mathrm{even}}}^{{\sf T}-1}{\alpha_t}\norm{\left(I-\Xi_{g(i)} \right)\Pi_t\ket{w^t(i)}}^2
= 2\sum_{\substack{t=0\\ \mathrm{even}}}^{{\sf T}-1}\bar{p}_i(t){\alpha_t}\epsilon^t_i
\leq 2\mathbb{E}\left[\alpha_{T_i}\epsilon_i^{T_i}\right].
\end{align*}

For the second part of the claim, we note that:
\begin{align*}
&\sum_{t=0}^{{\sf T}}\sqrt{\alpha_t}(-1)^t\ket{w^t(i)}\ket{t}
\\
={}&\sqrt{\alpha_0}\ket{w^0(i)}\ket{0} -\sqrt{\alpha_{{\sf T}}}\ket{w^{{\sf T}}(i)}\ket{{\sf T}}\\
&\quad -\sum_{\substack{t=1\\ \mathrm{odd}}}^{{\sf T}-2}\sqrt{\alpha_t}\Pi_{t}\ket{w^t(i)}\ket{t}
-\sum_{\substack{t=1\\ \mathrm{odd}}}^{{\sf T}-2}\left(\sqrt{\alpha_t}\Pi_{\geq t+1}\ket{w^t(i)}\ket{t}
- \sqrt{\alpha_{t+1}}U_{t+1}^i\Pi_{\geq t+1}\ket{w^{t}(i)}\ket{t+1}\right)\\
={}&\ket{0,0}\ket{0} 
 -\sum_{\substack{t=1\\ \mathrm{odd}}}^{{\sf T}}\sqrt{\alpha_t}\Pi_{t}\ket{w^t(i)}\ket{t}
-\sum_{\substack{t=1\\ \mathrm{odd}}}^{{\sf T}-2}\left(\sqrt{\alpha_t}\Pi_{\geq t+1}\ket{w^t(i)}\ket{t}
- \sqrt{\alpha_{t+1}}U_{t+1}^i\Pi_{\geq t+1}\ket{w^{t}(i)}\ket{t+1}\right),
\end{align*}
since $\alpha_0=1$, $\ket{w^0(i)}=\ket{0,0}$, and $\ket{w^{{\sf T}}(i)}$ is in the support of $\Pi_{\geq {\sf T}}=\Pi_{{\sf T}}$.
Similar to above, we have
\begin{align*}
& (\ket{\rightarrow}\ket{i}-\ket{\leftarrow}A\ket{i})\sum_{\substack{t=1\\ \mathrm{odd}}}^{{\sf T}-2}\left(\sqrt{\alpha_t}\Pi_{\geq t+1}\ket{w^t(i)}\ket{t}
- \sqrt{\alpha_{t+1}}U_{t+1}^i\Pi_{\geq t+1}\ket{w^{t}(i)}\ket{t+1}\right)
\in \Psi_1,
\end{align*}
and just as in \eq{variable-neg-error}, 
\begin{align*}
&(I-\Pi_{\Psi_1})(\ket{\rightarrow}\ket{i}-\ket{\leftarrow}A\ket{i})\sum_{\substack{t=1\\ \mathrm{odd}}}^{{\sf T}}\sqrt{\alpha_t}\Pi_{t}\ket{w^t(i)}\ket{t}\\
={}& (I-\Pi_{\Psi_1})(\ket{\rightarrow}\ket{i}-\ket{\leftarrow}A\ket{i})\sum_{\substack{t=1\\ \mathrm{odd}}}^{{\sf T}}\sqrt{\alpha_t}\Pi_{t}(I-\Xi_{g(i)})\ket{w^t(i)}\ket{t}.
\end{align*}
Thus, we can see that:
\begin{multline*}
(I-\Pi_{\Psi_1})\left(\ket{w_-(i)}-(\ket{\rightarrow}\ket{i}-\ket{\leftarrow}A\ket{i})\ket{0,0}\ket{0}\right)\\
=(I-\Pi_{\Psi_1})(\ket{\rightarrow}\ket{i}-\ket{\leftarrow}A\ket{i})\left(\ket{0,0}\ket{0}  - \sum_{\substack{t=1\\ \mathrm{odd}}}^{{\sf T}}\sqrt{\alpha_t}\Pi_t(I-\Xi_{g(i)})\ket{w^t(i)}\ket{t}
-\ket{0,0}\ket{0}\right)
\end{multline*}
\begin{multline*}
\mbox{so }
\norm{(I-\Pi_{\Psi_1})\left(\ket{w_-(i)}-(\ket{\rightarrow}\ket{i}-\ket{\leftarrow}A\ket{i})\ket{0,0}\ket{0}\right)}^2\\
\leq 2\sum_{\substack{t=1\\ \mathrm{odd}}}^{{\sf T}}{\alpha_t}\norm{\Pi_t(I-\Xi_{g(i)})\ket{w^t(i)}}^2
=2\sum_{\substack{t=1\\ \mathrm{odd}}}^{{\sf T}}\bar{p}_i(t){\alpha_t}\epsilon_t^i\leq 2\mathbb{E}\left[\alpha_{T_i}\epsilon_i^{T_i}\right].
\qedhere
\end{multline*}
\end{proof}

\section{Composition in Quantum Walks}\label{sec:graph-composition}

In this section, we state and prove \thm{graph-fwk}, which extends the quantum walk electric network framework by allowing a variable-time subroutine to implement the edge transition subroutine. This is similar to \cite[Theorem 3.10]{jeffery2022kDist}, with the following differences:
\begin{enumerate}
\item We require that our transition subroutine be \emph{reversible}, meaning that the mapping $\ket{u,i}\mapsto \ket{v,j}$ must be done by first computing some auxiliary information $g(u,i)$, from which the transition becomes easy (see \sec{variable-time}). This is not required in \cite{jeffery2022kDist}, but it is not a particularly stringent constraint, and we know of no quantum walk application where the transition is not done this way. 
\item \cite{jeffery2022kDist} does not allow the running time $T_{u,v}$ for some fixed $(u,v)$ to be a random variable, but instead depends on a parameter ${\sf T}_{u,v}$, which is an upper bound on the running time on input $(u,v)$, and requires that this value be efficiently computable from $(u,i)$ and $(v,j)$ (where $f_u(i)=v$ and $f_v(j)=u$). We do not require $\mathbb{E}[T_{u,v}]$ to be efficiently computable.
\item In \cite{jeffery2022kDist}, there are no parameters $\{\alpha_t\}_t$ (equivalently, $\alpha_t=1$ for all $t$). 
\item \cite{jeffery2022kDist} also extends the electric network framework by allowing the use of \emph{alternative neighbourhoods}. For simplicity, we do not incude this extension here, but there is no reason it could not be combined with our results.
\end{enumerate}

We first state \thm{graph-fwk}, which is too general to be easily understood. Immediately afterwards, we discuss some more easily understood consequences of this theorem, for quantum walks (\cor{variable-time-walk} and \cor{variable-time-MNRS}) and the special case of black-box search (\cor{variable-time} and \tabl{variable-time-search}), before proving the theorem in \sec{graph-fwk-proof}.

\begin{theorem}[Quantum Walks with Edge Composition]\label{thm:graph-fwk}
Fix the following, which may implicitly depend on some input $x$:
\begin{itemize}
\item a network $G$ with disjoint sets $V_0,V_\M\subset V(G)$ such that for any vertex, checking if $v\in V_0$ (resp. if $v\in V_\M$) can be done in unit complexity;
\item a reversible variable-time subroutine (see \sec{variable-time}) that implements the transition map (see \defin{QW-access}) with stopping times $\{T_{u,v}\}_{(u,v)\in\overrightarrow{E}(G)}$, which are random variables on $[{\sf T}]$, and errors $\{\epsilon_{u,v}^t\}_{(u,v)\in \overrightarrow{E}(G),t\in[{\sf T}]}$;
\item a subset $M\subseteq V_\M$, and distribution $\sigma$ on $V_0$.
\end{itemize}
Fix some positive real numbers $\{\alpha_t\}_{t=0}^{{\sf T}}$ such that $\alpha_0=1$; and positive real-valued ${\cal W}$ and ${\cal R}$, that may scale with $|x|$.
Suppose the following conditions hold.
\begin{description}
\item[Setup Subroutine:] The state $\ket{\sigma}=\sum_{u\in V_0}\sqrt{\sigma(u)}\ket{u}$ can be generated in cost ${\sf S}$, and furthermore, for any $u\in V_0$, $\sigma(u)$ can be computed in unit cost. 
\item[Star State Generation Subroutine:] There is a subroutine that generates $\{\ket{\psi_\star^G(u)}\}_{u\in V(G)}$ in unit cost (see \defin{QW-access}).
\item[Checking Subroutine:] There is an algorithm that checks, for any $u\in V_\M$, if $u\in M$, in unit cost.
\item[Positive Condition:] If $M\neq \emptyset$, then there exists a flow $\theta$ on $G$ (see \defin{flow}) such that:
	\begin{description}
	\item[P1] The boundary of $\theta$ is in $V_0\cup M$, and $\sum_{u\in V_0}\theta(u)=1$.\footnote{A $\sigma$-$M$-flow always satisfies conditions \textbf{P1} and \textbf{P2}. While we do not make it a strict requirement that all sources are in $V_0$ and all sinks in $M$, \textbf{P1} implies that we do not simply have flow coming in at $V_0$ and then leaving again at $V_0$. \textbf{P2} implies that the flow in $V_0$ is roughly distributed as $\sigma$. }
	\item[P2] $\frac{1}{3}\leq \sum_{u\in V_0}\frac{\theta(u)^2}{\sigma(u)}\leq 3$.
	\item[P3] $\displaystyle\sum_{(u,v)\in\overrightarrow{E}(G)}\frac{\theta(u,v)^2}{\w_{u,v}}\mathbb{E}\left[\sum_{t=0}^{T_{u,v}}\frac{1}{\alpha_t}\right]\leq {\cal R}$.
	\item[P4] $\displaystyle\sum_{(u,v)\in\overrightarrow{E}(G)}\frac{\theta(u,v)^2}{\w_{u,v}}\mathbb{E}\left[\frac{\epsilon_{u,v}^{T_{u,v}}}{\alpha_{T_{u,v}}}\right] =o(1/{\cal W})$;
	\end{description}
\item[Negative Condition:] If $M=\emptyset$, then 
	\begin{description}
	\item[N1] $\displaystyle\sum_{(u,v)\in\overrightarrow{E}(G)}\w_{u,v}\mathbb{E}\left[\sum_{t=0}^{T_{u,v}-1}\alpha_t\right]\leq {\cal W}$. 
	\item[N2] $\displaystyle \sum_{(u,v)\in\overrightarrow{E}(G)}\w_{u,v}\mathbb{E}\left[\alpha_{T_{u,v}}\epsilon_{u,v}^{T_{u,v}}\right]= o(1/{\cal R})$;
	\end{description}
\end{description}
Then there is a quantum algorithm that decides if $M=\emptyset$ or not with bounded error in complexity:
$$O\left({\sf S}+\sqrt{{\cal R}{\cal W}}\cdot\log{\sf T}\right).$$
\end{theorem}

We prove \thm{graph-fwk} in \sec{graph-fwk-proof} by defining a phase estimation algorithm and analysing it using \thm{phase-est-fwk}. First, we give some less general, easier to digest, corollaries of \thm{graph-fwk}. One thing that probably makes \thm{graph-fwk} particularly hard to parse is the parameters $\{\alpha_t\}_{t=0}^{\sf T}$, so we mention three cases of special interest. First, if we set $\alpha_t=1$ for all $t$, then we recover a version of \cite[Theorem~3.10]{jeffery2022kDist}: the resulting complexities are as if we had replaced the graph $G$ with a graph $G^{\sf T}$, in which each edge $(u,v)$ is replaced by a path of length $\mathbb{E}[T_{u,v}]+1$ (or an upper bound on $T_{u,v}+1$, in the case of \cite{jeffery2022kDist}), but now each transition has unit cost. 

We mention two other interesting choices of $\{\alpha_t\}_t$ that lead to two distinct ``variable-time'' quantum walk results. First, if we set $\alpha_t=t+1$, we satisfy the requirement $\alpha_0=1$, and we get:
$$\mathbb{E}\left[\sum_{t=0}^{T_{u,v}}\frac{1}{\alpha_t}\right] = \mathbb{E}\left[\sum_{t=1}^{T_{u,v}+1}\frac{1}{t}\right] = \Theta(\mathbb{E}[\log T_{u,v}])
\mbox{ and }
\mathbb{E}\left[\sum_{t=0}^{T_{u,v}}{\alpha_t}\right] = \mathbb{E}\left[\sum_{t=1}^{T_{u,v}+1}t\right] = \Theta(\mathbb{E}[T_{u,v}^2]).$$
Thus \textbf{P3} has only a logarithmic dependence on the transition times, whereas we get a weighted $\ell_2$-average in \textbf{N1}. 
On the other hand, by setting $\alpha_t=\frac{1}{1+t}$, $\mathbb{E}[\log T_{u,v}]$ and $\mathbb{E}[T_{u,v}^2]$ are swapped: now \textbf{P3} has a (different!) weighted $\ell_2$-average, whereas \textbf{N1} just has a log factor. 

If we additionally know the values $\mathbb{E}[T_{u,v}]$, then we can use $\alpha_t=1$ for all $t$, and scale $\w_{u,v}$ by a factor of $\mathbb{E}[T_{u,v}]$ to get:
$$\mathbf{P3:}\;\sum_{(u,v)\in\overrightarrow{E}(G)}\frac{\theta(u,v)^2}{\w_{u,v}} \leq {\cal R}
\quad
\mathbf{N1:}\;\sum_{(u,v)\in\overrightarrow{E}(G)}\w_{u,v}\mathbb{E}[T_{u,i}]^2 \leq {\cal W},$$
or alternatively, scaling $\w_{u,v}$ by $1/\mathbb{E}[T_{u,v}]$ gives:
$$\mathbf{P3:}\;\sum_{(u,v)\in\overrightarrow{E}(G)}\frac{\theta(u,v)^2}{\w_{u,v}}\mathbb{E}[T_{u,i}]^2 \leq {\cal R}
\quad
\mathbf{N1:}\;\sum_{(u,v)\in\overrightarrow{E}(G)}\w_{u,v} \leq {\cal W}.$$
Thus, when the expected stopping times are known, we can save a log factor, and also replace $\mathbb{E}[T_{u,v}^2]$ with $\mathbb{E}[T_{u,v}]^2$, which is better the higher the variance. 
We formally state these results for $\alpha_t=t+1$ in the following corollary (the $\alpha_t=1/(t+1)$ case can be worked out from \thm{graph-fwk}, but is less elegantly stated). We simplify things by only considering flows that are exactly from $\sigma$ to $M$, and assuming our subroutine has no error, but by \thm{graph-fwk}, it is clear that some deviation would also be tolerable.\footnote{In all the applications discussed in this section, we assume the subroutine has no error, for simplicity. We believe it is also interesting to further investigate the different ways the error can propagate depending on the choice of $\{\alpha_t\}_t$, but we leave this for future work. It is also possible, as in the case of the results of \sec{alg-composition}, that our error terms can be improved.}
\begin{corollary}[Variable-time quantum walks ($\alpha_t=t+1$)]\label{cor:variable-time-walk}
Fix a network $G$; marked set $M\subset V(G)$, and initial distribution $\sigma$ on $V(G)$, all of which may implicitly depend on an input $x$, although for simplicity, we assume ${\cal W}(G)$ only depends on $|x|$. Let $P$ be the transition matrix for the random walk on $G$, and $\pi$ its stationary distribution. Let ${\sf S}$ be the cost of generating $\ket{\sigma}$, and suppose $\{\ket{\psi_\star^G(u)}\}_{u\in V(G)}$ can be generated in unit cost.  Let ${\cal C}_{\sigma,\M}$ be an upper bound on ${\cal R}_{\sigma,M}(G){\cal W}(G)$ (see \defin{flow}) whenever $M\neq \emptyset$.

Suppose there is a variable-time subroutine\footnote{A variable-time subroutine that computes a bit is always reversible.} that checks, for any $u\in V(G)$, if $u\in M$, with stopping times $\{C_u\}_{u\in V(G)}$, which are random variables on $[{\sf C}]$, and all errors zero. Fix some (known) \emph{terminal} distribution $\tau$ on $V(G)$ such that $\tau(M)\geq \eps$ whenever $M\neq\emptyset$, and let ${\sf C}_{\mathrm{avg}}$ and ${\sf C}_{\mathrm{avg}}'$ be upper bounds such that whenever $M=\emptyset$,
$$\sqrt{\sum_{u\in V(G)}\tau(u)\mathbb{E}[C_u^2]} \leq {\sf C}_{\mathrm{avg}},
\quad\mbox{and}\quad
\sqrt{\sum_{u\in V(G)}\tau(u)\mathbb{E}[C_u]^2} \leq {\sf C}_{\mathrm{avg}}'.$$

Suppose there is a reversible variable-time subroutine that implements the transition map with stopping times $\{T_{u,v}\}_{(u,v)\in\overrightarrow{E}(G)}$, which are random variables on $[{\sf T}]$, and all errors zero. Let ${\sf T}_{\mathrm{avg}}$ and ${\sf T}_{\mathrm{avg}}'$ be upper bounds such that whenever $M=\emptyset$, 
$$\sqrt{\sum_{(u,v)\in\overrightarrow{E}(G)}\pi(u) P_{u,v}\mathbb{E}[T_{u,v}^2]} \leq {\sf T}_{\mathrm{avg}},
\quad\mbox{and}\quad
\sqrt{\sum_{(u,v)\in\overrightarrow{E}(G)}\pi(u) P_{u,v}\mathbb{E}[T_{u,v}]^2} \leq {\sf T}_{\mathrm{avg}}'.$$

\noindent\emph{\textbf{Claim 0:}} Suppose the checking cost is trivial, meaning for all $u\in V(G)$, $\mathbb{E}[C_u]=O(1)$. Then there is a quantum algorithm that detects if $M=\emptyset$ with bounded error in complexity: 
$$O\left({\sf S}+\sqrt{{\cal C}_{\sigma,\M}}{\sf T}_{\mathrm{avg}}\log^{1.5}{\sf T} \right).$$
Suppose in addition that the values $\mathbb{E}[T_{u,v}]=\mathbb{E}[T_{u,i}]$ are computable in the strong sense that for any $u$, we can generate a superposition proportional to $\sum_{i\in L(u)}\sqrt{\w_{u,i}\mathbb{E}[T_{u,i}]}\ket{i}$, and we have query access to $\w_u'=\sum_{i\in L(u)}\w_{u,i}\mathbb{E}[T_{u,i}]$. Then there is a quantum algorithm that detects if $M=\emptyset$ with bounded error in complexity:
$$O\left({\sf S}+\sqrt{{\cal C}_{\sigma,\M}}{\sf T}_{\mathrm{avg}}'\log{\sf T} \right).$$

\noindent\emph{\textbf{Claim I:}} Let ${\cal D}_{\M}$ be an upper bound such that $\sum_{u\in M}{\theta(u)^2}/{\tau_{M}(u)}\leq {\cal D}_{\M}$, where $\theta$ is the unit flow from $\sigma$ to $M$ with minimal energy and $\tau_{M}(u)=\tau(u)/\tau(M)$ is the normalized restriction of $\tau$ to $M$, whenever $M\neq \emptyset$. 
Then there is a quantum algorithm that detects if $M=\emptyset$ with bounded error in complexity:
$$O\left( {\sf S}+ \left(\sqrt{{\cal C}_{\sigma,\M}}{\sf T}_{\mathrm{avg}}+\sqrt{\frac{{\cal D}_{\M}}{\eps}}{\sf C}_{\mathrm{avg}}\right)\log^{1.5}({\sf TC})\right).$$ 
Suppose in addition that the values $\mathbb{E}[T_{u,v}]$ are computable in the same strong sense as Claim 0, and $\mathbb{E}[C_u]$ are computable in the the sense that for any $u\in V(G)$, we can query $\mathbb{E}[C_u]$ in unit cost. Then there is a quantum algorithm that detects in $M=\emptyset$ with bounded error in complexity:
$$O\left( {\sf S}+ \left(\sqrt{{\cal C}_{\sigma,\M}}{\sf T}_{\mathrm{avg}}'+\sqrt{\frac{{\cal D}_{\M}}{\eps}}{\sf C}_{\mathrm{avg}}'\right)\log({\sf TC})\right).$$

\noindent\emph{\textbf{Claim II:}} Suppose we are promised that either $M=\emptyset$, or $|M|=1$. Then there is a quantum algorithm that detects if $M=\emptyset$ with bounded error in complexity:
$$O\left( {\sf S}+ \left(\sqrt{{\cal C}_{\sigma,\M}}{\sf T}_{\mathrm{avg}}+\frac{1}{\sqrt{\eps}}{\sf C}_{\mathrm{avg}}\right)\log^{1.5}({\sf TC})\right).$$ 
Suppose in addition that the values $\mathbb{E}[T_{u,v}]$ and $\mathbb{E}[C_u]$ are computable as in Claim I. Then there is a quantum algorithm that detects in $M=\emptyset$ with bounded error in complexity:
$$O\left( {\sf S}+ \left(\sqrt{{\cal C}_{\sigma,\M}}{\sf T}_{\mathrm{avg}}'+\frac{1}{\sqrt{\eps}}{\sf C}_{\mathrm{avg}}'\right)\log({\sf TC})\right).$$ 

\noindent\emph{\textbf{Claim III:}} Let ${\cal C}_{\sigma,\tau}$ be an upper bound on ${\cal W}(G){\cal R}_{\sigma,\tau_{M}}(G)$ (see \defin{flow}). 
Then there is a quantum algorithm that detects if $M=\emptyset$ with bounded error in complexity:
$$O\left( {\sf S}+ \left(\sqrt{{\cal C}_{\sigma,\tau}}{\sf T}_{\mathrm{avg}}+\frac{1}{\sqrt{\eps}}{\sf C}_{\mathrm{avg}}\right)\log^{1.5}{\sf TC}\right).$$ 
Suppose in addition that the values $\mathbb{E}[T_{u,v}]$ and $\mathbb{E}[C_u]$ are computable as in Claim I. Then there is a quantum algorithm that detects in $M=\emptyset$ with bounded error in complexity:
$$O\left( {\sf S}+ \left(\sqrt{{\cal C}_{\sigma,\tau}}{\sf T}_{\mathrm{avg}}'+\frac{1}{\sqrt{\eps}}{\sf C}_{\mathrm{avg}}'\right)\log{\sf TC}\right).$$ 
\end{corollary}
\begin{proof}
Let $g_u=1$ if $u\in M$, and $g_u=0$ else. We will make a new graph $G'$, which is obtained from $G$ by adding, for each vertex $u\in V(G)$, a new vertex $(u,g_u)$ connected only to $u$ by an edge of weight $\w_{u,\M}$, to be defined.  That is, 
\begin{align*}
V(G') &:=\underbrace{V(G)}_{=:V_0}\cup \underbrace{\{(u,g_u):u\in V(G)\}}_{=:V_\M}
\mbox{and }\overrightarrow{E}(G') :=\overrightarrow{E}(G)\cup \{(u,(u,g_u)):u\in V(G)\},
\end{align*}
with $\w_{u,(u,g_u)}=\w_{u,\M}$. Then let $M'=M\times\{1\}\subseteq V_\M$ be the marked set of $G'$. The new edges $(u,(u,g_u))$ will have label $0$, in both directions, so: $f_u(0)=(u,g_u)$ and $f_{(u,g_u)}(0)=u$.  
Since we can generate $\{\ket{\psi_\star^G(u)}\}_{u\in V(G)}$ in unit cost, we can also generate 
$$\{\ket{\psi_\star^{G'}(u)}=\ket{\psi_\star^G(u)}+\sqrt{\w_{u,\M}}\ket{0}\}_{u\in V(G)}$$ 
in unit cost, assuming $\w_{u,\M}$ is chosen so that we can prepare a state proportional to $\sqrt{\w_u}\ket{1}+\sqrt{\w_{u,\M}}\ket{0}$. 
For any $u\in V_\M$, we can check if $u\in M'$ in unit cost by examining $g_u$, and generate $\ket{\sigma}$ in cost ${\sf S}$, by assumption. Thus, to apply \thm{graph-fwk}, we just need to check the positive and negative conditions. \textbf{P4} and \textbf{N2} are satisfied because we assume no error, and for the positive condition, we will always choose $\theta$ to be a $\sigma$-$M'$ flow, so \textbf{P1} and \textbf{P2} are automatically satisifed. We thus only need to check \textbf{P3} and \textbf{N1}.
We will use $\alpha_t=t+1$ for the first half of the proof of each claim, and $\alpha_t=1$ for the second half. We note that we immediately have a variable-time subroutine for the transition map with $T_{u,(u,g_u)}=C_u$. 

\vskip7pt

\noindent\textbf{Proof of Claim 0:} To prove Claim 0, we let
$\w_{u,\M}:=\w_u.$
Note that this makes preparing a state proportional to $\sqrt{\w_u}\ket{0}+\sqrt{\w_{u,\M}}\ket{1}$ trivial.  

\noindent\textbf{P3:} 
Let $\theta$ be the flow from $\sigma$ to $M$ with minimal energy ${\cal R}_{\sigma,M}(G)$ (see \defin{flow}), and extend it to a flow on $G'$ from $\sigma$ to $M'$ by setting $\theta(u,(u,g_u)) = \theta(u)$ for all $u\in M$. Then:
\begin{equation}\label{eq:claim-0-pos}
\begin{split}
\sum_{(u,v)\in\overrightarrow{E}(G')}\frac{\theta(u,v)^2}{\w_{u,v}}\mathbb{E}\left[\sum_{t=0}^{T_{u,v}}\frac{1}{t+1}\right]
&= \sum_{(u,v)\in\overrightarrow{E}(G)}\frac{\theta(u,v)^2}{\w_{u,v}}\mathbb{E}\left[\sum_{t=0}^{T_{u,v}}\frac{1}{t+1}\right] + \sum_{u\in M}\frac{\theta(u)^2}{\w_{u,\M}}\mathbb{E}\left[\sum_{t=0}^{C_{u}}\frac{1}{t+1}\right]\\
&\leq {\cal R}_{\sigma,M}(G)(\ln ({\sf T}+1)+1)+\frac{3}{2}\sum_{u\in M}\frac{\theta(u)^2}{\w_u}
\end{split}
\end{equation}
since $C_u=1$ for all $u$, by assumption.
By Jensen's inequality, we have:
\begin{align*}
\left(\frac{\sum_{v\in\Gamma(u)}\w_{u,v}\frac{\theta(u,v)}{\w_{u,v}}}{\sum_{v\in\Gamma(u)}\w_{u,v}} \right)^2 &\leq \frac{\sum_{v\in\Gamma(u)}\w_{u,v}\left(\frac{\theta(u,v)}{\w_{u,v}}\right)^2}{\sum_{v\in\Gamma(u)}\w_{u,v}}\\
\frac{\left(\sum_{v\in\Gamma(u)}\theta(u,v)\right)^2}{\sum_{v\in\Gamma(u)}\w_{u,v}} &\leq \sum_{v\in\Gamma(u)}\frac{\theta(u,v)^2}{\w_{u,v}}\\
\sum_{u\in M}\frac{\theta(u)^2}{\w_u} &\leq \sum_{u\in M}\sum_{v\in\Gamma(u)}\frac{\theta(u,v)^2}{\w_{u,v}} \leq {\cal R}_{\sigma,M}(G),
\end{align*}
so we can continue from \eq{claim-0-pos}, using ${\cal W}(G){\cal R}_{\sigma,M}(G)\leq {\cal C}_{\sigma,\M}$:
\begin{equation}\label{eq:claim-0-pos-2}
\begin{split}
\sum_{(u,v)\in\overrightarrow{E}(G')}\frac{\theta(u,v)^2}{\w_{u,v}}\mathbb{E}\left[\sum_{t=0}^{T_{u,v}}\frac{1}{t+1}\right]
&\leq 4\frac{{\cal C}_{\sigma,\M}}{{\cal W}(G)}\ln {\sf T}=:{\cal R}.
\end{split}
\end{equation}

\noindent\textbf{N1:} We use $\pi(u) P_{u,v} = \frac{\w_{u,v}}{2{\cal W}(G)}$ from \eq{detailed-balance}. 
If $M= \emptyset$, we have:
\begin{equation}\label{eq:claim-0-neg}
\begin{split}
\sum_{(u,v)\in \overrightarrow{E}(G')}\w_{u,v}\mathbb{E}\left[\sum_{t=0}^{T_{u,v}}({t+1})\right]
&= \sum_{(u,v)\in\overrightarrow{E}(G)}\w_{u,v}\mathbb{E}\left[\sum_{t=0}^{T_{u,v}}({t+1})\right] + \sum_{u\in V(G)}\w_{u,\M}\mathbb{E}\left[\sum_{t=0}^{C_{u}}({t+1})\right]\\
&\leq 2{\cal W}(G)\sum_{(u,v)\in\overrightarrow{E}(G)}\pi(u)P_{u,v}\mathbb{E}[T_{u,v}^2] + 2\sum_{u\in V(G)}\w_u\\
&\leq 2{\cal W}(G){\sf T}_{\mathrm{avg}}^2+2{\cal W}(G)=:{\cal W}.
\end{split}
\end{equation}

\noindent To complete the first half of the proof of Claim 0, we apply \thm{graph-fwk} using the fact that:
\begin{align*}
\sqrt{\cal RW} &= O\left(\sqrt{{\cal C}_{\sigma,\M}{\sf T}_{\mathrm{avg}}^2\log {\sf T}}\right).
\end{align*}

\noindent\textbf{Known $\mathbb{E}[T_{u,v}]$:} For the second half of Claim 0, we scale the weights of $G'$ as follows:
$$\forall (u,v)\in\overrightarrow{E}(G),\; \w_{u,v}' := \w_{u,v}\mathbb{E}[T_{u,v}]
\quad\mbox{and}\quad
\forall u\in {V}(G),\; \w_{u,\M}':=\w_{u,\M}=\w_u.$$ 
Note that by changing the weights, we have changed the star states, but we can generate them by first generating a state proportional to $\sum_{i\in L(u)}\sqrt{\w_{u,i}\mathbb{E}[T_{u,i}]}\ket{i}$, which we assume can be done in unit cost, and then putting a phase in front of each $\ket{i}$ depending on if $i\in L^+(u)$ or $i\in L^-(u)$. We have also changed the total weight of each vertex, $\w_u'=\sum_{i}\w_{u,i}\mathbb{E}[T_{u,i}]$, but we assume we still have query access to these values.

Then for \textbf{P3}, now using $\alpha_t=1$, in place of \eq{claim-0-pos} and \eq{claim-0-pos-2}, we have:
\begin{equation}\label{eq:claim-0-2-pos}
\begin{split}
\sum_{(u,v)\in\overrightarrow{E}(G')}\frac{\theta(u,v)^2}{\w_{u,v}'}\mathbb{E}\left[\sum_{t=0}^{T_{u,v}}1\right]
&= \sum_{(u,v)\in\overrightarrow{E}(G)}\frac{\theta(u,v)^2}{\w_{u,v}\mathbb{E}[T_{u,v}]}\mathbb{E}\left[T_{u,v}+1\right] + \sum_{u\in M}\frac{\theta(u)^2}{\w_{u}}\mathbb{E}\left[C_{u}+1\right]\\
&\leq 2{\cal R}_{\sigma,M}(G)+2{\cal R}_{\sigma,M}(G)
\leq 4\frac{{\cal C}_{\sigma,\M}}{{\cal W}(G)}=:{\cal R}',
\end{split}
\end{equation}
and for \textbf{N1}, in place of \eq{claim-0-neg}, we have:
\begin{equation}\label{eq:claim-I-2-neg}
\begin{split}
\sum_{(u,v)\in \overrightarrow{E}(G')}\w_{u,v}'\mathbb{E}\left[\sum_{t=0}^{T_{u,v}}1\right]
&= \sum_{(u,v)\in\overrightarrow{E}(G)}\w_{u,v}\mathbb{E}[T_{u,v}]\mathbb{E}\left[T_{u,v}+1\right] + \sum_{u\in V(G)}\w_{u}\mathbb{E}\left[{C_{u}}+1\right]\\
&\leq 4{\cal W}(G)\sum_{(u,v)\in\overrightarrow{E}(G)}\pi(u)P_{u,v}\mathbb{E}[T_{u,v}]^2 + 2{\cal W}(G)\\
&\leq 4{\cal W}(G)({\sf T}_{\mathrm{avg}}')^2+2{\cal W}(G)=:{\cal W}'.
\end{split}
\end{equation}
Then a simple calculation of $\sqrt{{\cal R}'{\cal W}'}$ shows that the second part of Claim 0 follows from \thm{graph-fwk}.

\vskip7pt

\noindent\textbf{Proof of Claim I:} 
For the proof of Claim I, we now let 
$$\w_{u,\M}:=\frac{{\cal W}(G){\cal D}_\M\tau(u)}{\eps{\cal C}_{\sigma,\M}}=\frac{{\cal W}(G){\cal D}_\M\tau_{M}(u)\tau(M)}{\eps{\cal C}_{\sigma,\M}}\mbox{ if }u\in M.$$
With this choice of $\w_{u,\M}$, preparing a state proportional to $\sqrt{\w_u}\ket{1}+\sqrt{\w_{u,\M}}\ket{0}$ can be done using a single qubit rotation that depends on $\frac{{\cal W}(G){\cal D}_{\M}}{\eps {\cal C}_{\sigma,\M}}\frac{\tau(u)}{\w_u}$.

\noindent\textbf{P3:} 
Let $\theta$ be the flow from $\sigma$ to $M$ with minimal energy ${\cal R}_{\sigma,M}(G)$ (see \defin{flow}), and extend it to a flow on $G'$ from $\sigma$ to $M'$ by setting $\theta(u,(u,g_u)) = \theta(u)$ for all $u\in M$. Then we have:
\begin{equation}\label{eq:claim-I-pos}
\begin{split}
\sum_{(u,v)\in\overrightarrow{E}(G')}\frac{\theta(u,v)^2}{\w_{u,v}}\mathbb{E}\left[\sum_{t=0}^{T_{u,v}}\frac{1}{t+1}\right]
&= \sum_{(u,v)\in\overrightarrow{E}(G)}\frac{\theta(u,v)^2}{\w_{u,v}}\mathbb{E}\left[\sum_{t=0}^{T_{u,v}}\frac{1}{t+1}\right] + \sum_{u\in M}\frac{\theta(u)^2}{\w_{u,\M}}\mathbb{E}\left[\sum_{t=0}^{C_{u}}\frac{1}{t+1}\right]\\
&\leq 2{\cal R}_{\sigma,M}(G)\ln {\sf T}+2\frac{{\cal C}_{\sigma,\M}}{{\cal W}(G){\cal D}_\M}\frac{\eps}{\tau(M)}\sum_{u\in M}\frac{\theta(u)^2}{\tau_{M}(u)}\ln{\sf C}\\
&\leq 2\frac{{\cal C}_{\sigma,\M}}{{\cal W}(G)}\ln {\sf T}+2\frac{{\cal C}_{\sigma,\M}}{{\cal W}(G)}\ln{\sf C}=:{\cal R},
\end{split}
\end{equation}
where we used ${\cal W}(G){\cal R}_{\sigma,M}\leq {\cal C}_{\sigma,\M}$, $\sum_{u\in M}\theta(u)^2/\tau_{M}(u)\leq {\cal D}_\M$ and $\eps\leq \tau(M)$. 

\vskip7pt

\noindent\textbf{N1:} If $M= \emptyset$, we have:
\begin{equation}\label{eq:claim-I-neg}
\begin{split}
\sum_{(u,v)\in \overrightarrow{E}(G')}\w_{u,v}\mathbb{E}\left[\sum_{t=0}^{T_{u,v}}({t+1})\right]
&= \sum_{(u,v)\in\overrightarrow{E}(G)}\w_{u,v}\mathbb{E}\left[\sum_{t=0}^{T_{u,v}}({t+1})\right] + \sum_{u\in V(G)}\w_{u,\M}\mathbb{E}\left[\sum_{t=0}^{C_{u}}({t+1})\right]\\
&\leq 2{\cal W}(G)\sum_{(u,v)\in\overrightarrow{E}(G)}\pi(u)P_{u,v}\mathbb{E}[T_{u,v}^2] + \sum_{u\in V(G)}\frac{{\cal W}(G){\cal D}_\M \tau(u)}{\eps{\cal C}_{\sigma,\M}}\mathbb{E}[C_u^2]\\
&\leq 2{\cal W}(G){\sf T}_{\mathrm{avg}}^2+\frac{{\cal W}(G){\cal D}_\M}{{\cal C}_{\sigma,\M}}\frac{1}{\eps}{\sf C}_{\mathrm{avg}}^2=:{\cal W}.
\end{split}
\end{equation}

\noindent To complete the first half of the proof of Claim I, we apply \thm{graph-fwk} using the fact that:
\begin{align*}
\sqrt{\cal RW} &= O\left(\sqrt{\left( {\cal C}_{\sigma,\M}{\sf T}_{\mathrm{avg}}^2+\frac{{\cal D}_\M}{\eps}{\sf C}_{\mathrm{avg}}^2 \right)\log {\sf TC}}\right).
\end{align*}

\noindent\textbf{Known $\mathbb{E}[T_{u,v}]$ and $\mathbb{E}[C_u]$:} For the second half of Claim I, we scale the weights of $G'$ as follows:
\begin{equation}\label{eq:claim-I-weight-scaling}
\forall (u,v)\in\overrightarrow{E}(G),\; \w_{u,v}' := \w_{u,v}\mathbb{E}[T_{u,v}]
\mbox{ and }
\forall u\in V(G),\; \w_{u,\M}':=\w_{u,\M}\mathbb{E}[C_u]=\frac{{\cal W}(G){\cal D}_{\M}\tau(u)\mathbb{E}[C_u]}{\eps{\cal C}_{\sigma,\M}}.
\end{equation}
Note that by changing the weights, we have changed the star states, but we can generate them as follows. Let $\w_u'=\sum_{i\in L(u)}\w_{u,i}\mathbb{E}[T_{u,i}]$, which we can query, by assumption. Then we first prepare a state proportional to 
$\sqrt{\w_u'}\ket{1}+\sqrt{\w_{u,\M}\mathbb{E}[C_u]}\ket{0}$, using a single qubit rotation, and then map $\ket{1}$ to a state proportional to $\sum_{i\in L(u)}\sqrt{\w_{u,i}\mathbb{E}[T_{u,i}]}\ket{i}$, which we assume can be done in unit cost. Then all that remains to be done is to put a phase in front of each $\ket{i}$ depending on if $i\in L^+(u)$ or $i\in L^-(u)$.

Then for the positive analysis, now using $\alpha_t=1$, in place of \eq{claim-I-pos}, we have:
\begin{equation}\label{eq:claim-I-2-pos}
\begin{split}
\sum_{(u,v)\in\overrightarrow{E}(G')}\frac{\theta(u,v)^2}{\w_{u,v}'}\mathbb{E}\left[\sum_{t=0}^{T_{u,v}}1\right]
&= \sum_{(u,v)\in\overrightarrow{E}(G)}\frac{\theta(u,v)^2}{\w_{u,v}\mathbb{E}[T_{u,v}]}\mathbb{E}\left[T_{u,v}+1\right] + \sum_{u\in M}\frac{\theta(u)^2}{\w_{u,\M}\mathbb{E}[C_u]}\mathbb{E}\left[C_{u}+1\right]\\
&\leq 2{\cal R}_{\sigma,M}(G)+2\frac{{\cal C}_{\sigma,\M}}{{\cal W}(G){\cal D}_\M}\frac{\eps}{\tau(M)}\sum_{u\in M}\frac{\theta(u)^2}{\tau_M(u)}\\
&\leq 2\frac{{\cal C}_{\sigma,\M}}{{\cal W}(G)}+2\frac{{\cal C}_{\sigma,M}}{{\cal W}(G)}=:{\cal R}',
\end{split}
\end{equation}
and for the negative analysis, in place of \eq{claim-I-neg}, we have:
\begin{equation}\label{eq:claim-I-2-neg}
\begin{split}
\sum_{(u,v)\in \overrightarrow{E}(G')}\w_{u,v}'\mathbb{E}\left[\sum_{t=0}^{T_{u,v}}1\right]
&= \sum_{(u,v)\in\overrightarrow{E}(G)}\w_{u,v}\mathbb{E}[T_{u,v}]\mathbb{E}\left[T_{u,v}+1\right] + \sum_{u\in V(G)}\w_{u,\M}\mathbb{E}[C_u]\mathbb{E}\left[{C_{u}}+1\right]\\
&\leq 4{\cal W}(G)\sum_{(u,v)\in\overrightarrow{E}(G)}\pi(u)P_{u,v}\mathbb{E}[T_{u,v}]^2 + 2\sum_{u\in V(G)}\frac{{\cal W}(G){\cal D}_\M}{\eps{\cal C}_{\sigma,\M}}\tau(u)\mathbb{E}[C_u]^2\\
&\leq 4{\cal W}(G)({\sf T}_{\mathrm{avg}}')^2+2\frac{{\cal W}(G){\cal D}_\M}{{\cal C}_{\sigma,\M}}\frac{1}{\eps}({\sf C}_{\mathrm{avg}}')^2=:{\cal W}'.
\end{split}
\end{equation}
Then a simple calculation of $\sqrt{{\cal R}'{\cal W}'}$ show that the second part of Claim I follows from \thm{graph-fwk}.

\vskip7pt

\noindent\textbf{Proof of Claim II:} Next, we consider the special case where there is a unique marked element if $M\neq \emptyset$. In that case, when $M=\{m\}$  is non-empty, we have, for any $\sigma$-$M$ flow, $\theta(m)=-1$ (all flow leaves at $m$), so:
$$\sum_{u\in M} \frac{\theta(u)^2}{\tau_M(u)}=\frac{1}{\tau_M(m)} = \frac{\tau(M)}{\tau(m)}=1.$$
Thus Claim II follows from applying Claim I with ${\cal D}_{\M}=1$.

\vskip7pt

\noindent\textbf{Proof of Claim III:} 
Finally, let $\tau$ be the known distribution described in Claim III. If such a $\tau$ is known, we can define 
$$\w_{u,\M} := \frac{\tau(u) {\cal W}(G)}{\eps{\cal C}_{\sigma,\tau}}=\frac{\tau_M(u)\tau(M){\cal W}(G)}{\eps{\cal C}_{\sigma,\tau}}\mbox{ when }u\in M.$$
With this choice of $\w_{u,\M}$, preparing a state proportional to $\sqrt{\w_u}\ket{1}+\sqrt{\w_{u,\M}}\ket{0}$ can be done using a single qubit rotation that depends on $\frac{{\cal W}(G)}{\eps{\cal C}_{\sigma,\tau}}\frac{\tau(u)}{\w_u}$.

\vskip7pt

\noindent\textbf{P3:} 
Let $\theta$ be the flow from $\sigma$ to $\tau_M$ with minimal energy ${\cal R}_{\sigma,\tau_M}(G)$ (see \defin{resistance}), and extend it to a flow on $G'$ from $\sigma$ to $M'$ by setting $\theta(u,(u,g_u)) = \tau_M(u)$ for all $u\in M$. Then, similar to \eq{claim-I-pos}, we have:
\begin{align*}
\sum_{(u,v)\in\overrightarrow{E}(G')}\frac{\theta(u,v)^2}{\w_{u,v}}\mathbb{E}\left[\sum_{t=0}^{T_{u,v}}\frac{1}{t+1}\right]
&\leq 2\sum_{(u,v)\in\overrightarrow{E}(G)}\frac{\theta(u,v)^2}{\w_{u,v}}\ln {\sf T} + 2\sum_{u\in M}\frac{\tau_M(u)^2}{\w_{u,\M}}\ln{\sf C}\\
&= 2{\cal R}_{\sigma,\tau_M}(G)\ln {\sf T}+2\frac{\eps}{\tau(M)}\frac{{\cal C}_{\sigma,\tau}}{{\cal W}(G)}\sum_{u\in M}\frac{\tau_M(u)^2}{\tau_M(u)}\ln{\sf C}\\
&\leq 2\frac{{\cal C}_{\sigma,\tau}}{{\cal W}(G)}\ln {\sf T}+2\frac{{\cal C}_{\sigma,\tau}}{{\cal W}(G)}\ln{\sf C}=:{\cal R}.
\end{align*}

\noindent\textbf{N1:} If $M=\emptyset$, we have, similar to \eq{claim-I-neg}:
\begin{align*}
\sum_{(u,v)\in \overrightarrow{E}(G')}\w_{u,v}\mathbb{E}\left[\sum_{t=0}^{T_{u,v}}{t+1}\right]
&\leq 2\sum_{(u,v)\in\overrightarrow{E}(G)}\w_{u,v}\mathbb{E}[T_{u,v}^2] + 2\sum_{u\in V(G)}\w_{u,\M}\mathbb{E}[C_u^2]\\
&\leq 4 {\cal W}(G){\sf T}_{\mathrm{avg}}^2 + \frac{{\cal W}(G)}{{\cal C}_{\sigma,\tau}}\frac{1}{\eps}{\sf C}_{\mathrm{avg}}^2=:{\cal W}.
\end{align*}

\noindent To complete the proof of the first half of Claim III, we apply \thm{graph-fwk} using the fact that
\begin{align*}
\sqrt{\cal RW} &= O\left(\left( {\cal C}_{\sigma,\tau}{\sf T}_{\mathrm{avg}}^2+\frac{1}{\eps}{\sf C}_{\mathrm{avg}}^2 \right)\log {\sf TC}\right).
\end{align*}
The proof of the second half of Claim III proceeds just as in Claim I. 
\end{proof}

The expressions we achieve in \cor{variable-time-walk} in the case when checking is non-trivial and $M$ is not restricted to being a singleton are somewhat complicated. However, in the special case when $\sigma=\pi$, we can get a bound comparable to the MNRS framework. The following corollary also gives an alternative proof to the one in \cite{apers2019UnifiedFrameworkQWSearch} that the MNRS framework is a special case of the electric network framework (up to $\log{1}/{\pi_{\min}}$ factors, which we suspect can be removed).

\begin{corollary}[Variable-time MNRS]\label{cor:variable-time-MNRS}
Fix a network $G$ and marked set $M\subset V(G)$ that may implicitly depend on an input $x$. Let $P$ be the transition matrix for the random walk on $G$, and $\pi$ its stationary distribution. Let ${\sf S}$ be the cost of generating $\ket{\pi}$, and suppose $\{\ket{\psi_\star^G(u)}\}_{u\in V(G)}$ can be generated in unit cost.
Let $\delta$ be a lower bound on the spectral gap of $P$, and $\eps$ a lower bound on $\pi(M)$ whenever $M\neq\emptyset$. Let $\pi_{\min}=\min_{u\in V(G)}\pi(u)$.

Suppose there is a variable-time subroutine that checks, for any $u\in V(G)$, if $u\in M$, with stopping times $\{C_u\}_{u\in V(G)}$, which are random variables on $[{\sf C}]$, and all errors zero. Let ${\sf C}_{\mathrm{avg}}$ and ${\sf C}_{\mathrm{avg}}'$ be upper bounds such that whenever $M=\emptyset$,
$$\sqrt{\sum_{u\in V(G)}\pi(u)\mathbb{E}[C_u^2]} \leq {\sf C}_{\mathrm{avg}},
\quad\mbox{and}\quad
\sqrt{\sum_{u\in V(G)}\pi(u)\mathbb{E}[C_u]^2} \leq {\sf C}_{\mathrm{avg}}'.$$
Suppose there is a reversible variable-time subroutine that implements the transition map with stopping times $\{T_{u,v}\}_{(u,v)\in\overrightarrow{E}(G)}$, which are random variables on $[{\sf T}]$, and all errors zero. Let ${\sf T}_{\mathrm{avg}}$ and ${\sf T}_{\mathrm{avg}}'$ be upper bounds such that whenever $M=\emptyset$, 
$$\sqrt{\sum_{(u,v)\in\overrightarrow{E}(G)}\pi(u) P_{u,v}\mathbb{E}[T_{u,v}^2]} \leq {\sf T}_{\mathrm{avg}},
\quad\mbox{and}\quad
\sqrt{\sum_{(u,v)\in\overrightarrow{E}(G)}\pi(u) P_{u,v}\mathbb{E}[T_{u,v}]^2} \leq {\sf T}_{\mathrm{avg}}'.$$
Then there is a quantum algorithm that detects if $M=\emptyset$ with bounded error in complexity: 
$$O\left({\sf S}+\frac{1}{\sqrt{\eps}}\left(\frac{1}{\sqrt{\delta}}{\sf T}_{\mathrm{avg}}+{\sf C}_{\mathrm{avg}}\right)\sqrt{\log\frac{1}{\pi_{\min}}}\log^{1.5}({\sf T}{\sf C})\right).$$
Suppose in addition that the values $\mathbb{E}[T_{u,v}]=\mathbb{E}[T_{u,i}]$ are computable in the strong sense that for any $u$, we can generate a superposition proportional to $\sum_{i\in L(u)}\sqrt{\w_{u,i}\mathbb{E}[T_{u,i}]}\ket{i}$, and we have query access to $\w_u'=\sum_{i\in L(u)}\w_{u,i}\mathbb{E}[T_{u,i}]$ and $\mathbb{E}[C_u]$. Then there is a quantum algorithm that decides if $M=\emptyset$ with bounded error in complexity:
$$O\left({\sf S}+\frac{1}{\sqrt{\eps}}\left(\frac{1}{\sqrt{\delta}}{\sf T}_{\mathrm{avg}}'+{\sf C}_{\mathrm{avg}}'\right)\sqrt{\log\frac{1}{\pi_{\min}}}\log({\sf T}{\sf C})\right).$$
\end{corollary}
\begin{proof}
As in the proof of \cor{variable-time-walk}, we will apply \thm{graph-fwk} to a modified graph $G'$ in which we connect each $u\in V(G)$ to a new vertex $(u,g_u)$, where $g_u=1$ if and only if $u\in M$, by an edge of weight $\w_{u,\M}=\delta\w_u$.

We will use the following process -- which is similar to a random walk on $G'$ starting from $\pi$ except for the log factor in $p$ -- to design a flow on $G'$.
\begin{enumerate}
    \item Sample a vertex $u$ from $\pi$.
    \item Repeat:
    \begin{enumerate}
        \item With probability $p=\frac{\delta}{\log({2}/{\pi_{\min}^2})}$, check if $u\in M$, and if so, halt.
        \item Take a step to some $v$ according to $P$ (the random walk on $G$) by setting $u\leftarrow v$.
    \end{enumerate}
\end{enumerate}
This process gives rise to a Markov chain $X_1,\dots,X_{\tau}$ on $V(G)$ for a random variable $\tau$ with  
$$\mathbb{E}[\tau]=O\left(\frac{1}{p\pi(M)}\right)=O\left(\frac{\log\frac{1}{\pi_{\min}}}{\eps\delta}\right).$$
From this process, we can define a flow on $G$ as follows. Let $[u\rightarrow v]$ be the number of times the above process moves from $u$ to $v$ and define for any $\{u,v\}\in E(G)$:
$$\theta(u,v) := \mathbb{E}[u\rightarrow v] - \mathbb{E}[v\rightarrow u].$$
This satisfies $\theta(u,v) = -\theta(v,u)$, and letting $N_u^\rightarrow$ be the number of times we leave $u$, and $N_u^{\leftarrow}$ be the number of times we enter $u$, we have:
$$\theta(u) = \sum_v\theta(u,v) =\left(\mathbb{E}[N_u^\rightarrow]-\mathbb{E}[N_u^\leftarrow]\right)
= \Pr[X_1=u] - \Pr[X_{\tau}=u].$$
It is clear that $\Pr[X_1=u]=\pi(u)$. We also only ever output if $u\in M$, so $\Pr[X_{\tau}=u]=0$ unless $u\in M$. We can extend this to a flow on $G'$ by setting $\theta(u,(u,1))=\theta(u)$ for all $u\in M$. Let $\tilde\pi$ be the distribution on $M$ defined $\tilde\pi(u) = \pi(u)-\theta(u)$ for all $u\in M$ -- in other words, $\tilde\pi(u)$ is the probability that $X_{\tau}=u$. Since there are $1/p=\frac{1}{\delta}\log\frac{1}{\pi_{\min}^2}$ steps between  checks, in expectation, each time we check if $u\in M$, we are distributed according to some distribution $\pi'$ that is expected $\pi_{\min}/2$-close to $\pi$~\cite[Theorem 12.3]{levin2017MarkovChainsMixingTimes}, meaning that for each $u\in M$
$$\mathbb{E}[\pi_M'(u)] \leq \frac{\pi(u)+\pi_{\min}/2}{\pi(M)-\pi_{\min}/2}\leq \frac{\frac{3}{2}\pi(u)}{\frac{1}{2}\pi(M)}=3\pi_M(u)$$
and so $\tilde\pi(u)\leq 3\pi_M(u)$, and thus:
$$|\theta(u)| \leq \tilde\pi(u)+\pi(u) \leq 3\pi_M(u)+\pi(M)\pi_M(u) \leq 4\pi_M(u).$$
Then we have (using $\alpha_t=t+1$):
\begin{align*}
2{\cal W}(G)\sum_{(u,v)\in\overrightarrow{E}(G')}\frac{\theta(u,v)^2}{\w_{u,v}}\mathbb{E}\left[\sum_{t=0}^{T_{u,v}}\frac{1}{t+1}\right] &\leq \sum_{(u,v)\in\overrightarrow{E}(G)}\frac{\theta(u,v)^2}{\pi(u)P_{u,v}}\mathbb{E}[\log T_{u,v}] + \sum_{u\in M}\frac{\theta(u)^2}{\delta\pi(u)}\mathbb{E}[\log C_u]\\
&\leq \sum_{(u,v)\in\overrightarrow{E}(G)}\frac{\theta(u,v)^2}{\pi(u)P_{u,v}}\log{\sf T} + \frac{16}{\delta\pi(M)^2}\sum_{u\in M}{\pi(u)}\log{\sf C}\\
&\leq \sum_{(u,v)\in\overrightarrow{E}(G)}\frac{\theta(u,v)^2}{\pi(u)P_{u,v}}\log{\sf T} + \frac{16}{\delta\eps}\log{\sf C}.
\end{align*}
Above, we used $\w_{u,v} = \pi(u)P_{u,v}/(2{\cal W}(G))$, $\theta(u)^2 \leq 16 \pi_M(u)^2 = 16 \pi(u)^2/\pi(M)^2$, and $\eps \leq \pi(M) = \sum_{u\in M}\pi(u)$.
We will show that,
\begin{equation}\label{eq:MNRS-eps-delta}
\sum_{(u,v)\in \overrightarrow{E}(G)}\frac{\theta(u,v)^2}{\pi(u)P_{u,v}} \leq \mathbb{E}[\tau] \leq \frac{c\log \frac{1}{\pi_{\min}}}{\eps\delta}
\end{equation}
for some constant $c$, from which it follows that
\begin{align*}
\sum_{(u,v)\in\overrightarrow{E}(G')}\frac{\theta(u,v)^2}{\w_{u,v}}\mathbb{E}\left[\sum_{t=0}^{T_{u,v}}\frac{1}{t+1}\right] &\leq \frac{1}{2{\cal W}(G)}\frac{1}{\eps\delta}\left( c\log\frac{1}{\pi_{\min}}\log{\sf T} + 16\log{\sf C}\right)=:{\cal R}.
\end{align*}
Then since:
\begin{align*}
\sum_{(u,v)\in\overrightarrow{E}(G')}\w_{u,v}\mathbb{E}\left[\sum_{t=0}^{T_{u,v}}(t+1)\right] &\leq 2{\cal W}(G)\left(\sum_{(u,v)\in\overrightarrow{E}(G)}\pi(u)P_{u,v}\mathbb{E}[T_{u,v}^2]+\sum_{u\in V(G)}\delta\pi(u)\mathbb{E}[C_u^2]\right)=:{\cal W}
\end{align*}
we have
$${\cal R W}=O\left(\left(\frac{1}{\eps\delta}\sum_{(u,v)\in\overrightarrow{E}(G)}\pi(u)P_{u,v}\mathbb{E}[T_{u,v}^2]+\frac{1}{\eps}\sum_{u\in V(G)}\pi(u)\mathbb{E}[C_u^2]\right)\log\frac{1}{\pi_{\min}}\log{\sf TC} \right)$$
and so the claim follows from \thm{graph-fwk}. Thus, it remains only to establish \eq{MNRS-eps-delta}, which we now undertake.


Let $[u\rightarrow v]_t$ be the event that we move from $u$ to $v$ in the $t$-th step, and let $E_t(u)$ be the event that $X_t=u$, and $\tau>t$ (i.e. we don't stop there). Then:
\begin{align*}
    \mathbb{E}[u\rightarrow v] &= \sum_{t=1}^{\infty}\Pr[[u\rightarrow v]_t]
    = 
    \sum_{t=1}^{\infty}\Pr[E_{t}(u)]P_{u,v}.
\end{align*}
Then since $\theta(u,v) = \mathbb{E}[u\rightarrow v] - \mathbb{E}[v\rightarrow u]$, by definition, we have:
\begin{equation}\label{eq:MNRS-intermediate}
\begin{split}
    \sum_{(u,v)\in \overrightarrow{E}(G)}\frac{\theta(u,v)^2}{\pi(u)P_{u,v}}
    &= \frac{1}{2}\sum_{u\in V(G)}\sum_{v\in \Gamma(u)} \pi(u)P_{u,v}\left(\frac{\mathbb{E}[u\rightarrow v]}{\pi(u)P_{u,v}}-\frac{\mathbb{E}[v\rightarrow u]}{\pi(v)P_{v,u}}\right)^2\\
    &= \frac{1}{2}\sum_{u\in V(G)}\sum_{v\in \Gamma(u)} \pi(u)P_{u,v}\left(\frac{1}{\pi(u)}\sum_{t=1}^{\infty}\Pr[E_t(u)]-\frac{1}{\pi(v)}\sum_{t=1}^{\infty}\Pr[E_t(v)]\right)^2\\
    &= \frac{1}{2}\sum_{u\in V(G)}\sum_{v\in \Gamma(u)} \pi(u)P_{u,v}\left(\frac{\mathbb{E}[N_u^{\rightarrow}]}{\pi(u)}-\frac{\mathbb{E}[N_v^{\rightarrow}]}{\pi(v)}\right)^2. 
\end{split}
\end{equation}

From \eq{MNRS-intermediate} and the detailed balance condition $\pi(u)P_{u,v} = \pi(v)P_{v,u}$, it follows that:
\begin{equation}\label{eq:MNRS-big-step}
\begin{split}
\sum_{(u,v)\in \overrightarrow{E}(G)} \frac{\theta(u,v)^2}{\pi(u)P_{u,v}} &= \frac{1}{2}\sum_{u\in V(G)}\sum_{v\in \Gamma(u)}\pi(u)P_{u,v}\left(\frac{\mathbb{E}[N_u^{\rightarrow}]^2}{\pi(u)^2}+ \frac{\mathbb{E}[N_v^{\rightarrow}]^2}{\pi(v)^2} - 2\frac{\mathbb{E}[N_u^{\rightarrow}]}{\pi(u)}\frac{\mathbb{E}[N_v^{\rightarrow}]}{\pi(v)}\right)\\
&= \sum_{u\in V(G)}\sum_{v\in \Gamma(u)}\pi(u)P_{u,v} \frac{\mathbb{E}[N_u^{\rightarrow}]^2}{\pi(u)^2}
- \sum_{u\in V(G)}\sum_{v\in \Gamma(u)}\pi(u)P_{u,v} \frac{\mathbb{E}[N_u^{\rightarrow}]}{\pi(u)}\frac{\mathbb{E}[N_v^{\rightarrow}]}{\pi(v)}\\
&= \sum_{u\in V(G)}\sum_{v\in \Gamma(u)}P_{u,v} {\mathbb{E}[N_u^{\rightarrow}]}\left(\frac{\mathbb{E}[N_u^{\rightarrow}]}{\pi(u)}-\frac{\mathbb{E}[N_v^{\rightarrow}]}{\pi(v)}\right).
\end{split}
\end{equation}

The number of visits to $u$ depends on if it is visited at step 1, which happens with probability $\pi(u)$, as well as the number of visits to its neighbours. Every time we leave a neighbour $v$ of $u$, we visit $u$ next with probability $P_{v,u}$. Thus, the expected number of visits to $u$ is:
\begin{align*}
\pi(u)+\sum_{v\in \Gamma(u)}P_{v,u}\mathbb{E}[N_v^{\rightarrow}].
\end{align*}
Then since we leave $u$ if we visit and don't halt, the expected number of times we leave $u\in\overline{M}$ is
\begin{equation}\label{eq:N-leave-notM}
    \mathbb{E}[N_u^{\rightarrow}] = \pi(u)+\sum_{v\in \Gamma(u)}P_{v,u}\mathbb{E}[N_v^{\rightarrow}],
\end{equation}
since we never halt on $u\in \overline{M}$; and for $u\in M$, since every visits ends in halting with probability $p$,
\begin{align}
    \mathbb{E}[N_u^\rightarrow] &= (1-p)\pi(u)+(1-p)\sum_{v\in \Gamma(u)} P_{v,u}\mathbb{E}[N_v^\rightarrow]\nonumber\\
\mbox{so }\mathbb{E}[N_u^\rightarrow] - \sum_{v\in \Gamma(u)}P_{v,u}\mathbb{E}[N_v^\rightarrow] &= \pi(u)-\frac{p}{1-p}\mathbb{E}[N_u^{\rightarrow}]. \label{eq:N-leave-M}
\end{align}
Thus:
\begin{align*}
\sum_{v\in \Gamma(u)}P_{u,v}\left(\frac{\mathbb{E}[N_u^{\rightarrow}]}{\pi(u)}-\frac{\mathbb{E}[N_v^{\rightarrow}]}{\pi(v)}\right) 
&= \frac{\mathbb{E}[N_u^{\rightarrow}]}{\pi(u)}-   \sum_{v\in \Gamma(u)}P_{u,v}\frac{\mathbb{E}[N_v^{\rightarrow}]}{\pi(v)}\\
    &= \frac{1}{\pi(u)}\left(\mathbb{E}[N_u^{\rightarrow}]-\sum_{v\in \Gamma(u)}P_{v,u}\mathbb{E}[N_v^{\rightarrow}]\right) & \mbox{since }\frac{P_{u,v}}{\pi(v)}=\frac{P_{v,u}}{\pi(u)}\\
    &=
    \left\{\begin{array}{ll}
    1 & \mbox{if }u\in\overline{M},\mbox{ by \eq{N-leave-notM}}\\
    1-\frac{p}{1-p}\frac{\mathbb{E}[N_u^\rightarrow]}{\pi(u)} & \mbox{if }u\in M,\mbox{ by \eq{N-leave-M}}.
    \end{array}\right.
\end{align*}
Thus, continuing from \eq{MNRS-big-step}, we have
\begin{align*}
    \sum_{(u,v)\in \overrightarrow{E}(G)}\frac{\theta(u,v)^2}{\pi(u)P_{u,v}} &= \sum_{u\in\overline{M}}\mathbb{E}[N_u^{\rightarrow}] +\sum_{u\in M}\mathbb{E}[N_u^{\rightarrow}]\left(1-\frac{p}{1-p}\frac{\mathbb{E}[N_u^{\rightarrow}]}{\pi(u)}\right)\\
    &\leq \sum_{u\in V(G)}\mathbb{E}[N_u^{\rightarrow}]=\mathbb{E}[\tau],
\end{align*}
establishing \eq{MNRS-eps-delta}.
\end{proof}

We now turn to the special case of search, where we have the following corollary, again assuming, for simplicity, that there are no errors:
\begin{corollary}[Variable-time Search ($\alpha_t=t+1$)]\label{cor:variable-time}
Fix a search space $[n]$ and a distribution $\pi$ on $[n]$ such that we can generate $\ket{\pi}$ in unit cost. Suppose we have a variable-time subroutine that computes, for any $i\in [n]$, a bit $g(i)$, in time $T_i$, which is a random variable on $[{\sf T}]$. Assume all errors are 0. Suppose $\eps>0$ is a lower bound on $\sum_{i:g(i)=1}\pi(i)$ whenever $\bigvee_{i\in [n]}g(i)=1$, and ${\sf T}_{\mathrm{avg}}$ and ${\sf T}_{\mathrm{avg}}'$ are upper bounds such that whenever $\bigvee_{i\in [n]}g(i)=0$:
$$\sqrt{\sum_{i\in [n]}\pi(i)\mathbb{E}[T_i^2]}\leq {\sf T}_{\mathrm{avg}}
\;\mbox{ and }\;
\sqrt{\sum_{i\in [n]}\pi(i)\mathbb{E}[T_i]^2}\leq {\sf T}_{\mathrm{avg}}'.
$$
Then there is a quantum algorithm that decides $\bigvee_{i\in [n]}g(i)$ with bounded error in complexity:
$$O\left(\frac{1}{\sqrt{\eps}}{\sf T}_{\mathrm{avg}} \log^{1.5}{\sf T}\right).$$
If, in addition, we assume the values $\mathbb{E}[T_i]$ are known, in the strong sense that we can generate a superposition proportional to $\sum_{i\in [n]}\sqrt{\pi(i)\mathbb{E}[T_i]}\ket{i}$ in unit cost, there is a quantum algorithm that decides $\bigvee_{i\in [n]}g(i)$ with bounded error in complexity:
$$O\left(\frac{1}{\sqrt{\eps}}{\sf T}_{\mathrm{avg}}'\log{\sf T}\right).$$
\end{corollary}

\begin{proof}
First note that, letting $A\ket{i}=(-1)^{g(i)}$, a variable-time algorithm for $g$ is trivially a reversible variable-time algorithm (just use $\tilde{U}_t=U_t$). We will define a graph $G$ as follows. Let $V(G)=\{u_0,\dots,u_n\}$, where $u_0=(0,0)$ and for all $i\in [n]$, $u_i=(i,g(i))$. Let $V_0=\{u_0\}$, so $\sigma$ is just a point function, and $M=\{u_i:g(i)=1\}$. The graph will be a star with $u_0$ the centre, connected to each $u_i$ by an edge of weight $\pi(i)$.
Define label sets $L(u_0)=L^+(u_0)=[n]$, and for all $i\in [n]$, $L(u_i)=L^-(u_i)=\{\leftarrow\}$, and transition function $f_{u_0}(i)=(i,g(i))$, $f_{u_i}(\leftarrow)=u_0$ for $i\in [n]$. Thus, the variable-time subroutine can be used to implement the edge transition subroutine in this graph. 

We note that it is trivial to generate $\ket{\sigma}=\ket{u_0}$, as well as to check if $(i,b)$ is marked, by checking if $b=1$. Generating star states is trivial for all $\{u_i:i\in [n]\}$, each of which has a single neighbour, and generating $\ket{\psi_\star(u_0)}=\ket{\pi}$ is assumed to have unit cost. Thus, to apply \thm{graph-fwk}, we just need to check the positive and negative conditions. We will use $\alpha_t=t+1$ for the first part of the proof.

\vskip7pt

\noindent\textbf{Positive Condition:} Define $\pi(M):=\sum_{u_i\in M}\pi(i)\geq \eps$. Let $\theta$ be defined as $\theta(u_0,u_i)=\frac{\pi(i)}{\pi(M)}$ for all $u_i\in M$, and 0 otherwise.
Then it is easy to see that $\theta$ satisfies \textbf{P1} and \textbf{P2}. \textbf{P4} is trivially satisfied, because all errors are 0. For \textbf{P3}, we have, using $\alpha_t=t+1$:
$$\sum_{(u,v)\in\overrightarrow{E}(G)}\frac{\theta(u,v)^2}{\w_{u,v}}\mathbb{E}\left[\sum_{t=0}^{T_{u,v}}\frac{1}{t+1}\right]
\leq \frac{2}{\pi(M)^2}\sum_{u_i\in M}\frac{\pi(i)^2}{\pi(i)}\mathbb{E}[\ln T_{u,v}]
\leq \frac{2}{\eps}\ln {\sf T}=:{\cal R}.$$

\vskip7pt

\noindent\textbf{Negative Condition:} \textbf{N2} is trivially satisfied because all errors are 0. For \textbf{N1}, we have: 
$$\sum_{(u,v)\in\overrightarrow{E}(G)}\w_{u,v}\mathbb{E}\left[\sum_{t=0}^{T_{u,v}}(t+1)\right]
\leq \sum_{i\in [n]}\pi(i)\mathbb{E}\left[T_{u,v}^2\right]\leq {\sf T}_{\mathrm{avg}}^2=:{\cal W}.$$
\noindent The stated complexity follows from \thm{graph-fwk}

For the case when the expected stopping times are known, we can set $\w_{u_0,u_i}=\pi(i) \mathbb{E}[T_i]$, but this means that we need to be able to generate a state proportional to
$\ket{\psi_\star(u_0)}=\sum_{i\in [n]}\sqrt{\pi(i)\mathbb{E}[T_i]}\ket{i}$
in unit time, which is precisely what we assume. Then to apply \thm{graph-fwk}, we just need to check the positive and negative conditions. We will now use $\alpha_t=1$.  

\vskip7pt

\noindent\textbf{Positive Condition:} Let $\theta$ be defined as $\theta(u_0,u_i)=\frac{\pi(i)}{\pi(M)}$ for all $u_i\in M$, and 0 otherwise.
Then it is easy to see that $\theta$ satisfies \textbf{P1} and \textbf{P2}. \textbf{P4} is trivially satisfied, because all errors are 0. For \textbf{P3}, we have:
$$\sum_{(u,v)\in\overrightarrow{E}(G)}\frac{\theta(u,v)^2}{\w_{u,v}}\mathbb{E}\left[\sum_{t=0}^{T_{u,v}}1\right]
= \frac{1}{\pi(M)^2}\sum_{u_i\in M}\frac{\pi(i)^2}{\pi(i)\mathbb{E}[T_i]}\mathbb{E}\left[T_i+1\right]
\leq \frac{2}{\eps}=:{\cal R}'.$$

\vskip7pt

\noindent\textbf{Negative Condition:} \textbf{N2} is trivially satisfied because all errors are 0. For \textbf{N1}, we have: 
$$\sum_{(u,v)\in\overrightarrow{E}(G)}\w_{u,v}\mathbb{E}\left[\sum_{t=0}^{T_{u,v}}1\right]
=\sum_{i\in [n]}\pi(i)\mathbb{E}\left[T_{u,v}\right]\mathbb{E}\left[T_{u,v}+1\right]
\leq 2 ({\sf T}_{\mathrm{avg}}')^2 =:{\cal W}'. $$
\noindent The stated complexity follows from \thm{graph-fwk}.
\end{proof}

\noindent We proved the first part of \cor{variable-time} (i.e. $\mathbb{E}[T_i]$ unknown) by applying \thm{graph-fwk} with the setting $\alpha_t=t+1$. 
By setting $\alpha_t=1$, or $\alpha_t=1/(t+1)$, we can get similar results with different complexities, as shown in \tabl{variable-time-search}. The proofs proceed as in the $\alpha_t=t+1$ case, except that we use different flow settings, $\theta$, to upper bound \textbf{P3}. For the $\alpha_t=1$ case, we obtain the expression in \tabl{variable-time-search} row (2) by setting $\theta(u_0,u_i)=\frac{\pi(i)/\mathbb{E}[T_i]}{\sum_{j\in M}\pi(j)/\mathbb{E}[T_j]}$. For the $\alpha_t=1/(t+1)$ case, we obtain the expression in \tabl{variable-time-search} row (3) by setting $\theta(u_0,u_i)=\frac{\pi(i)/\mathbb{E}[T_i^2]}{\sum_{j\in M}\pi(j)/\mathbb{E}[T_j^2]}$. One can verify that these settings of $\theta$ satsify the conditions of \thm{graph-fwk}.
Analogous to the second part of \cor{variable-time}, if the values $\mathbb{E}[T_i]$ are known to the extent that we can efficiently generate a state proportional to $\sum_{i\in [n]}\sqrt{\pi(i)/\mathbb{E}[T_i]}\ket{i}$, which allows the edge weights to be set to $\w_{u_0,u_i}=\pi(i)/\mathbb{E}[T_i]$, then 
by using $\alpha_t=1$ and $\theta(u_0,u_i)=\frac{\pi(i)/\mathbb{E}[T_i]^2}{\sum_{j\in M}\pi(j)/\mathbb{E}[T_j]^2}$, we can get a complexity like the one in (3), but better by a $\sqrt{\log{\sf T}}$ factor, and with $\mathbb{E}[T_i]^2$ in place of $\mathbb{E}[T_i^2]$.

\begin{table}
\centering
\begin{tabular}{|c|c|c|c|}
\hline
& $\alpha_t$ & Complexity & Special case $M=\{m\}$, $\pi(m)\geq \eps$\\
\hline
(1) & $\alpha_t=t+1$ & $\displaystyle \sqrt{\frac{\sum_{i\in [n]}\pi(i)\mathbb{E}[T_i^2]}{\min_{M:M\neq \emptyset}\sum_{i\in M}\pi(i)}}$ 
 & $\displaystyle \sqrt{\frac{1}{\eps}\sum_{i\in [n]}\pi(i)\mathbb{E}[T_i^2]}$
\\
\hline
(2) & $\alpha_t = 1$ & $\displaystyle \sqrt{\frac{\sum_{i\in [n]}\pi(i)\mathbb{E}[T_i]}{\min_{M:M\neq \emptyset}\sum_{i\in M}\frac{\pi(i)}{\mathbb{E}[T_i]}}}$
& $\displaystyle \sqrt{\frac{1}{\eps}\sum_{i\in [n]}\pi(i)\mathbb{E}[T_i]\mathbb{E}[T_m]}$\\
\hline
(3) & $\alpha_t=1/(t+1)$ & $\displaystyle\frac{1}{\sqrt{\min_{M:M\neq \emptyset}\sum_{i\in M}\frac{\pi(i)}{\mathbb{E}[T_i^2]}}}$ & 
$\displaystyle \sqrt{\frac{1}{\eps}\mathbb{E}[T_m^2]}$\\
\hline
\end{tabular}
\caption{Here we show three different versions of \cor{variable-time} that can be obtained by different settings of $\alpha_t$. (1) corresponds to \cor{variable-time}. We neglect log factors. For (3), as with (1) (\cor{variable-time}) if the weights are known, we can get a similar expression with $\mathbb{E}[T_i]^2$ instead of $\mathbb{E}[T_i^2]$, which is better when the variance of $T_i$ is large, and a slightly better dependence on $\log {\sf T}$. 
To allow a comparison of these three complexities, we show the special case where there is a single marked element $M=\{m\}$, promised to have weight at least $\eps$. In that case, we can see that which of the three expressions is optimal depends on how $\mathbb{E}[T_m]$ compares to the average: Suppose for simplicity that the stopping times have 0 variance, so the $T_i$ are just some natural numbers. If $T_m > \sum_{i\in [n]}\pi(i)T_i^2/\sum_{i\in [n]}\pi(i) T_i$, then (1) is the smallest. If $\sum_{i\in [n]}\pi(i) T_i < T_m < \sum_{i\in [n]}\pi(i)T_i^2/\sum_{i\in [n]}\pi(i) T_i$, then (2) is the smallest. And finally, if $T_m<\sum_{i\in [n]}\pi(i) T_i$, then (3) is the smallest.
This does not guarantee that each of the three variable-time search results is useful. This is somehow just saying that if we know $T_m$ is relatively small for the marked $m$, then we should not spend too much time in any subroutine, whereas if we know $T_m$ is relatively large, then we should. Perhaps there would be a simpler way to exploit such information.
}\label{tab:variable-time-search}
\end{table}

We stress that these results could also be applied to other quantum walk search algorithms, for example, in a version of element distinctness in which the query costs are variable. We leave specific applications for future work.

\subsection{Proof of \thm{graph-fwk}}\label{sec:graph-fwk-proof}

We start by defining the parameters of the phase estimation algorithm (see \sec{phase-estimation}) that will prove \thm{graph-fwk}.
Our algorithm will work on the space: 
\begin{multline}\label{eq:graph-fwk-H}
H=\mathrm{span}\left\{\ket{\rightarrow}\ket{u,i}\ket{0,0}\ket{0},\ket{\leftarrow}\ket{u,j}\ket{0,0}\ket{0}:u\in V(G),i\in L^+(u),j\in L^-(u) \right\}\\
\oplus\bigoplus_{t=1}^{\sf T}\mathrm{span}\left\{\ket{\rightarrow}\ket{u,i}\ket{a,z}\ket{t},\ket{\leftarrow}\ket{u,j}\ket{a,z}\ket{t}:u\in V(G),i\in L^+(u),j\in L^-(u),a\in{\cal A},z\in {\cal Z}_{\geq t}\right\}\\
\subset\mathrm{span}\left\{ \ket{d}\ket{u,i}\ket{a,z}\ket{t}:d\in\{\leftarrow,\rightarrow\}, u\in V(G),i\in L(u), a\in{\cal A},z\in{\cal Z},t\in\{0,\dots,{\sf T}\}\right\}.
\end{multline}
We first define the \emph{star states} in $H$. To this end, suppose we obtain a graph $G'$ from $G$ by adding a single new vertex, $v_0$, connected to each vertex $u\in V_0$ by an edge of weight ${\w_0\sigma(u)}$ and each vertex in $M$ by an edge of weight $\w_\M$  (for $\w_0$ and $\w_\M$ to be determined later). We can suppose each of these new edges have the label 0, and so we set:
$$L^+_{G'}(u) = \left\{\begin{array}{ll}
L^+_G(u)\cup\{0\} & \mbox{if }u\in V_0\cup M\\
L^+_G(u) & \mbox{else.}
\end{array}\right.$$
Then we define the star states:
\begin{equation}
\forall u\in V(G),\; \ket{\psi_\star(u)}=\sum_{i\in L_{G'}^+(u)}\sqrt{\w_{u,i}}\ket{\rightarrow}\ket{u,i}\ket{0,0}\ket{0}-\sum_{j\in L_{G}^-(u)}\sqrt{\w_{u,j}}\ket{\leftarrow}\ket{u,j}\ket{0,0}\ket{0}.\label{eq:graph-fwk-star-states}
\end{equation}
Throughout this section, we will let $L^+(u)=L^+_{G'}(u)$, so that we are always talking about the neighbours with respect to $G'$ unless explicitly stated otherwise.

Our algorithm will be based on the star states, as well as the \emph{transition states} $\Psi_0\cup \Psi_1$ (see \defin{transition-states}) of the variable stopping time transition subroutine we are assuming. We recall the form of those states in the setting of the graph transition algorithm, where 
$${\cal I}=\{(u,i):u\in V(G),i\in L^+(i)\}
\;\mbox{ and }\;
A\ket{u,i}=\ket{v,j}$$
where $v=f_u(i)$, and $j=f_v^{-1}(u)$:
\begin{equation}\label{eq:graph-fwk-alg-states}
\begin{split}
\Psi_t^{u,v,\rightarrow}:=\Psi_t^{u,i,\rightarrow}&=\left\{\ket{\psi_{a,z,t}^{u,i,\rightarrow}}=\ket{\rightarrow}\ket{u,i}(\sqrt{\alpha_t}\ket{a,z}\ket{t} - \sqrt{\alpha_{t+1}}U_{t+1}^{u,i}\ket{a,z}\ket{t+1}):a\in{\cal A},z\in{\cal Z}_{>t}\right\}\\
\Psi_t^{u,v,\leftarrow}:=\Psi_t^{u,i,\leftarrow}&=\left\{\ket{\psi_{a,z,t}^{u,i,\leftarrow}}=\ket{\leftarrow}\ket{v,j}(\sqrt{\alpha_t}\ket{a,z}\ket{t} - \sqrt{\alpha_{t+1}}U_{t+1}^{u,i}\ket{a,z}\ket{t+1}):a\in{\cal A},z\in{\cal Z}_{>t}\right\}\\
\Psi_t^{u,v,\leftrightarrow}:=\Psi_t^{u,i,\leftrightarrow}&=\left\{\ket{\psi_{a,z,t}^{u,i,\leftrightarrow}}=\sqrt{\alpha_t}\left(\ket{\rightarrow}\ket{u,i}-\ket{\leftarrow}A_a\ket{u,i}\right)\ket{a,z}\ket{t}:a\in{\cal A},z\in{\cal Z}_{t}\right\}.
\end{split}
\end{equation}
\noindent Also recall from \defin{transition-states}:
\begin{align*}
\Psi_b = \bigcup_{\substack{u\in V(G)\\ i\in L^+(u)}}\bigcup_{\substack{t=0\\ t=b\;\mathrm{mod} 2}}^{{\sf T}-1}\left(\Psi_t^{u,i,\rightarrow}\cup\Psi_t^{u,i,\leftarrow}\right)\cup\bigcup_{\substack{t=1\\ t=b\;\mathrm{mod}2}}^{{\sf T}}\Psi_{t}^{u,i,\leftrightarrow}.
\end{align*}
Then we define:
\begin{equation}
\begin{split}
\Psi^{\cal A} &:=\Psi_0
\qquad\mbox{and}\qquad
\Psi^{\cal B} := \Psi_1\cup \underbrace{\{\ket{\psi_\star(u)}:u\in V(G)\}}_{=:\Psi_\star}.
\end{split}\label{eq:graph-fwk-Psi}
\end{equation}
Each of these sets is pairwise orthogonal. This follows from the fact that each of $\Psi_0$ and $\Psi_1$ is a pairwise orthogonal set (\clm{variable-orthog}), the star states are pairwise orthogonal (easily verified), and each star state is orthogonal to $\Psi_1$ (easily verified). Another way of seeing this orthogonality is that starting from the graph $G$, we can form an orthogonality graph, where we replace each edge of $G$ with a \emph{ladder gadget} consisting of the orthogonality graph in \fig{variable-overlap-graph}. An example of such a graph when $G$ is a triangle is shown in \fig{G-overlap-graph}. Each node of the graph represents a set of pairwise orthogonal vectors, and there is an edge between two nodes if and only if states in the sets overlap. Since this graph is bipartite, and $\Psi^{\cal A}$ and $\Psi^{\cal B}$ are a bipartition (i.e. they are each independent sets) they are each pairwise orthogonal.  

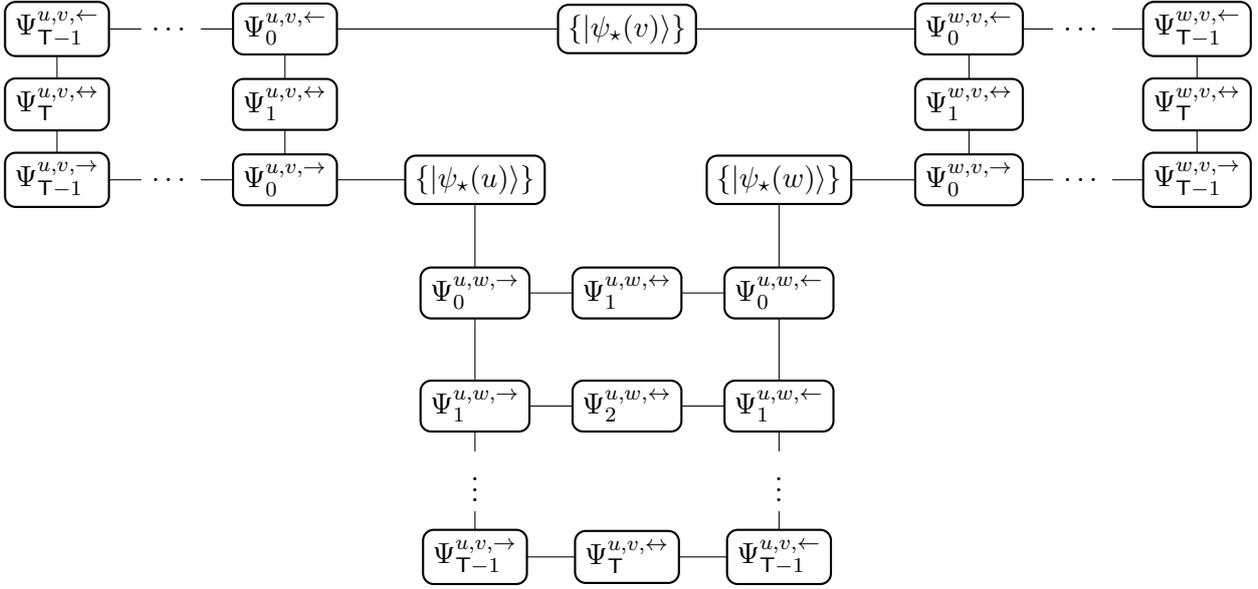
\begin{figure}
\begin{tikzpicture}
\draw (-2,0)--(-7.5,0)--(-7.5,2)--(0,2);
\draw (-4.5,0)--(-4.5,2);

\draw (2,0)--(7.5,0)--(7.5,2)--(0,2);
\draw (4.5,0)--(4.5,2);

\draw(-2,0)--(-2,-5);	\draw(2,0)--(2,-5);
\draw(-2,-1.5)--(2,-1.5);
\draw(-2,-3)--(2,-3);
\draw(-2,-5)--(2,-5);

\node[rectangle, rounded corners, draw, thick, fill=white] at (-4.5,2) {$\Psi_0^{u,v,\leftarrow}$};
\node[rectangle, rounded corners, draw, thick, fill=white] at (-4.5,1) {$\Psi_1^{u,v,\leftrightarrow}$};
\node[rectangle, rounded corners, draw, thick, fill=white] at (-4.5,0) {$\Psi_0^{u,v,\rightarrow}$};

\node[rectangle, fill=white] at (-6,2) {$\dots$};
\node[rectangle, fill=white] at (-6,0) {$\dots$};

\node[rectangle, rounded corners, draw, thick, fill=white] at (-7.5,2) {$\Psi_{{\sf T}-1}^{u,v,\leftarrow}$};
\node[rectangle, rounded corners, draw, thick, fill=white] at (-7.5,1) {$\Psi_{{\sf T}}^{u,v,\leftrightarrow}$};
\node[rectangle, rounded corners, draw, thick, fill=white] at (-7.5,0) {$\Psi_{{\sf T}-1}^{u,v,\rightarrow}$};

\node[rectangle, rounded corners, draw, thick, fill=white] at (0,2) {$\{\ket{\psi_\star(v)}\}$};
\node[rectangle, rounded corners, draw, thick, fill=white] at (-2,0) {$\{\ket{\psi_\star(u)}\}$};
\node[rectangle, rounded corners, draw, thick, fill=white] at (2,0) {$\{\ket{\psi_\star(w)}\}$};

\node[rectangle, rounded corners, draw, thick, fill=white] at (4.5,2) {$\Psi_0^{w,v,\leftarrow}$};
\node[rectangle, rounded corners, draw, thick, fill=white] at (4.5,1) {$\Psi_1^{w,v,\leftrightarrow}$};
\node[rectangle, rounded corners, draw, thick, fill=white] at (4.5,0) {$\Psi_0^{w,v,\rightarrow}$};

\node[rectangle, fill=white] at (6,2) {$\dots$};
\node[rectangle, fill=white] at (6,0) {$\dots$};

\node[rectangle, rounded corners, draw, thick, fill=white] at (7.5,2) {$\Psi_{{\sf T}-1}^{w,v,\leftarrow}$};
\node[rectangle, rounded corners, draw, thick, fill=white] at (7.5,1) {$\Psi_{{\sf T}}^{w,v,\leftrightarrow}$};
\node[rectangle, rounded corners, draw, thick, fill=white] at (7.5,0) {$\Psi_{{\sf T}-1}^{w,v,\rightarrow}$};

\node[rectangle, rounded corners, draw, thick, fill=white] at (-2,-1.5) {$\Psi_0^{u,w,\rightarrow}$};
\node[rectangle, rounded corners, draw, thick, fill=white] at (0,-1.5) {$\Psi_1^{u,w,\leftrightarrow}$};
\node[rectangle, rounded corners, draw, thick, fill=white] at (2,-1.5) {$\Psi_0^{u,w,\leftarrow}$};

\node[rectangle, rounded corners, draw, thick, fill=white] at (-2,-3) {$\Psi_1^{u,w,\rightarrow}$};
\node[rectangle, rounded corners, draw, thick, fill=white] at (0,-3) {$\Psi_2^{u,w,\leftrightarrow}$};
\node[rectangle, rounded corners, draw, thick, fill=white] at (2,-3) {$\Psi_1^{u,w,\leftarrow}$};

\node[rectangle, fill=white] at (-2,-4) {$\vdots$};		\node[rectangle, fill=white] at (2,-4) {$\vdots$};

\node[rectangle, rounded corners, draw, thick, fill=white] at (-2,-5) {$\Psi_{{\sf T}-1}^{u,v,\rightarrow}$};
\node[rectangle, rounded corners, draw, thick, fill=white] at (0,-5) {$\Psi_{{\sf T}}^{u,v,\leftrightarrow}$};
\node[rectangle, rounded corners, draw, thick, fill=white] at (2,-5) {$\Psi_{{\sf T}-1}^{u,v,\leftarrow}$};
\end{tikzpicture}
\caption{If $G$ is a triangle with vertices $u$, $v$ and $w$, and $\protect\overrightarrow{E}(G)=\{(u,v),(w,v),(u,w)\}$, this figure shows the overlap graph for the spaces that make up $\Psi^{\cal A}\cup\Psi^{\cal B}$. If $G$ is any graph, we can get a similar overlap graph by replacing each edge of $G$ with a ``ladder'' gadget, like those shown here. While the height of each ladder is ${\sf T}={\sf T}_{\max}$, if there is some probability of stopping at an earlier time, we can get flow through earlier rungs, thus saving time. The weights of the rungs should correspond to the probability that the algorithm halts at that time.}\label{fig:G-overlap-graph}
\end{figure}

The algorithm in \thm{graph-fwk} will be a phase estimation algorithm, a la \thm{phase-est-fwk}, on a unitary $U_{\cal AB}=(2\Pi_{\cal A}-I)(2\Pi_{\cal B}-I)$, where ${\cal A}=\mathrm{span}\{\Psi^{\cal A}\}$ and ${\cal B}=\mathrm{span}\{\Psi^{\cal B}\}$, on initial state:
\begin{equation}
\ket{\psi_0} = \ket{\rightarrow}\ket{\sigma}\ket{0,0}\ket{0}.\label{eq:graph-fwk-init}
\end{equation}

We have the following corollary of \lem{variable-time-unitary}, and the fact that $\Psi_\star$ can be generated in unit cost:
\begin{corollary}\label{cor:unitary}
The unitary $U_{\cal AB}$ can be implemented in $O(\log {\sf T})$ complexity.
\end{corollary}

\subsubsection{Positive Analysis}

Suppose there is a flow $\theta$ satisfying conditions \textbf{P1-P4} of \thm{graph-fwk}. We can extend $\theta$ to a circulation (see \defin{flow}) on $G'$ by assigning:
$$\forall u\in V_0\cup M,\;\theta(v_0,u) = \theta(u).$$
Then letting $\theta(u,i):=\theta(u,f_u(i))$, we define: 
\begin{equation}
\ket{w}:= \sum_{u\in V(G),i\in L_{G'}^+(u)}\frac{\theta(u,i)}{\sqrt{\w_{u,i}}}\ket{w_+(u,i)},\label{eq:graph-fwk-pos-witness}
\end{equation}
where 
$$\ket{w_+(u,i)}=\left(\ket{\rightarrow}\ket{u,i}+\ket{\leftarrow}\ket{v,j}\right)\sum_{t=0}^{{\sf T}}\frac{1}{\sqrt{\alpha_t}}\ket{w^t(u,i)}\ket{t}$$ 
is the subroutine's positive history state on input $(u,i)$, as defined in \defin{algorithm-states}. We will prove that this is a positive witness. 

\begin{lemma}\label{lem:graph-pos-witness}
The state $\ket{w}$ defined in \eq{graph-fwk-pos-witness} is a $\delta$-positive witness for 
\begin{align*}
\delta &= 3\w_0\sum_{u\in V(G),i\in L^+(u)}\frac{\theta(u,i)^2}{\w_{u,i}}\mathbb{E}\left[\frac{\epsilon_{u,i}^{T_{u,i}}}{\alpha_{T_{u,i}}}\right]. 
\end{align*}
Letting $\w_0={\cal R}^{-1}$, for sufficiently large $\w_\M$, $\ket{w}$ has complexity 
\begin{align*}
\frac{\norm{\ket{w}}^2}{|\braket{\psi_0}{w}|^2} &\leq 2\w_0\sum_{(u,v)\in \overrightarrow{E}(G)}\frac{\theta(u,v)^2}{\w_{u,v}}\mathbb{E}\left[ \sum_{t=0}^{T_{u,i}-1}\frac{1}{\alpha_t} \right]+4 \leq 6=: c_+.
\end{align*}
\end{lemma}
\begin{proof}
We begin by showing that $\ket{w}$ is almost orthogonal to ${\cal A}+{\cal B}$. 
We first note that by \clm{ortho-alg-trans} and the definition of $\ket{w}$, $\ket{w}$ is orthogonal to all $\ket{\psi_{a,z,t}^{\leftarrow,u,i}}$ and all $\ket{\psi_{a,z,t}^{\rightarrow,u,i}}$. We now show that $\ket{w}$ is also orthogonal to all star states. First note that star states have $\ket{0}$ in the last (time) register. 
Referring to \defin{algorithm-states}, and using $A\ket{u,i}=\ket{v,j}$, we have:
$$(I\otimes\ket{0}\bra{0})\ket{w_+(u,i)} = \frac{1}{\sqrt{\alpha_0}}(\ket{\rightarrow}\ket{u,i}+\ket{\leftarrow}\ket{v,j})\ket{w^0(u,i)}\ket{0}.$$
Since $\ket{w^0(u,i)}=\ket{0,0}$ and $\alpha_0=1$, we have, using the notation $\overrightarrow{(v,j)}=(u,i)$ for $i\in L^+(u)$, $v=f_u(i)$, and $j=f_v^{-1}(u)$:
\begin{align*}
\braket{\psi_\star(u)}{w} &= \sum_{i\in L^+(u)}\sqrt{\w_{u,i}}\frac{\theta(u,i)}{\sqrt{\w_{u,i}}}\frac{1}{\sqrt{\alpha_0}}\braket{0,0}{{w}^0(u,i)}+\sum_{j\in L^-(u)}\sqrt{\w_{u,j}}\frac{\theta(u,j)}{\sqrt{\w_{u,j}}}\frac{1}{\sqrt{\alpha_0}}\braket{0,0}{{w}^0\overrightarrow{(u,j)}}\\
&=\sum_{i\in L(u)}\theta(u,i)=0,
\end{align*}
since $\theta$ is a circulation on $G'$. Thus, the only states in $\Psi^{\cal A}\cup\Psi^{\cal B}$ not orthogonal to $\ket{w}$ are the states $\ket{\psi_{a,z,t}^{u,i,\leftrightarrow}}$ in \eq{graph-fwk-alg-states}. We claim that such a state only overlaps $\ket{w}$ in the $\ket{w^+(u,i)}$ part. That is because $\ket{w^+(u,i)}$ has $(\ket{\rightarrow}\ket{u,i}+\ket{\leftarrow}\ket{v,j})$ in the first register, whereas $\ket{\psi_{a,z,t}^{\leftrightarrow,u,i}}$ has $(\ket{\rightarrow}\ket{u,i}-\ket{\leftarrow}A_a\ket{u,i})$, and we are assuming that $\bra{v',j'}A_a\ket{u,i}=0$ for all $a$ whenever $(v',j')\neq (v,j)$, where $v=f_u(i)$ and $j=f_v^{-1}(u)$. Thus:
\begin{align*}
\braket{\psi_{a,z,t}^{u,i,\leftrightarrow}}{w} &= \frac{\theta(u,i)}{\sqrt{\w_{u,i}}}\braket{\psi_{a,z,t}^{u,i,\leftrightarrow}}{w_+(u,i)}
= \frac{\theta(u,i)}{\sqrt{\w_{u,i}}}\left(1-\bra{u,i}A_a^\dagger\ket{v,j}\right)\braket{a,z}{w^t(u,i)}
\end{align*}
by \eq{ortho-alg-trans}. From this we get 
\begin{align*}
|\braket{\psi_{a,z,t}^{u,i,\leftrightarrow}}{w}|^2 &\leq \left\{\begin{array}{ll}
0 & \mbox{if }a = g(u,i)\\
4\frac{\theta(u,i)^2}{{\w_{u,i}}}|\braket{a,z}{w^t(u,i)}|^2 & \mbox{else.}
\end{array}\right.
\end{align*}
Thus, using the fact that reversal states are pairwise orthogonal, we have:
\begin{align*}
\norm{\Pi_{\cal A}\ket{w} + \Pi_{\cal B}\ket{w}}^2 &= \sum_{\substack{u\in V(G),\\ i\in L^+(u)}}\sum_{\substack{a\in {\cal A},\\ t\in \{0,\dots,{\sf T}\},\\ z\in {\cal Z}_t}}\frac{|\braket{\psi_{a,z,t}^{u,i,\leftrightarrow}}{w}|^2}{\norm{\ket{\psi_{a,z,t}^{u,i,\leftrightarrow}}}^2}
= 4\sum_{\substack{u\in V(G),\\ i\in L^+(u)}}\sum_{\substack{a\in{\cal A}\setminus\{g(u,i)\},\\ t\in \{0,\dots,{\sf T}\},\\ z\in {\cal Z}_t}}\frac{\theta(u,i)^2}{\w_{u,i}}\frac{|\braket{a,z}{w^t(u,i)}|^2}{2\alpha_t}.
\end{align*}
Letting $\Lambda_t^{\mathrm{bad}}$ be the projector onto the part of the state that outputs at time $t$ and is incorrect, we have
\begin{align*}
\sum_{\substack{a\in{\cal A}\setminus\{g(u,i)\},\\ z\in {\cal Z}_t}}|\braket{a,z}{w^t(u,i)}|^2=\norm{\Lambda_t^{\mathrm{bad}}\ket{w^t(u,i)}}^2 = \bar{p}_{u,i}(t)\epsilon_{u,i}^t,
\end{align*}
where $\bar{p}_{u,i}(t)$ is the probability that the algorithm outputs at time $t$, and $\epsilon_{u,i}^t$ is the probability that the algorithm errs on input $u,i$ given that it outputs at time $t$. Thus:
\begin{align*}
\norm{\Pi_{\cal A}\ket{w} + \Pi_{\cal B}\ket{w}}^2 
&= 2\sum_{\substack{u\in V(G),\\ i\in L^+(u)}}\frac{\theta(u,i)^2}{\w_{u,i}}\sum_{\substack{t=0}}^{{\sf T}}\frac{\bar{p}_{u,i}(t)\epsilon_{u,i}^t}{\alpha_t}
= 2\sum_{\substack{u\in V(G),\\ i\in L^+(u)}}\frac{\theta(u,i)^2}{\w_{u,i}}\mathbb{E}\left[\frac{\epsilon_{u,i}^{T_{u,i}}}{\alpha_{T_{u,i}}}\right],
\end{align*}
where $T_{u,i}=T_{u,v}$ is the stopping time of the subroutine on input $(u,i)$. 
We can conclude that $\ket{w}$ is a $\delta$-positive witness for any $\delta$ such that
\begin{equation}
\delta \geq \frac{2}{\norm{\ket{w}}^2}\sum_{\substack{u\in V(G),\\ i\in L^+(u)}}\frac{\theta(u,i)^2}{\w_{u,i}}\mathbb{E}\left[\frac{\epsilon_{u,i}^{T_{u,i}}}{\alpha_{T_{u,i}}}\right].\label{eq:graph-pos-delta}
\end{equation}

Next we compute $\norm{\ket{w}}^2$. First we have, 
by \cor{pos-witness-complexity}:
\begin{align*}
\norm{\ket{w_+(u,i)}}^2 &= 2\mathbb{E}\left[ \sum_{t=0}^{T_{u,i}-1}\frac{1}{\alpha_t} \right],
\end{align*}
from which we can compute:
\begin{align}
\norm{\ket{w}}^2 &= 2\sum_{\substack{u\in V(G), i\in L^+_G(u)}}\frac{\theta(u,i)^2}{\w_{u,i}}\mathbb{E}\left[ \sum_{t=0}^{T_{u,i}-1}\frac{1}{\alpha_t} \right] + \sum_{u\in V_0}\frac{\theta(u)^2}{\w_0\sigma(u)}+\sum_{u\in M}\frac{\theta(u)^2}{\w_\M}.\label{eq:graph-pos-wit-compl}
\end{align}
This already gives a lower bound of $\norm{\ket{w}}^2 \geq \frac{1}{\w_0}\sum_{u\in V_0}\frac{\theta(u)^2}{\sigma(u)}\geq \frac{1}{3\w_0}$
by \textbf{P2}, which, combined with \eq{graph-pos-delta}, yields the desired bound on $\delta$. Continuing from \eq{graph-pos-wit-compl}, and using $\sum_{u\in V_0}\frac{\theta(u)^2}{\sigma(u)}\leq 3$ (\textbf{P2}), we have:
\begin{align}
\norm{\ket{w}}^2 &\leq 2\sum_{u\in V(G), i\in L^+_G(u)}\frac{\theta(u,i)^2}{\w_{u,i}}\mathbb{E}\left[ \sum_{t=0}^{T_{u,i}-1}\frac{1}{\alpha_t} \right] + \frac{3}{\w_0}+\frac{1}{\w_0}\label{eq:graph-pos-wit-norm}
\end{align}
for sufficiently large $\w_\M$. Finally, referring to \eq{graph-fwk-init}, we compute
\begin{align*}
\braket{\psi_0}{w} &= \sum_{u\in V_0}\sqrt{\sigma(u)}\frac{\theta(u)}{\sqrt{\w_0\sigma(u)}} =\frac{1}{\sqrt{\w_0}}
\end{align*}
by \textbf{P1}, which, combined with \eq{graph-pos-wit-norm} and \textbf{P3}, yields the desired upper bound on the complexity. 
\end{proof}

\subsubsection{Negative Analysis} 

Define
\begin{equation}
\begin{split}
\ket{w_{\cal A}} &:= -\frac{1}{\sqrt{\w_0}}\sum_{\substack{u\in V(G)\\ i\in L_G^+(u)}}\sqrt{\w_{u,i}}\ket{w_-(u,i)}\\
\ket{w_{\cal B}} &:= -\frac{1}{\sqrt{\w_0}}\sum_{\substack{u\in V(G)\\ i\in L_G^+(u)}}\sqrt{\w_{u,i}}\left((\ket{\rightarrow}\ket{u,i}-\ket{\leftarrow}A\ket{u,i})\ket{0,0}\ket{0}-\ket{w_-(u,i)}\right)
+\frac{1}{\sqrt{\w_0}}\sum_{u\in V(G)}\ket{\psi_\star(u)}\!\!\!\!\!\!
\end{split}\label{eq:graph-neg-wtiness}
\end{equation}
where 
\begin{equation*}
\ket{w_-(u,i)}=(\ket{\rightarrow}\ket{u,i}-\ket{\leftarrow}\ket{v,j})\sum_{t=0}^{{\sf T}}\alpha_t(-1)^t\ket{w^t(u,i)}\ket{t}
\end{equation*}
is the subroutine's negative history state on input $(u,i)$, as defined in \defin{algorithm-states}. We will prove that $\ket{w_{\cal A}},\ket{w_{\cal B}}$ are a negative witness. We start by showing the following.
\begin{claim}\label{clm:graph-neg-witness-error}
\begin{align*}
\norm{(I-\Pi_{\cal A})\ket{w_{\cal A}}}^2 & \leq 2\sum_{u\in V(G),i\in L_G^+(u)}\frac{\w_{u,i}}{\w_0}\mathbb{E}\left[\alpha_{T_{u,i}}\epsilon_{u,i}^{T_{u,i}}\right]\\
\norm{(I-\Pi_{\cal B})\ket{w_{\cal B}}}^2 & \leq 2\sum_{u\in V(G),i\in L_G^+(u)}\frac{\w_{u,i}}{\w_0}\mathbb{E}\left[\alpha_{T_{u,i}}\epsilon_{u,i}^{T_{u,i}}\right].
\end{align*}
\end{claim}
\begin{proof}
Refer to \eq{graph-fwk-Psi} for the definitions of $\Psi^{\cal A}$ and $\Psi^{\cal B}$, from which ${\cal A}$ and ${\cal B}$ are defined. By \clm{neg-witness-error},
\begin{align*}
\norm{(I-\Pi_{\Psi_0})\ket{w_-(u,i)}}^2 \leq 2\mathbb{E}\left[\alpha_{T_{u,i}}\epsilon_{u,i}^{T_{u,i}}\right],
\end{align*}
where $\epsilon_{u,i}^t$ is the probability of erring given that the algorithm stops at time $t$ on input $(u,i)$. Thus, since $\Psi^{\cal A}=\Psi_0$: 
\begin{align*}
\norm{(I-\Pi_{\cal A})\ket{w_{\cal A}}}^2 &= \sum_{u\in V(G),i\in L_G^+(u)}\frac{\w_{u,i}}{\w_0}\norm{(I-\Pi_{\Psi_0})\ket{w_-(u,i)}}^2 
\leq 2\sum_{u\in V(G),i\in L_G^+(u)}\frac{\w_{u,i}}{\w_0}\mathbb{E}\left[\alpha_{T_{u,i}}\epsilon_{u,i}^{T_{u,i}}\right],
\end{align*}
where we used the fact that the $\ket{w_-(u,i)}$ are pairwise orthogonal, because of the state in their first register (see \defin{algorithm-states}). 

We have $\ket{\psi_\star(u)}\in {\cal B}$ for all $u\in V(G)$. Combining this with $\Psi^{\cal B}=\Psi_1\cup\Psi_\star$, we have: \phantom\qedhere
\begin{align*}
\norm{(I-\Pi_{\cal B})\ket{w_{\cal B}}}^2 &= \sum_{u\in V(G),i\in L_G^+(u)}\frac{\w_{u,i}}{\w_0}\norm{(I-\Pi_{\Psi_1})\left((\ket{\rightarrow}\ket{u,i}-\ket{\leftarrow}A\ket{u,i})\ket{0,0}\ket{0}-\ket{w_-(u,i)}\right)}^2\!\!\!\!\!\!\!\!\!\!\!\!\!\!\!\!\!\!\!\!\!\!\!\!\!\!\!\!\!\!\!\!\!\!\!\!\!\!\!\!\!\!\!\!\\
&\leq 2\sum_{u\in V(G),i\in L_G^+(u)}\frac{\w_{u,i}}{\w_0}\mathbb{E}\left[\alpha_{T_{u,i}}\epsilon_{u,i}^{T_{u,i}}\right], & \mbox{by \clm{neg-witness-error}.}\qed
\end{align*}
\end{proof}

\begin{lemma}\label{lem:graph-neg-witness}
The states $\ket{w_{\cal A}},\ket{w_{\cal B}}$ defined in \eq{graph-neg-wtiness} are a $\delta'$ negative witness for
$$\delta'=2\sum_{u\in V(G),i\in L^+(u)}\frac{\w_{u,i}}{\w_0}\mathbb{E}\left[\alpha_{T_{u,i}}\epsilon_{u,i}^{T_{u,i}}\right].$$
Letting $\w_0={\cal R}^{-1}$, as in \lem{graph-pos-witness}, $\ket{w_{\cal A}}$ has complexity:
$\norm{\ket{w_{\cal A}}}^2= 2{\cal RW}=:{\cal C}_-.$
\end{lemma}
\begin{proof}
We first note that:
\begin{align}
\sqrt{\w_0}(\ket{w_{\cal A}}+\ket{w_{\cal B}}) &= -\sum_{u\in V(G),i\in L^+_G(u)}\sqrt{\w_{u,i}}\left(\ket{\rightarrow}\ket{u,i}-\ket{\leftarrow}A\ket{u,i}\right)\ket{0,0}\ket{0}+\sum_{u\in V(G)}\ket{\psi_\star(u)}.\label{eq:graph-neg-wit-sum}
\end{align}
Using the fact that $A\ket{u,i}=\ket{v,j}$ whenever $i\in L^+(u)$, $f_u(i)=v$ and $j=f_v^{-1}(u)$, we have:
\begin{equation}
\begin{split}
& \sum_{u\in V(G),i\in L^+_G(u)}\sqrt{\w_{u,i}}\left(\ket{\rightarrow}\ket{u,i}-\ket{\leftarrow}A\ket{u,i}\right)\\
={}& \sum_{u\in V(G),i\in L^+_G(u)}\sqrt{\w_{u,i}}\ket{\rightarrow}\ket{u,i}-\sum_{u\in V(G),j\in L_G^-(u)}\sqrt{\w_{u,j}}\ket{\leftarrow}\ket{v,j}.
\end{split}\label{eq:graph-neg-wit-sum1}
\end{equation}
On the other hand, referring to \eq{graph-fwk-star-states}:
\begin{equation}
\sum_{u\in V(G)}\ket{\psi_\star(u)} = \sum_{u\in V(G)}\left(\sum_{i\in L^+_{G'}(u)}\sqrt{\w_{u,i}}\ket{\rightarrow}\ket{u,i}-\sum_{j\in L^-_{G}(u)}\sqrt{\w_{u,j}}\ket{\leftarrow}\ket{u,j}\right)\ket{0,0}\ket{0}.\label{eq:graph-neg-wit-sum2}
\end{equation}
Recall that  $G'$ only differs from $G$ in that vertices in $M\cup V_0=V_0$ (since $M=\emptyset$ by assumption) are connected to an additonal vertex $v_0$ by an edge labelled $0\in L^+$. 
Thus, putting \eq{graph-neg-wit-sum1} and \eq{graph-neg-wit-sum2} into \eq{graph-neg-wit-sum}, we have:
we have:
\begin{align*}
\sqrt{\w_0}(\ket{w_{\cal A}}+\ket{w_{\cal B}}) &= \sum_{u\in V(G)}\sum_{i\in L^+_{G'}(u)\setminus L^+_G(u)}\sqrt{\w_{u,i}}\ket{\rightarrow}\ket{u,i}\ket{0,0}\ket{0}\\
&= \sum_{u\in V_0}\sqrt{\w_0\sigma(u)}\ket{\rightarrow}\ket{u,0}\ket{0,0}\ket{0}
= \sqrt{\w_0}\ket{\psi_0}.
\end{align*}
Combining this with \clm{graph-neg-witness-error}, we see that $\ket{w_{\cal A}},\ket{w_{\cal B}}$ form a $\delta'$-negative witness in the sense of \defin{neg-witness}. We upper bound the complexity:
\begin{align*}
\norm{\ket{w_{\cal A}}}^2 &= \frac{1}{\w_0}\sum_{u\in V(G),i\in L^+(u)}\w_{u,i}\norm{\ket{w_-(u,i)}}^2
= \frac{2}{\w_0}\sum_{u\in V(G),i\in L^+(u)}\w_{u,i}\mathbb{E}\left[\sum_{t=0}^{T_{u,i}-1}\alpha_t\right]\leq \frac{2{\cal W}}{\w_0},
\end{align*}
by \cor{pos-witness-complexity}, and condition \textbf{N1} of \thm{graph-fwk}. The result follows.
\end{proof}

\subsubsection{Conclusion of Proof of \thm{graph-fwk}}

Fix $H$ as in \eq{graph-fwk-H}, $\ket{\psi_0}$ as in \eq{graph-fwk-init}, and $\Psi^{\cal A}$ and $\Psi^{\cal B}$ as in \eq{graph-fwk-Psi}. Since 
$\ket{\psi_0}=\ket{\rightarrow}\ket{\sigma}\ket{0,0}\ket{0}$
and we assume we can generate $\ket{\sigma}$ in cost ${\sf S}$, we can generate $\ket{\psi_0}$ in cost ${\sf S}'={\sf S}+O(\log {\sf T})$, since we need to initialize $O(\log {\sf T})$ qubits in addition to $\ket{\sigma}$. By \cor{unitary}, we can implement 
$U_{\cal AB}$ in cost ${\sf A}=\log{\sf T}$. Let
$$c_+ = 6 
\mbox{ and }
{\cal C}_-=2{\cal RW}$$
and let 
$$\delta = 3{\cal R}^{-1}\sum_{(u,v)\in\overrightarrow{E}(G)}\frac{\theta(u,v)^2}{\w_{u,v}}\mathbb{E}\left[\frac{\epsilon_{u,v}^{T_{u,v}}}{\alpha_{T_{u,v}}}\right]
\mbox{ and }
\delta'=2{\cal R}\sum_{(u,v)\in\mathbb{E}(G)}\w_{u,v}\mathbb{E}\left[\alpha_{T_{u,v}}\epsilon_{u,v}^{T_{u,v}}\right]$$
as in \lem{graph-pos-witness} and \lem{graph-neg-witness}.
Then we can verify that 
$$\delta = {\cal R}^{-1}o(1/{\cal W}) \leq \frac{1}{(8c_+)^3\pi^8{\cal C}_-}$$
by condition \textbf{P4} of \thm{graph-fwk}, and
$$\delta' = {\cal R}o(1/{\cal R}) \leq \frac{3}{4}\frac{1}{\pi^4c_+},$$
by \textbf{N2}.
Furthermore:
\begin{description}
\item[Positive Condition:] by \lem{graph-pos-witness}, if $M\neq\emptyset$, there is a $\delta$-positive witness with $\frac{\norm{\ket{w}}^2}{|\braket{w}{\psi_0}|^2}\leq c_+$; and
\item[Negative Condition:] by \lem{graph-neg-witness}, if $M=\emptyset$, there is a $\delta'$-negative witness with $\norm{\ket{w_{\cal A}}}^2\leq {\cal C}_-$.
\end{description}
Thus, by \thm{phase-est-fwk}, there is a quantum algorithm that can distinguish these two cases with bounded error in complexity:
$$O\left({\sf S}'+\sqrt{{\cal C}_-}{\sf A} \right)=O\left({\sf S}+\sqrt{\cal RW}\log{\sf T}\right).$$

\section{Composition in Quantum Algorithms}\label{sec:alg-composition}

Fix $f:\{0,1\}^n\rightarrow\{0,1\}$, and for each $i\in [n]$, fix $g_i:\{0,1\}^m\rightarrow\{0,1\}$. Then we can define $g:[n]\times\{0,1\}^m\rightarrow\{0,1\}$ by $g(i,x)=g_i(x)$, and 
$f\circ g:\{0,1\}^m\rightarrow\{0,1\}$ by $f\circ g (x) = f(g_1(x),\dots,g_n(x))$. Given a bounded error quantum algorithm for $f$ that takes time $\sf L$ and $\sf Q$ queries, and a quantum algorithm for $g$ with sufficiently small error, that takes time at most ${\sf T}={\sf T}_{\max}$, we can compose these algorithms to get a quantum algorithm for the composed function $f\circ g$ that takes time ${\sf L}+{\sf Q}\cdot {\sf T}$. In this section, we present a better way to compose these algorithms when the subroutine's running time varies in $i$, and its own randomness. Before we state our main theorem, we fix some notation.

\paragraph{Outer Algorithm:} We suppose we have a quantum algorithm that computes $f:\{0,1\}^n\rightarrow\{0,1\}$ using quantum queries to an input $g\in\{0,1\}^n$. 
Let $V_1,\dots, V_{\sf L}$ be some unitaries acting on a space:
$$H_{\cal I}\otimes H_{\cal A}\otimes H_{\cal Y}=\mathrm{span}\{\ket{i}\ket{b,y}:i\in [n]\cup\{0\},b\in\{0,1\},y\in{\cal Y}\},$$
where $H_{\cal I}$ is the query register, $H_{\cal A}$ is the answer register, and $H_{\cal Y}$ is the workspace. 
Let ${\cal Q}\subset [{\sf L}]$ be such that:
\begin{itemize}
\item If $\ell\in{\cal Q}$, $V_{\ell+1}={\cal O}_g$, where ${\cal O}_g\ket{i}\ket{b,y}=(-1)^{g_i}\ket{i}\ket{b,y}$.
\item If $\ell\not\in {\cal Q}$, $V_{\ell+1}$ is an input-independent unitary that can be implemented in unit cost.
\end{itemize}
Let ${\sf Q}=|{\cal Q}|$. 
We suppose this algorithm computes $f$ with bounded error $\eps_O$. We assume, without loss of generality, that ${\sf L}$ is even, and that $\ell\in {\cal Q}$ only if $\ell$ is even.

We fix some notation for discussing this outer algorithm. First, define:
\begin{equation}
\begin{split}
\forall \ell\in\{0,\dots,{\sf L}\},\;\ket{w_O^{\ell}} &:= V_\ell\dots V_1\ket{0,0,0},
\end{split}\label{eq:alg-fwk-w_O}
\end{equation}
the algorithm's state at time $\ell$. For $\ell\in \{0,\dots,{\sf L}\}$, $y\in {\cal Y}$, $b\in\{0,1\}$, define:
\begin{equation}
\beta_{i,b,y}^{\ell}:=\braket{i,b,y}{w_O^{\ell}}.\label{eq:alg-comp-beta}
\end{equation}
Of special interest will be the \emph{query weight} on $i$ at time $\ell$, for $\ell\in {\cal Q}$:
\begin{equation}\label{eq:alg-comp-q_i-ell}
\q_{i,\ell}:=\norm{(\ket{i}\bra{i}\otimes I)\ket{w_O^{\ell}}}^2 = \sum_{b\in\{0,1\},y\in{\cal Y}}|\beta_{i,b,y}^{\ell}|^2.
\end{equation}
Then the \emph{average query weight} of $i\in [n]$ is:
\begin{equation}\label{eq:alg-comp-bar-q_i}
\bar\q_{i}:=\frac{1}{\sf Q}\sum_{\ell\in {\cal Q}}\q_{i,\ell}.
\end{equation}
Note that $\sum_{i\in [n]}\bar\q_i = 1$.

Similar to variable-time algorithms, we assume we can implement, in unit cost, the unitary:
$$\sum_{\ell\in\{0,\dots,{\sf L}-1\}\setminus{\cal Q}}\ket{\ell}\bra{\ell}\otimes V_{\ell+1}+\sum_{\ell\in{\cal Q}\cup\{{\sf L}\}}\ket{\ell}\bra{\ell}\otimes I.$$

\paragraph{Inner Subroutine:} In addition to the outer algorithm, we suppose we have an inner variable-time subroutine, as defined in \sec{variable-time}, that computes the function $g(i) = g_i = g_i(x)$ with maximum time ${\sf T}={\sf T}_{\max}$. For $i\in [n]$, let $\epsilon_i$ denote the error on input $i$, so in the notation of \sec{variable-time}, we have:
$$\epsilon_i = \sum_{t=0}^{\sf T}\bar{p}_t(i)\epsilon_i^t = \mathbb{E}[\epsilon_i^{T_i}].$$
As in \sec{variable-time}, we let $T_i$ denote the time at which the inner subroutine stops on input $i$, which is a random variable on $\{1,\dots,{\sf T}\}$. Then our main theorem of this section is the following.

\begin{theorem}\label{thm:alg-composition}
Fix functions $f:\{0,1\}^n\rightarrow\{0,1\}$ and $\{g_i:\{0,1\}^m\rightarrow\{0,1\}\}_{i\in [n]}$, and an outer algorithm for $f$ with time complexity ${\sf L}$, query complexity ${\sf Q}$, and bounded error $\eps_O$; and inner subroutine for $g$ with maximum time ${\sf T}={\sf T}_{\max}$, as described above. 
Let ${\sf T}_{\mathrm{avg}}$ be an upper bound such that:
$$\sum_{i\in [n]}\bar\q_i\mathbb{E}[T_i]\leq {\sf T}_{\mathrm{avg}}$$
and suppose the subroutine's errors satisfy the following condition:
$$\epsilon_{\mathrm{avg}}:=\sum_{i\in [n]}\bar\q_i \epsilon_i \leq \frac{\eta}{{\sf Q}({\sf L}+{\sf Q}{\sf T}_{\mathrm{avg}})}.$$
Then if $\eps_O$ and $\eta$ are sufficiently small constants, there is a quantum algorithm that computes $f\circ g$ with bounded error in complexity
$$O\left(({\sf L}+{\sf Q}{\sf T}_{\mathrm{avg}})\log ({\sf L}{\sf T})\right).$$
\end{theorem}

\noindent The remainder of this section is devoted to proving \thm{alg-composition}.

\subsection{Parameters of the Phase Estimation Algorithm}

We let 
\begin{equation}\label{eq:alg-comp-H}
\begin{split}
H &= \mathrm{span}\{\ket{d}\ket{i}\ket{a,z}\ket{t}\ket{b,y}\ket{\ell}:d\in\{\leftarrow,\rightarrow\},i\in [n], a,b\in\{0,1\}, z\in{\cal Z}, y\in {\cal Y},\\
&\qquad\qquad\qquad\qquad\qquad\qquad\qquad\qquad\qquad\qquad\qquad\qquad\qquad\quad t\in \{0,\dots,{\sf T}\}, \ell\in\{0,\dots,{\sf L}\}\}\\
&\qquad\;\oplus \mathrm{span}\{\ket{\bot}\ket{i}\ket{0,0}\ket{0}\ket{b,y}\ket{\ell}:i\in [n], b\in\{0,1\}, z\in{\cal Z}, y\in {\cal Y},\ell\in\{0,\dots,{\sf L}\}\}\\
&\qquad\quad\oplus \mathrm{span}\{\ket{\circ}\ket{i}\ket{0,0}\ket{0}\ket{b,y}\ket{{\sf L}}:i\in [n], b\in\{0,1\},y\in {\cal Y}\}
\oplus \mathrm{span}\{\ket{\circ}\ket{0}\ket{0,0}\ket{0}\ket{0,0}\ket{0}\}.
\end{split}
\end{equation}

We first define outer transition states, extending the unitaries $V_{\ell}$ to $H$ by tensoring in identities:
\begin{multline}
\forall i\in [n],  \ell\in\{0,\dots,{\sf L}-1\}\setminus {\cal Q}\\
\Psi_{i,\ell}^{O}:=\left\{\ket{\psi_{i,b,y,\ell}^O} := \ket{\bot}\ket{i}\ket{0,0}\ket{0}\ket{b,y}\ket{\ell} - \ket{\bot}V_{\ell+1}(\ket{i}\ket{0,0}\ket{0}\ket{b,y})\ket{\ell+1}: b\in\{0,1\},y\in {\cal Y}\right\}.
\label{eq:alg-fwk-outer-states}
\end{multline}

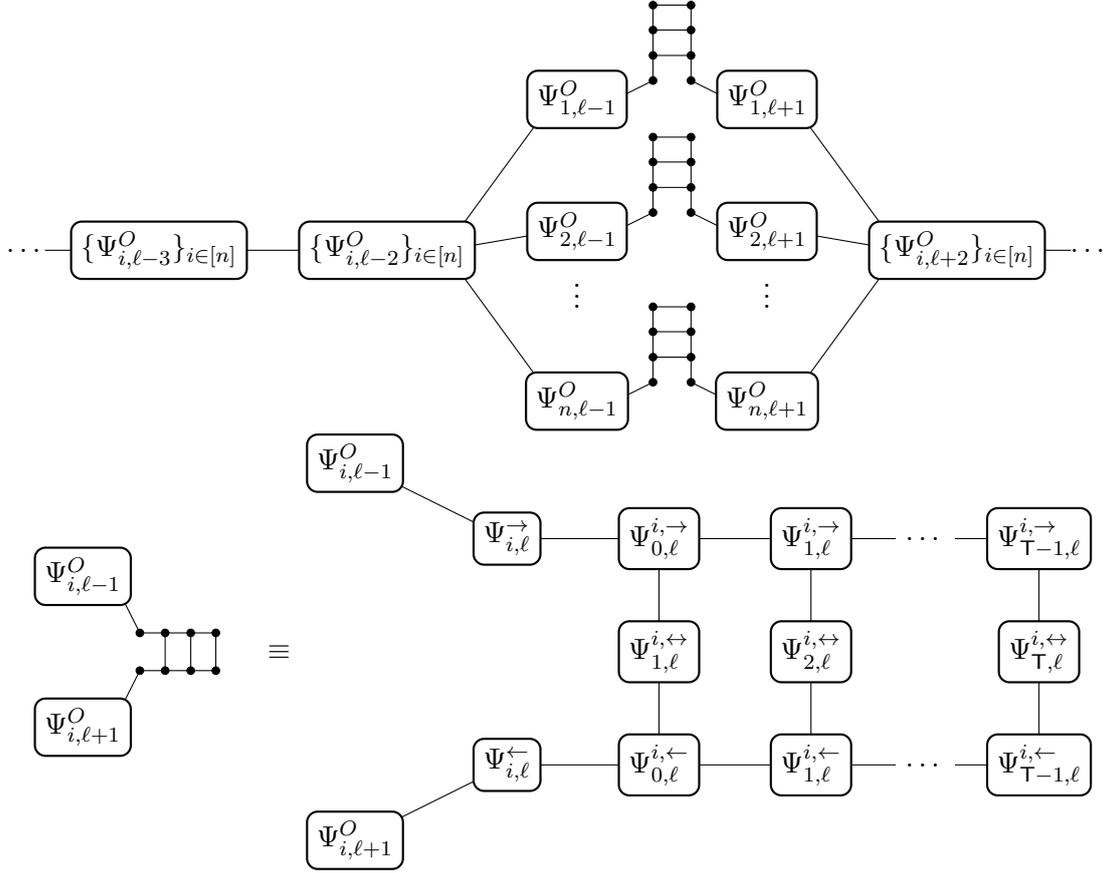
\begin{figure}
\centering
\begin{tikzpicture}
\draw (-1.5,0)--(3.75,0)--(5.2,2);
		\draw (3.75,0)--(5.2,.25);
		\draw (3.75,0)--(5.2,-2);

\draw (6,2)--(6.5,2.25)--(6.5,3.25)--(7,3.25)--(7,2.25)--(7.5,2);
	\filldraw (6.5,3.25) circle (.05);				\filldraw (7,3.25) circle (.05);
	\filldraw (6.5,2.92) circle (.05);	\draw (6.5,2.92)--(7,2.92);	\filldraw (7,2.92) circle (.05);
	\filldraw (6.5,2.583) circle (.05);	\draw (6.5,2.583)--(7,2.583);	\filldraw (7,2.583) circle (.05);
	\filldraw (6.5,2.25) circle (.05);				\filldraw (7,2.25) circle (.05);

\draw (6,.25)--(6.5,.5)--(6.5,1.5)--(7,1.5)--(7,.5)--(7.5,.25);
	\filldraw (6.5,1.5) circle (.05);				\filldraw (7,1.5) circle (.05);
	\filldraw (6.5,1.17) circle (.05);	\draw (6.5,1.17)--(7,1.17);	\filldraw (7,1.17) circle (.05);
	\filldraw (6.5,.833) circle (.05);	\draw (6.5,.833)--(7,.833);	\filldraw (7,.833) circle (.05);
	\filldraw (6.5,.5) circle (.05);				\filldraw (7,.5) circle (.05);

\draw (6,-2)--(6.5,-1.75)--(6.5,-.75)--(7,-.75)--(7,-1.75)--(7.5,-2);
	\filldraw (6.5,-.75) circle (.05);				\filldraw (7,-.75) circle (.05);
	\filldraw (6.5,-1.08) circle (.05);	\draw (6.5,-1.08)--(7,-1.08);	\filldraw (7,-1.08) circle (.05);
	\filldraw (6.5,-1.417) circle (.05);	\draw (6.5,-1.417)--(7,-1.417);	\filldraw (7,-1.417) circle (.05);
	\filldraw (6.5,-1.75) circle (.05);				\filldraw (7,-1.75) circle (.05);

\draw (12,0)--(9.75,0)--(8.3,2);
		\draw (9.75,0)--(8.3,.25);
		\draw (9.75,0)--(8.3,-2);

\node[rectangle, rounded corners, draw, thick, fill=white] at (0,0) {$\{\Psi_{i,\ell-3}^{O}\}_{i\in [n]}$};

\node[rectangle, rounded corners, draw, thick, fill=white] at (3,0) {$\{\Psi_{i,\ell-2}^{O}\}_{i\in [n]}$};

\node[rectangle, rounded corners, draw, thick, fill=white] at (5.5,2) {$\Psi_{1,\ell-1}^{O}$};
\node[rectangle, rounded corners, draw, thick, fill=white] at (5.5,.25) {$\Psi_{2,\ell-1}^{O}$};
\node[rectangle, rounded corners, draw, thick, fill=white] at (5.5,-2) {$\Psi_{n,\ell-1}^{O}$};

\node[rectangle, rounded corners, draw, thick, fill=white] at (8,2) {$\Psi_{1,\ell+1}^{O}$};
\node[rectangle, rounded corners, draw, thick, fill=white] at (8,.25) {$\Psi_{2,\ell+1}^{O}$};
\node[rectangle, rounded corners, draw, thick, fill=white] at (8,-2) {$\Psi_{n,\ell+1}^{O}$};

\node[rectangle, rounded corners, draw, thick, fill=white] at (10.5,0) {$\{\Psi_{i,\ell+2}^{O}\}_{i\in [n]}$};

\node at (-1.75,0) {$\dots$};

\node at (5.5,-.5) {$\vdots$};
\node at (8,-.5) {$\vdots$};

\node at (12.25,0) {$\dots$};
\end{tikzpicture}

\begin{tikzpicture}
\node[rectangle,rotate=270] at (-3,0) {\begin{tikzpicture}
\draw (6,-2)--(6.5,-1.75)--(6.5,-.75)--(7,-.75)--(7,-1.75)--(7.5,-2);
	\filldraw (6.5,-.75) circle (.05);				\filldraw (7,-.75) circle (.05);
	\filldraw (6.5,-1.08) circle (.05);	\draw (6.5,-1.08)--(7,-1.08);	\filldraw (7,-1.08) circle (.05);
	\filldraw (6.5,-1.417) circle (.05);	\draw (6.5,-1.417)--(7,-1.417);	\filldraw (7,-1.417) circle (.05);
	\filldraw (6.5,-1.75) circle (.05);				\filldraw (7,-1.75) circle (.05);
\node[rectangle, rounded corners, draw, thick, fill=white, rotate=90] at (5.75,-2.5) {$\Psi_{i,\ell-1}^O$};
\node[rectangle, rounded corners, draw, thick, fill=white, rotate=90] at (7.75,-2.5) {$\Psi_{i,\ell+1}^O$};
\end{tikzpicture}};

\node at (-1,0) {$\equiv$};

\draw (0,2.5)--(2,1.5)--(9,1.5)--(9,-1.5)--(2,-1.5)--(0,-2.5);
	\draw (4,1.5)--(4,-1.5);
	\draw (6,1.5)--(6,-1.5);

\node[fill=white] at (7.5,1.5) {$\dots$};
\node[fill=white] at (7.5,-1.5) {$\dots$};

\node[rectangle, rounded corners, draw, thick, fill=white] at (0,2.5) {$\Psi_{i,\ell-1}^O$};
\node[rectangle, rounded corners, draw, thick, fill=white] at (0,-2.5) {$\Psi_{i,\ell+1}^O$};

\node[rectangle, rounded corners, draw, thick, fill=white] at (2,1.5) {$\Psi_{i,\ell}^{\rightarrow}$};
\node[rectangle, rounded corners, draw, thick, fill=white] at (2,-1.5) {$\Psi_{i,\ell}^{\leftarrow}$};

\node[rectangle, rounded corners, draw, thick, fill=white] at (4,1.5) {$\Psi_{0,\ell}^{i,\rightarrow}$};
\node[rectangle, rounded corners, draw, thick, fill=white] at (4,0) {$\Psi_{1,\ell}^{i,\leftrightarrow}$};
\node[rectangle, rounded corners, draw, thick, fill=white] at (4,-1.5) {$\Psi_{0,\ell}^{i,\leftarrow}$};

\node[rectangle, rounded corners, draw, thick, fill=white] at (6,1.5) {$\Psi_{1,\ell}^{i,\rightarrow}$};
\node[rectangle, rounded corners, draw, thick, fill=white] at (6,0) {$\Psi_{2,\ell}^{i,\leftrightarrow}$};
\node[rectangle, rounded corners, draw, thick, fill=white] at (6,-1.5) {$\Psi_{1,\ell}^{i,\leftarrow}$};

\node[rectangle, rounded corners, draw, thick, fill=white] at (9,1.5) {$\Psi_{{\sf T}-1,\ell}^{i,\rightarrow}$};
\node[rectangle, rounded corners, draw, thick, fill=white] at (9,0) {$\Psi_{{\sf T},\ell}^{i,\leftrightarrow}$};
\node[rectangle, rounded corners, draw, thick, fill=white] at (9,-1.5) {$\Psi_{{\sf T}-1,\ell}^{i,\leftarrow}$};

\end{tikzpicture}\caption{A graph showing the overlap between the states in $\Psi^{\cal A}\cup \Psi^{\cal B}$. Each node of the graph represents a set of pairwise orthogonal states. An edge between two nodes indicates that the sets are not orthogonal. The states in $\Psi^{\cal A}$ (or $\Psi^{\cal B}$) are an independent set of this graph. Here we show just a part of the graph, for some $\ell\in{\cal Q}$.}\label{fig:alg-fwk-ortho}
\end{figure}

In addition to the outer algorithm, we are assuming we have an inner subroutine, which is a variable-time subroutine $\{U_t^i\}_{i\in [n],t\in [{\sf T}]}$ as in \sec{variable-time}, computing $g(i)=g_i$ -- that is, we have input space ${\cal I}=[n]$, and answer space ${\cal A}=\{0,1\}$. Note that letting $A\ket{i}=(-1)^{g_i}\ket{i}$, and $A_a\ket{i}=(-1)^a\ket{i}$ for $a\in\{0,1\}$, it is easy to verify that this variable-time subroutine \emph{reversibly computes $A$}, simply letting $\tilde{U}_t = U_t$. 
We will use \emph{inner transition states}, which are basically the transition states of the subroutine (see \defin{transition-states}), but we carry around the extra information of the current state of the outer algorithm, $\ket{b,y}\ket{\ell}$ as well. Throughout this section, we will use $\alpha_t=1$ for all $t$, so we have:
\begin{equation}
\begin{split}
&\forall i\in [n],t\in\{0,\dots,{\sf T}-1\},\ell\in {\cal Q}\\
&\quad\Psi_{t,\ell}^{i,\rightarrow}:=\left\{\ket{\psi_{a,z,t}^{i,\rightarrow}}\ket{b,y}\ket{\ell}
= \ket{\rightarrow}\ket{i}\left(\ket{a,z}\ket{t} - U_{t+1}^i\ket{a,z}\ket{t+1} \right)\ket{b,y}\ket{\ell}\right\}_{
\substack{a\in\{0,1\}, z\in {\cal Z}_{>t},\\ b\in\{0,1\},y\in{\cal Y}}}\\
&\quad\Psi_{t,\ell}^{i,\leftarrow}:=\left\{\ket{\psi_{a,z,t}^{i,\leftarrow}}\ket{b,y}\ket{\ell}
= \ket{\leftarrow}\ket{i}\left(\ket{a,z}\ket{t} - U_{t+1}^i\ket{a,z}\ket{t+1} \right)\ket{b,y}\ket{\ell}\right\}_{\substack{a\in\{0,1\}, z\in {\cal Z}_{>t},\\ b\in\{0,1\},y\in{\cal Y}}}\\
&\forall i\in [n],t\in[{\sf T}],\ell\in {\cal Q}\\
&\quad\Psi_{t,\ell}^{i,\leftrightarrow}:=\left\{\ket{\psi_{a,z,t}^{i,\leftrightarrow}}\ket{b,y}\ket{\ell}
= \left(\ket{\rightarrow} - (-1)^a\ket{\leftarrow}\right)\ket{i}\ket{a,z}\ket{t}\ket{b,y}\ket{\ell}\right\}_{\substack{a\in\{0,1\}, z\in {\cal Z}_{t},\\ b\in\{0,1\},y\in{\cal Y}}}.
\end{split}\label{eq:alg-fwk-inner-states}
\end{equation}
We note that for $d\in\{\leftarrow,\rightarrow,\leftrightarrow\}$, we have 
$$\Psi_{t,\ell}^{i,d} = \bigcup_{b\in\{0,1\},y\in{\cal Y}}\Psi_t^{i,d}\otimes \ket{b,y}\ket{\ell},$$
where $\Psi_{t}^{i,d}$ are as in \defin{transition-states}.

\noindent Next, we define \emph{connecting states}, which represent transitioning into and out of the inner subroutine:
\begin{equation}
\begin{split}
& \forall i\in [n], \ell\in {\cal Q}\\
& \qquad\qquad\qquad\Psi_{i,\ell}^{\rightarrow}:=\left\{\ket{\psi_{i,b,y,\ell}^\rightarrow} := \ket{\bot}\ket{i}\ket{0,0}\ket{0}\ket{b,y}\ket{\ell}-\ket{\rightarrow}\ket{i}\ket{0,0}\ket{0}\ket{b,y}\ket{\ell}\right\}_{b\in\{0,1\},y\in {\cal Y}}\\
& \qquad\qquad\qquad\Psi_{i,\ell}^{\leftarrow}:=\left\{\ket{\psi_{i,b,y,\ell}^\leftarrow} := \ket{\leftarrow}\ket{i}\ket{0,0}\ket{0}\ket{b,y}\ket{\ell} - \ket{\bot}\ket{i}\ket{0,0}\ket{0}\ket{b,y}\ket{\ell+1}\right\}_{b\in\{0,1\},y\in {\cal Y}}
\end{split}\label{eq:alg-fwk-connecting-states}
\end{equation}

Analogous to how we let star states for $u\in V_0$ overlap with the initial state, we define an \emph{initial transition state}: 
\begin{equation}
\ket{\psi_{0,v_0}}:= \sqrt{\w_0}\ket{\circ}\ket{0}\ket{0,0}\ket{0}\ket{0,0}\ket{0}-\ket{\bot}\ket{0}\ket{0,0}\ket{0}\ket{0,0}\ket{0}\label{eq:alg-fwk-V1-state}
\end{equation}

Finally, analogous to how we let star states of marked vertices in \sec{graph-composition} have an extra edge coming out of them (to $v_0$), we will consider a state at time $\sf L$ to be marked, but we would like it to be ``more marked'' if it has $b=1$ in the answer register than not, which we reflect through the use of different weights on the extra edge. 
With this intuition, we include the following \emph{terminal states}, for some weights $\w_{1,\mathrm{out}}>\w_{0,\mathrm{out}}$:
\begin{equation}
\begin{split}
\Psi_{\mathrm{out}}:=\left\{\ket{\psi_{i,b,y}^{\mathrm{out}}} := \ket{\bot}\ket{i}\ket{0,0}\ket{0}\ket{b,y}\ket{{\sf L}} - \sqrt{\w_{b,{\mathrm{out}}}}\ket{\circ}\ket{i}\ket{0,0}\ket{0}\ket{b,y}\ket{{\sf L}}: i\in [n],b\in\{0,1\},y\in {\cal Y}\right\}.
\end{split}\label{eq:alg-fwk-M-states}
\end{equation}

Recall that the transition states of the subroutine, $\ket{\psi_{a,z,t}^{i,\rightarrow}}$, $\ket{\psi_{a,z,t}^{i,\leftarrow}}$, and $\ket{\psi_{a,z,t}^{i,\leftrightarrow}}$, can be divided into two pairwise orthogonal sets, $\Psi_0$ and $\Psi_1$, based on the parity of $t$ (see \defin{transition-states}). We define the following two sets of states.

\begin{equation}
\begin{split}
\Psi^{\cal A} ={}& \bigcup_{i\in [n]}\bigcup_{\ell\in\{0,\dots,{\sf L}-1\}\setminus{\cal Q}\;\mathrm{even}}\Psi_{i,\ell}^O
\cup \bigcup_{i\in [n]}\bigcup_{\ell\in{\cal Q}}\left(\Psi_{i,\ell}^{\rightarrow}\cup\Psi_{i,\ell}^{\leftarrow}\right)\cup\Psi_{\mathrm{out}}\cup \bigcup_{\substack{b\in\{0,1\}\\ y\in {\cal Y}}}\bigcup_{\substack{\ell\in {\cal Q}}}\left(\Psi_1\otimes\ket{b,y}\ket{\ell}\right)
\\
\Psi^{\cal B}={} &\bigcup_{i\in [n]}\bigcup_{\ell\in\{0,\dots,{\sf L}-1\}\setminus{\cal Q}\;\mathrm{odd}}\Psi_{i,\ell}^O\cup \{\ket{\psi_{0,v_0}}\}\cup \bigcup_{\substack{b\in\{0,1\}\\ y\in {\cal Y}}}\bigcup_{\substack{\ell\in {\cal Q}}}\left(\Psi_0\otimes\ket{b,y}\ket{\ell}
\right)
.
\end{split}\label{eq:alg-fwk-Psi}
\end{equation}

We claim that the states in $\Psi^{\cal A}$ (or $\Psi^{\cal B}$) are pairwise orthogonal, for which the assumptions that ${\sf L}$ is even, and every $\ell\in {\cal Q}$ is even are necessary. To illustrate this, \fig{alg-fwk-ortho} shows a diagram of which of the states in $\Psi^{\cal A}\cup \Psi^{\cal B}$ overlap.

The algorithm in \thm{alg-composition} will be a phase estimation algorithm, a la \thm{phase-est-fwk}, on a unitary $U_{\cal AB}=(2\Pi_{\cal A}-I)(2\Pi_{\cal B}-I)$, where ${\cal A}=\mathrm{span}\{\Psi^{\cal A}\}$ and ${\cal B}=\mathrm{span}\{\Psi^{\cal B}\}$, on initial state: 
\begin{equation}
\ket{\psi_0}=\ket{\circ}\ket{0}\ket{0,0}\ket{0}\ket{0,0}\ket{0}.\label{eq:alg-fwk-psi-init}
\end{equation}
Then we have the following:
\begin{lemma}\label{lem:alg-comp-unitary}
$U_{\cal AB}$ can be implemented in complexity $O(\log {\sf T}+\log{\sf L})$.
\end{lemma}
\begin{proof}
We describe how to reflect around the states in $\Psi^{\cal A}$ one part at a time. We use the fact that reflecting around a set of pairwise orthogonal states can be reduced to generating the states and then performing a much simpler reflection, as in \lem{variable-time-unitary}.

As in \lem{variable-time-unitary}, we can reflect around $\Psi_1$ in complexity $O(\log {\sf T})$, from which it follows that we can reflect around $\bigcup_{b,y,\ell\in {\cal Q}}(\Psi_1\otimes\ket{b,y}\ket{\ell})$ in this cost. 
By a similar argument to \lem{variable-time-unitary}, we can reflect around the outer algorithm states with $\ell$ even in complexity $O(\log {\sf L})$. 
We can reflect around the two types of connecting states in $O(\log {\sf L})$, using the maps $\ket{\bot}\mapsto \propto\ket{\bot}-\ket{\rightarrow}$ and 
$\ket{\bot}\ket{\ell}\mapsto \propto \ket{\rightarrow}\ket{\ell}-\ket{\bot}\ket{\ell+1}$. 
Finally, we can reflect around the states in $\Psi_{\mathrm{out}}$ using the $O(1)$ complexity map $\ket{\bot}\ket{b}\mapsto \propto\ket{\bot}-\sqrt{\w_{b,\mathrm{out}}}\ket{\circ}$. 

The case for $\Psi^{\cal B}$ is similar, but we note that we can reflect around $\ket{\psi_{0,v_0}}$ using the $O(1)$ complexity map $\ket{\circ}\mapsto \propto \sqrt{\w_0}\ket{\circ}-\ket{\bot}$.
\end{proof}

\subsection{Positive Analysis}

Recall from \defin{algorithm-states}, the positive history states of the inner subroutine (using $\alpha_t=1$):
$$\ket{w_+(i)}=(\ket{\rightarrow}\ket{i}+(-1)^{g_i}\ket{\leftarrow}\ket{i})\sum_{t=0}^{\sf T}\ket{w^t(i)}\ket{t},$$
where $\ket{w^t(i)}$ is the state of the inner subroutine at time $t$ on input $i$. 
Our positive witness will be the following, where $\beta_{i,b,y}^{\ell}=\braket{i,b,y}{w_O^{\ell}}$ are the amplitudes of the outer algorithm states, defined in \eq{alg-comp-beta}:
\begin{multline}
\ket{w} =\frac{1}{\sqrt{\w_0}}\ket{\circ}\ket{0}\ket{0,0}\ket{0}\ket{0,0}\ket{0}+ \ket{\bot}\sum_{\ell=0}^{{\sf L}}\sum_{\substack{i\in [n], b\in\{0,1\}, y\in {\cal Y}}}\beta_{i,b,y}^{\ell}\ket{i}\ket{0,0}\ket{0}\ket{b,y}\ket{\ell}\\
+\sum_{\ell\in{\cal Q}}\sum_{\substack{i\in [n], b\in\{0,1\}, y\in {\cal Y}}}\beta_{i,b,y}^{\ell}\ket{w_+(i)}\ket{b,y}\ket{\ell}
+\ket{\circ}\sum_{\substack{i\in [n], b\in\{0,1\}, y\in {\cal Y}}}\frac{1}{\sqrt{\w_{b,\mathrm{out}}}}\beta_{i,b,y}^{\sf L}\ket{i}\ket{0,0}\ket{0}\ket{b,y}\ket{{\sf L}}.
\label{eq:alt-fwk-pos-witness}
\end{multline}
Ignoring the first and last term, this can be seen (up to algorithm errors) as a history state of the algorithm obtained by composing the inner and outer algorithms. 
\begin{claim}\label{clm:alg-fwk-ortho}
$\ket{w}$ is orthogonal to all outer transition states (see \eq{alg-fwk-outer-states}), connecting states (see \eq{alg-fwk-connecting-states}), terminal states (see \eq{alg-fwk-M-states}), and the initial transition state (see \eq{alg-fwk-V1-state}). 
\end{claim}
\begin{proof}
\noindent\textbf{Outer Transition States:} For outer algorithms states, we fix some $\ell\not\in {\cal Q}$. The first register of $\ket{\psi_{i,b,y,\ell}^O}=\ket{\bot}(\ket{i}\ket{0,0,0}\ket{b,y}\ket{\ell}-V_{\ell+1}\ket{i}\ket{0,0,0}\ket{b,y}\ket{\ell+1})$ is only supported on $\ket{\bot}$, so we have:
\begin{align*}
\braket{\psi_{i,b,y,\ell}^O}{w} &= \bra{\psi_{i,b,y,\ell}^O}\sum_{\ell'=0}^{{\sf L}}\ket{\bot}\sum_{\substack{i'\in [n]\\ b'\in\{0,1\}, y'\in {\cal Y}}}\beta_{i',b',y'}^{\ell'}\ket{i'}\ket{0,0}\ket{0}\ket{b',y'}\ket{\ell}\\
&=\bra{i,b,y}\sum_{\substack{i'\in [n]\\ b'\in\{0,1\}, y'\in {\cal Y}}}\beta_{i',b',y'}^{\ell}\ket{i}\ket{b',y'}-\bra{i,b,y}V_{\ell+1}^\dagger\sum_{\substack{i'\in [n]\\ b'\in\{0,1\}, y'\in {\cal Y}}}\beta_{i',b',y'}^{\ell+1}\ket{i'}\ket{b',y'}\\
&= \braket{i,b,y}{w_{O}^{\ell}}-\bra{i,b,y}V_{\ell+1}^\dagger \ket{w_{O}^{\ell+1}}
= \braket{i,b,y}{w_{O}^{\ell}}-\bra{i,b,y}V_{\ell+1}^\dagger V_{\ell+1}\ket{w_{O}^{\ell}}=0.
\end{align*}

\noindent\textbf{Connecting States:} For any $i\in [n]$, $b\in\{0,1\}$, $y\in{\cal Y}$, $\ell\in{\cal Q}$,
we can see by referring to \eq{alg-fwk-connecting-states} that $\ket{\psi_{i,b,y,\ell}^{\rightarrow}}=(\ket{\bot}-\ket{\rightarrow})\ket{i}\ket{0,0,0}\ket{b,y}\ket{\ell}$, so:
\begin{align*}
\braket{\psi_{i,b,y,\ell}^\rightarrow}{w} &= \beta_{i,b,y}^{\ell} - \beta_{i,b,y}^{\ell}\braket{\rightarrow,i,0,0,0}{w_+(i)}
= \beta_{i,b,y}^{\ell} - \beta_{i,b,y}^{\ell}\braket{0,0}{w^0(i)}=0,
\end{align*}
where we have used $\ket{w^0(i)}=\ket{0,0}$ (see \defin{algorithm-states}).
Similarly, $\ket{\psi_{i,b,y,\ell}^{\leftarrow}}=\ket{\leftarrow}\ket{i}\ket{0,0,0}\ket{b,y}\ket{\ell}-\ket{\bot}\ket{i}\ket{0,0,0}\ket{b,y}\ket{\ell+1}$, so:
\begin{align*}
\braket{\psi_{i,b,y,\ell}^\leftarrow}{w} &= \beta_{i,b,y}^{\ell}\braket{\leftarrow,i,0,0,0}{w_+(i)} - \beta_{i,b,y}^{\ell+1}
=(-1)^{g_i}\beta_{i,b,y}^{\ell}\braket{0,0}{w^0(i)} - \beta_{i,b,y}^{\ell+1}
=(-1)^{g_i}\beta_{i,b,y}^{\ell} - \beta_{i,b,y}^{\ell+1}.
\end{align*}
To see that this is 0, we use the fact that since $\ell\in {\cal Q}$, $V_{\ell+1}={\cal O}_g$ is a query, so: 
\begin{equation}
\begin{split}
\beta_{i,b,y}^{\ell+1} &= \braket{i,b,y}{w_O^{\ell+1}}
=\bra{i,b,y}V_{\ell+1}\ket{w_O^{\ell}} = (-1)^{g_i}\braket{i,b,y}{w_O^{\ell}} = (-1)^{g_i}\beta_{i,b,y}^{\ell}.
\end{split}\label{eq:alg-fwk-beta-query}
\end{equation}

\noindent\textbf{Initial Transition State $\ket{\psi_{0,v_0}}=\sqrt{\w_0}\ket{\circ}\ket{0}\ket{0,0}\ket{0}\ket{0,0}\ket{0}-\ket{\bot}\ket{0}\ket{0,0}\ket{0}\ket{0,0}\ket{0}$:}
\begin{align*}
\braket{\psi_{0,v_0}}{w} &= \sqrt{\w_0}\frac{1}{\sqrt{\w_0}}\braket{\circ,0,0,0,0,0,0,0}{\circ,0,0,0,0,0,0,0} - \beta_{0,0,0}^0=0,
\end{align*}
since $\ket{w_O^0}=\ket{0}\ket{0,0}$, so $\beta_{0,0,0}^{0}=\braket{0,0,0}{w_O^0}=1$.

\noindent\textbf{Terminal States:} For $i\in [n]$, $b\in\{0,1\}$, and $y\in{\cal Y}$, $\ket{\psi_{i,b,y}^{\mathrm{out}}}=(\ket{\bot}-\sqrt{\w_{b,\mathrm{out}}}\ket{\circ})\ket{i}\ket{0,0}\ket{0}\ket{b,y}\ket{{\sf L}}$, so:
\begin{align*}
\braket{\psi_{i,y}^M}{w} &= \braket{\bot}{\bot}\beta_{i,b,y}^{{\sf L}}
-\braket{\circ}{\circ}\sqrt{\w_{b,\mathrm{out}}}\frac{1}{\sqrt{\w_{b,\mathrm{out}}}}\beta_{i,b,y}^{{\sf L}}=0.\qedhere
\end{align*}
\end{proof}

\begin{lemma}[Positive Witness]\label{lem:alg-comp-pos-witness}
Suppose $f\circ g(x)=1$, and let $\ket{w}$ be as in \eq{alt-fwk-pos-witness}, and $\ket{\psi_0}$ as in \eq{alg-fwk-psi-init}. 
Then letting 
$\delta={2\w_0}{\sf Q}\epsilon_{\mathrm{avg}},$
$\ket{w}$ is a $\delta$-positive witness with 
$$\frac{\norm{\ket{w}}^2}{|\braket{w}{\psi_0}|^2}
\leq 1+\w_0\left({{\sf L}+1}+{2}{\sf Q}\left({\sf T}_{\mathrm{avg}}+1\right)+\frac{\eps_O}{\w_{0,\mathrm{out}}}+\frac{1-\eps_O}{\w_{1,\mathrm{out}}}\right).$$
\end{lemma}
\begin{proof}
Let $\Xi_b$ be the orthogonal projector onto states with $b$ in the outer algorithm's answer register, and $p_0$ the probability that the outer algorithm outputs $0$ on input $g(x)$ (assuming ${\cal O}_g$ is implemented with no error), so in particular, $p_0\leq \eps_O$. Then we have:
\begin{align}
\norm{\ket{w}}^2 &= \frac{1}{\w_0}+\sum_{\ell=0}^{{\sf L}}\norm{\ket{w_{O}^\ell}}^2+\sum_{\ell\in {\cal Q}}\sum_{\substack{i\in[n]\\ b\in\{0,1\},y\in{\cal Y}}}|\beta_{i,b,y}^{\ell}|^2\norm{\ket{w_+(i)}}^2+\frac{\norm{\Xi_0\ket{w_{O}^{{\sf L}}}}^2}{\w_{0,\mathrm{out}}}+\frac{\norm{\Xi_1\ket{w_{O}^{{\sf L}}}}^2}{\w_{1,\mathrm{out}}}\!\!\!\!\!\!\!\!\!\!\!\!\!\!\!\!\!\!\!\!\!\!\!\!\!\!\!\!\!\!\!\!\!\!\!\!\label{eq:alg-fwk-w-norm}\\
&= \frac{1}{\w_0}+{\sf L}+1+\sum_{\ell\in {\cal Q}}\sum_{\substack{i\in[n]\\ b\in\{0,1\},y\in{\cal Y}}}|\beta_{i,b,y}^{\ell}|^2 2\mathbb{E}\left[T_i+1\right]+\frac{p_0}{\w_{0,\mathrm{out}}}+\frac{1-p_0}{\w_{1,\mathrm{out}}} & \mbox{\cor{pos-witness-complexity}}\nonumber\\
&\leq \frac{1}{\w_0}+{\sf L}+1+2\sum_{i\in [n]}\sum_{\ell\in {\cal Q}}\q_{i,\ell}\mathbb{E}[T_i+1] + \frac{\eps_O}{\w_{0,\mathrm{out}}}+\frac{1-\eps_O}{\w_{1,\mathrm{out}}} & \mbox{$\w_{1,\mathrm{out}}>\w_{0,\mathrm{out}}$ and \eq{alg-comp-q_i-ell}}\nonumber\\
&=\frac{1}{\w_0}+{\sf L}+1+{2}{{\sf Q}}\sum_{i\in [n]}\bar\q_i (\mathbb{E}[T_i]+1)+ \frac{\eps_O}{\w_{0,\mathrm{out}}}+\frac{1-\eps_O}{\w_{1,\mathrm{out}}} & \mbox{by \eq{alg-comp-bar-q_i}}\nonumber\\
&\leq \frac{1}{\w_0}+{\sf L}+1+2{\sf Q}\left({\sf T}_{\mathrm{avg}}+1\right)+ \frac{\eps_O}{\w_{0,\mathrm{out}}}+\frac{1-\eps_O}{\w_{1,\mathrm{out}}},\nonumber
\end{align}
where in the last line we have used the assumption from the theorem statement that ${\sf T}_{\mathrm{avg}}$ upper bounds the average stopping time. Combined with the observation that $\braket{\psi_0}{w}=\frac{1}{\sqrt{\w_0}}$, we get the desired complexity bound. 

To complete the proof, we upper bound $\norm{\Pi_{\cal A}\ket{w}}^2$ and $\norm{\Pi_{\cal B}\ket{w}}^2$, to show that $\ket{w}$ is a $\delta$-positive witness, as defined in \defin{pos-witness}.
By \clm{alg-fwk-ortho}, the only states in $\Psi^{{\cal A}}\cup \Psi^{{\cal B}}$ (see \eq{alg-fwk-Psi}) that overlap $\ket{w}$ are the inner algorithm states (see \eq{alg-fwk-inner-states}). Since the inner algorithm states do not have $\bot$ or $\circ$ in the first register, we have:
\begin{align*}
\Pi_{{\cal A}}\ket{w} &= \sum_{\substack{b\in\{0,1\}\\ y\in{\cal Y}}}\left(\sum_{\substack{\ell\in {\cal Q}}}\Pi_{\Psi_1}\otimes\ket{b,y,\ell}\bra{b,y,\ell}\right)\sum_{\ell\in{\cal Q}}\sum_{i\in[n]}\beta_{i,b,y}^{\ell}\ket{w_+(i)}\ket{b,y}\ket{\ell}\\
&= \sum_{i\in [n]}\sum_{\substack{\ell\in {\cal Q}}}\Pi_{\Psi_1}\ket{w_+(i)}\otimes\sum_{b\in\{0,1\},y\in{\cal Y}}\beta_{i,b,y}^{\ell}\ket{b,y}\ket{\ell}.
\end{align*}
Then applying \clm{ortho-alg-trans}, which says that (using $\alpha_t=1$)
$$\norm{\Pi_{\Psi_1}\ket{w_+(i)}}^2\leq 2\mathbb{E}\left[{\epsilon_i^{T_i}}\right] = 2\epsilon_i,$$
and since $\ket{w_+(i)}$ are orthogonal for different $i$, we have:
\begin{align*}
\norm{\Pi_{{\cal A}}\ket{w}}^2 &= \sum_{i\in [n]}\sum_{\substack{\ell\in {\cal Q}}}\norm{\Pi_{\Psi_1}\ket{w_+(i)}}^2\sum_{b\in\{0,1\},y\in{\cal Y}}|\beta_{i,b,y}^{\ell}|^2
\leq {2}{\sf Q}\sum_{i\in [n]}\bar\q_{i}\epsilon_i
& \mbox{by \eq{alg-comp-q_i-ell} and \eq{alg-comp-bar-q_i}}\\
&\leq 2{\sf Q}\epsilon_{\mathrm{avg}}\leq 2\w_0{\sf Q}\epsilon_{\mathrm{avg}}\norm{\ket{w}}^2 = \delta\norm{\ket{w}}^2,
\end{align*}
since, by \eq{alg-fwk-w-norm}, $\norm{\ket{w}}^2 \geq 1/\w_0$.
By a nearly identical proof where we simply swap $\Pi_{\Psi_0}$ with $\Pi_{\Psi_1}$, we have a similar upper bound on 
$\norm{\Pi_{\cal B}\ket{w}}^2$, and thus, $\ket{w}$ is a $\delta$-positive witness.  
\end{proof}

\subsection{Negative Analysis} 

We will use the negative history states $\ket{w_-(i)}$ defined in \defin{algorithm-states}, and the amplitudes of the states of the outer algorithm at time $\ell$, $\beta_{i,b,y}=\braket{i,b,y}{w_O^{\ell}}$, defined in \eq{alg-comp-beta}, to define a negative witness (see \defin{neg-witness}). We first define:
\begin{equation}
\begin{split}
\ket{\tilde{w}_{\cal A}} &:= \sum_{\substack{\ell\in\{0,\dots,{\sf L}-1\}\setminus {\cal Q}\\ \mathrm{even}}}\sum_{i,b,y}\beta_{i,b,y}^{\ell}\ket{\psi_{i,b,y,\ell}^O}
+\sum_{\substack{\ell\in {\cal Q}}}\sum_{i,b,y}\beta_{i,b,y}^{\ell}\Biggl(\ket{\psi_{i,b,y,\ell}^{\rightarrow}}+(-1)^{g_i}\ket{\psi_{i,b,y,\ell}^{\leftarrow}}\\
&\qquad\qquad\qquad\qquad -\ket{w_-(i)}\ket{b,y}\ket{\ell}+(\ket{\rightarrow}\ket{i}-(-1)^{g_i}\ket{\leftarrow}\ket{i})\ket{0,0}\ket{0}\ket{b,y}\ket{\ell}\Biggr)\\
\mbox{and }\ket{\tilde{w}_{\cal B}} &:= 
\sum_{\substack{\ell\in\{0,\dots,{\sf L}-1\}\setminus{\cal Q}\\ \mathrm{odd}}}\sum_{i,b,y}\beta_{i,b,y}^{\ell}\ket{\psi_{i,b,y,\ell}^O}
+\sum_{\substack{\ell\in {\cal Q}}}\sum_{i,b,y}\beta_{i,b,y}^{\ell}\ket{w_-(i)}\ket{b,y}\ket{\ell},
\end{split}\label{eq:alg-fwk-neg-witness-tilde}
\end{equation}
where 
the outer transition states $\ket{\psi_{i,b,y,\ell}^O}$ are defined in \eq{alg-fwk-outer-states}, and the connecting states $\ket{\psi_{i,b,y,\ell}^{\rightarrow}},\ket{\psi_{i,b,y,\ell}^{\leftarrow}}$ are defined in \eq{alg-fwk-connecting-states}. Recall that we assume that all $\ell\in {\cal Q}$ are even.

\begin{claim}\label{clm:alg-fwk-delta-prime}$\displaystyle
\norm{(I-\Pi_{\cal A})\ket{\tilde{w}_{\cal A}}}^2 \leq 2{\sf Q}\epsilon_{\mathrm{avg}}
\mbox{ and }
\norm{(I-\Pi_{\cal B})\ket{\tilde{w}_{\cal B}}}^2 \leq 2{\sf Q}\epsilon_{\mathrm{avg}}
.
$
\end{claim}
\begin{proof}
Referring to \eq{alg-fwk-Psi}, we can see that $\ket{\psi_{i,b,y,\ell}^O}\in \Psi^{\cal A}$ whenever $\ell$ is even, and in $\Psi^{\cal B}$ whenever $\ell$ is odd. Furthermore, whenever $\ell\in {\cal Q}$ (and so by assumption $\ell$ is even), $\ket{\psi_{i,b,y,\ell}^{\rightarrow}},\ket{\psi_{i,b,y,\ell}^{\leftarrow}}\in \Psi^{\cal A}$, $\Psi_1\otimes \ket{b,y}\ket{\ell}\subset \Psi^{\cal A}$, and $\Psi_0\otimes\ket{b,y}\ket{\ell}\subset \Psi^{\cal B}$, so:
\begin{align*}
\norm{(I-\Pi_{\cal A})\ket{\tilde{w}_{\cal A}}}^2 &\leq  \sum_{\ell\in {\cal Q}}\sum_{i,b,y}|\beta_{i,b,y}^{\ell}|^2\norm{(I-\Pi_{\Psi_1})\left(\ket{w_-(i)}-(\ket{\rightarrow}\ket{i}-(-1)^{g_i}\ket{\leftarrow}\ket{i})\ket{0,0}\ket{0}\right)}^2\!\!\!\!\!\!\!\!\!\!\!\!\!\!\!\!\!\!\!\!\!\!\!\!\!\!\!\!\!\!\!\!\\
&\leq \sum_{\ell\in {\cal Q}}\sum_{i,b,y}|\beta_{i,b,y}^{\ell}|^22\mathbb{E}[\epsilon_i^{T_i}] & \mbox{by \clm{neg-witness-error}},\\
&= 2\sum_{\ell\in {\cal Q}}\sum_{i\in [n]}\q_{i,\ell}\epsilon_i = 2{\sf Q}\sum_{i\in [n]}\bar\q_i\epsilon_i = 2{\sf Q}\epsilon_{\mathrm{avg}} & \mbox{by \eq{alg-comp-q_i-ell} and \eq{alg-comp-bar-q_i}}.
\end{align*}
Similarly,  
again applying \clm{neg-witness-error}:
\begin{align*}
\norm{(I-\Pi_{\cal B})\ket{\tilde{w}_{\cal B}}}^2 &\leq  \sum_{\ell\in {\cal Q}}\sum_{i,b,y}|\beta_{i,b,y}^{\ell}|^2\norm{(I-\Pi_{\Psi_0})\ket{w_-(i)}}^2
\leq \sum_{\ell\in {\cal Q}}\sum_{i,b,y}|\beta_{i,b,y}^{\ell}|^22\mathbb{E}[\epsilon_i^{T_i}] = 2{\sf Q}\epsilon_{\mathrm{avg}}.\qedhere
\end{align*}
\end{proof}

\begin{claim}\label{clm:alg-fwk-neg-tilde}
Let $\ket{\tilde{w}_{\cal A}}$ and $\ket{\tilde{w}_{\cal B}}$ be as defined in \eq{alg-fwk-neg-witness-tilde}. For $\ell\in\{0,\dots,{\sf L}\}$, define
$$\ket{\bar{w}_O^{\ell}}:=\sum_{i,b,y}\beta_{i,b,y}^{\ell}\ket{i}\ket{0,0}\ket{0}\ket{b,y},$$
which is just the state of the outer algorithm at time $\ell$, $\ket{w_O^{\ell}}$, with extra registers set to 0. 
Then
$$\ket{\tilde{w}_{\cal A}}+\ket{\tilde{w}_{\cal B}} = \ket{\bot}\ket{0}\ket{0,0}\ket{0}\ket{0,0}\ket{0} - \ket{\bot}\ket{\bar{w}_O^{\sf L}}\ket{{\sf L}}.$$
\end{claim}
\begin{proof}
Referring to \eq{alg-fwk-neg-witness-tilde}, we have:
\begin{equation}
\begin{split}
\ket{\tilde w_{\cal A}}+\ket{\tilde w_{\cal B}} &= \sum_{\ell=0}^{{\sf L}-1}\sum_{i,b,y}\beta_{i,b,y}^{\ell}\ket{\psi_{i,b,y,\ell}^O}\\
&+\sum_{\ell\in {\cal Q}}\sum_{i,b,y}\beta_{i,b,y}^{\ell} \left(\ket{\psi_{i,b,y,\ell}^{\rightarrow}}+(-1)^{g_i}\ket{\psi_{i,b,y,\ell}^{\leftarrow}}
+(\ket{\rightarrow}-(-1)^{g_i}\ket{\leftarrow})\ket{i}\ket{0,0}\ket{0}\ket{b,y}\ket{\ell}
\right).
\end{split}\label{eq:tilde-A-B}
\end{equation}
Referring to \eq{alg-fwk-connecting-states}, we can compute, for any $\ell\in{\cal Q}$:
\begin{equation}
\begin{split}
\sum_{i,b,y}\beta_{i,b,y}^{\ell}\ket{\psi_{i,b,y,\ell}^{\rightarrow}}
&= \sum_{i,b,y}\beta_{i,b,y}^{\ell}(\ket{\bot}\ket{i}\ket{0,0}\ket{0}\ket{b,y}\ket{\ell}-\ket{\rightarrow}\ket{i}\ket{0,0}\ket{0}\ket{b,y}\ket{\ell})
=(\ket{\bot}-\ket{\rightarrow})\ket{\bar{w}_O^{\ell}}\ket{\ell}.
\end{split}\label{eq:outer-sum-1}
\end{equation}
Similarly, since $\ell\in {\cal Q}$, $V_{\ell+1}$ is a query, so we have $\beta_{i,b,y}^{\ell+1}=(-1)^{g_i}\beta_{i,b,y}^{\ell}$ (see \eq{alg-fwk-beta-query}). Thus we have, again referring to \eq{alg-fwk-connecting-states}:
\begin{equation}
\begin{split}
\sum_{i,b,y}\beta_{i,b,y}^{\ell}(-1)^{g_i}\ket{\psi_{i,b,y,\ell}^{\leftarrow}}
&= \sum_{i,b,y}\beta_{i,b,y}^{\ell+1}(\ket{\leftarrow}\ket{i}\ket{0,0}\ket{0}\ket{b,y}\ket{\ell}-\ket{\bot}\ket{i}\ket{0,0}\ket{0}\ket{b,y}\ket{\ell+1})\\
&= \ket{\leftarrow}\ket{\bar{w}_O^{\ell+1}}\ket{\ell} - \ket{\bot}\ket{\bar{w}_O^{\ell+1}}\ket{\ell+1}.
\end{split}\label{eq:outer-sum-2}
\end{equation}
Still assuming $\ell\in{\cal Q}$, we have:
\begin{equation}
\begin{split}
&\sum_{i,b,y}\beta_{i,b,y}^{\ell}\left(\ket{\rightarrow}\ket{i}-(-1)^{g_i}\ket{\leftarrow}\ket{i}\right)\ket{0,0}\ket{0}\ket{b,y}\ket{\ell}\\
={}&\ket{\rightarrow}\sum_{i,b,y}\beta_{i,b,y}^{\ell}\ket{i}\ket{0,0}\ket{0}\ket{b,y}\ket{\ell} - \ket{\leftarrow}\sum_{i,b,y}\beta_{i,b,y}^{\ell+1}\ket{i}\ket{0,0}\ket{0}\ket{b,y}\ket{\ell}
= \ket{\rightarrow}\ket{\bar{w}_O^{\ell}}\ket{\ell} - \ket{\leftarrow}\ket{\bar{w}_O^{\ell+1}}\ket{\ell}.
\end{split}\label{eq:outer-sum-3}
\end{equation}
Thus, combining \eq{outer-sum-1}, \eq{outer-sum-2} and \eq{outer-sum-3}, we have that for all $\ell\in {\cal Q}$,
\begin{equation}
\begin{split}
&\sum_{i,b,y}\beta_{i,b,y}^{\ell}\left(\ket{\psi_{i,b,y,\ell}^{\rightarrow}}+(-1)^{g_i}\ket{\psi_{i,b,y,\ell}^{\leftarrow}}
+\left(\ket{\rightarrow}\ket{i}-(-1)^{g_i}\ket{\leftarrow}\ket{i}\right)\ket{0,0}\ket{0}\ket{b,y}\ket{\ell}\right)\\
={}&\ket{\bot}\left(\ket{\bar{w}_O^{\ell}}\ket{\ell}-\ket{\bar{w}_O^{\ell+1}}\ket{\ell+1}\right).
\end{split}\label{eq:outer-sum-queries-neg}
\end{equation}

\noindent Next, referring to \eq{alg-fwk-outer-states}, we can compute:
\begin{equation}
\begin{aligned}
&\sum_{\ell\in \{0,\dots,{\sf L}-1\}\setminus{\cal Q}}\sum_{i,b,y}\beta_{i,b,y}^{\ell}\ket{\psi_{i,b,y,\ell}^O}\\
={}& \sum_{\ell\in \{0,\dots,{\sf L}-1\}\setminus{\cal Q}}\sum_{i,b,y}\beta_{i,b,y}^{\ell}\ket{\bot}\left(\ket{i}\ket{0,0}\ket{0}\ket{b,y}\ket{\ell}-V_{\ell+1}(\ket{i}\ket{0,0}\ket{0}\ket{b,y})\ket{\ell+1}\right)\!\!\!\!\!\!\!\!\!\!\!\!\!\!\!\!\!\!\!\!\\
={}& \ket{\bot}\sum_{\ell\in \{0,\dots,{\sf L}-1\}\setminus{\cal Q}}\ket{\bar{w}_O^{\ell}}\ket{\ell}-\ket{\bot}\sum_{\ell\in \{0,\dots,{\sf L}-1\}\setminus{\cal Q}}V_{\ell+1}\ket{\bar{w}_O^{\ell}}\ket{\ell+1} \\
={}& \ket{\bot}\sum_{\ell\in \{0,\dots,{\sf L}-1\}\setminus{\cal Q}}\ket{\bar{w}_O^{\ell}}\ket{\ell}-\ket{\bot}\sum_{\ell\in \{0,\dots,{\sf L}-1\}\setminus{\cal Q}}\ket{\bar{w}_O^{\ell+1}}\ket{\ell+1}.
\end{aligned}\label{eq:outer-sum-neg}
\end{equation}

Combining \eq{outer-sum-queries-neg} and \eq{outer-sum-neg} in \eq{tilde-A-B}, we can compute:
\begin{align*}
\ket{\tilde w_{\cal A}}+\ket{\tilde w_{\cal B}} &= \ket{\bot}\left(\sum_{\ell\in \{0,\dots,{\sf L}-1\}\setminus{\cal Q}}\left(\ket{\bar{w}_O^{\ell}}\ket{\ell}-\ket{\bar{w}_O^{\ell+1}}\ket{\ell+1}\right)
+ \sum_{\ell\in {\cal Q}}\left(\ket{\bar{w}_O^{\ell}}\ket{\ell}-\ket{\bar{w}_O^{\ell+1}}\ket{\ell+1} \right)\right)\\
&= \sum_{\ell=0}^{{\sf L}-1}\ket{\bar{w}_O^{\ell}}\ket{\ell} - \sum_{\ell=0}^{{\sf L}-1}\ket{\bar{w}_O^{\ell+1}}\ket{\ell+1}
= \ket{\bot}\ket{\bar{w}_O^0}\ket{0} - \ket{\bot}\ket{\bar{w}_O^{{\sf L}}}\ket{{\sf L}}.
\end{align*}
The proof is completed by observing that $\ket{\bar{w}_O^0}=\ket{0}\ket{0,0}\ket{0}\ket{0,0}$ (see \eq{alg-fwk-w_O}).
\end{proof}

We are now ready to define the negative witness. We let:
\begin{equation}
\begin{split}
\ket{w_{\cal A}} &= \frac{1}{\sqrt{\w_0}}\left(\ket{\tilde w_{\cal A}} + \ket{\bot}\ket{\bar{w}_O^{{\sf L}}}\ket{{\sf L}}\right)
\;\mbox{and }\;\ket{w_{\cal B}} = \frac{1}{\sqrt{\w_0}}\left(\ket{\psi_{0,v_0}} + \ket{\tilde w_{\cal B}}\right),
\end{split}\label{eq:alg-fwk-neg-witness}
\end{equation}
where $\ket{\psi_{0,v_0}}$ is defined in \eq{alg-fwk-V1-state}, $\ket{\bar{w}_O^{{\sf L}}}$ is defined in \clm{alg-fwk-neg-tilde}, and $\ket{\tilde w_{\cal A}}$ and $\ket{\tilde w_{\cal B}}$ are defined in \eq{alg-fwk-neg-witness-tilde}. 

\begin{lemma}\label{lem:alg-comp-neg-witness}
Suppose $f(g(x))=0$. Let
$$\delta'= \frac{4}{\w_0}\left(2{\sf Q}\epsilon_{\mathrm{avg}} +\w_{1,\mathrm{out}}\eps_O+\w_{0,\mathrm{out}}(1-\eps_O)\right).$$
Then $\ket{w_{\cal A}}$, $\ket{w_{\cal B}}$ is a $\delta'$-negative witness (see \defin{neg-witness}) with complexity:
$$\norm{\ket{w_{\cal A}}}^2 \leq \frac{4}{\w_0}\left({\sf L}+2{\sf Q}({\sf T}_{\mathrm{avg}}+1)+1\right).$$
\end{lemma}
\begin{proof}
We first verify that $\ket{w_{\cal A}}+\ket{w_{{\cal B}}}=\ket{\psi_0}$ (see \eq{alg-fwk-psi-init}), letting $\ket{0^6}=\ket{0}\ket{0,0}\ket{0}\ket{0,0}$:
\begin{align*}
\sqrt{\w_0}(\ket{w_{\cal A}}+\ket{w_{\cal B}})
&= \ket{\psi_{0,v_0}}+\ket{\bot}\ket{0^6}\ket{0} - \ket{\bot}\ket{\bar{w}_O^{\sf L}}\ket{{\sf L}}+ \ket{\bot}\ket{\bar{w}_O^{{\sf L}}}\ket{{\sf L}} & \mbox{by \clm{alg-fwk-neg-tilde}}\\
&=\sqrt{\w_0}\ket{\circ}\ket{0^6}\ket{0} - \ket{\bot}\ket{0^6}\ket{0} +\ket{\bot}\ket{0^6}\ket{0} & \mbox{by \eq{alg-fwk-V1-state}}\\
&= \sqrt{\w_0}\ket{\psi_0} & \mbox{by \eq{alg-fwk-psi-init}.}
\end{align*}
We thus analyze the error of this negative witness. We have:
\begin{align}
\norm{(I-\Pi_{\cal A})\ket{w_{\cal A}}}^2 &\leq \frac{4}{\w_0}\left(\norm{(I-\Pi_{\cal A})\ket{\tilde w_{\cal A}}}^2 + \norm{(I-\Pi_{\cal A})\ket{\bot}\ket{\bar{w}_O^{{\sf L}}}\ket{{\sf L}}}^2\right).\label{eq:delta-prime}
\end{align}
For the second term, we note that since $\ket{\psi_{i,b,y}^{\mathrm{out}}}\in \Psi^{\cal A}$, we have:
\begin{align*}
(I-\Pi_{\cal A})\ket{\bot}\ket{\bar{w}_O^{{\sf L}}}\ket{{\sf L}} &= (I-\Pi_{\cal A})\left(\ket{\bot}\ket{\bar{w}_O^{{\sf L}}}\ket{{\sf L}}
-\sum_{i,b,y}\beta_{i,b,y}^{{\sf L}}\ket{\psi_{i,b,y}^{\mathrm{out}}}\right)
\end{align*}
and, letting $\Xi_b$ denote the orthogonal projector onto states with $b$ in the outer algorithm's answer register:
\begin{align*}
&\ket{\bot}\ket{\bar{w}_O^{{\sf L}}}\ket{{\sf L}}-\sum_{i,b,y}\beta_{i,b,y}^{{\sf L}}\ket{\psi_{i,b,y}^{\mathrm{out}}}\\
={}&\ket{\bot}\sum_{i,b,y}\beta_{i,b,y}^{\sf L}\ket{i}\ket{0,0}\ket{0}\ket{b,y}\ket{{\sf L}}
 - \sum_{i,b,y}\beta_{i,b,y}^{{\sf L}}\left(\ket{\bot}-\sqrt{\w_{b,\mathrm{out}}}\ket{\circ}\right)\ket{i}\ket{0,0}\ket{0}\ket{b,y}\ket{{\sf L}}\\
={}& \sqrt{\w_{1,\mathrm{out}}}\ket{\circ}\Xi_1\ket{\bar{w}_O^{\sf L}}\ket{{\sf L}}+\sqrt{\w_{0,\mathrm{out}}}\ket{\circ}\Xi_0\ket{\bar{w}_O^{\sf L}}\ket{{\sf L}}.
\end{align*}
Since $f(g)=0$, we have $\norm{\Xi_1\ket{w_O^{\sf L}}}^2\leq \eps_O$, since this is the probability that the outer algorithm outputs 1 (when ${\cal O}_g$ is implemented perfectly). Thus, using $\w_{1,\mathrm{out}}>\w_{0,\mathrm{out}}$:
\begin{equation*}
\norm{(I-\Pi_{\cal A})\ket{\bot}\ket{\bar{w}_O^{{\sf L}}}\ket{{\sf L}}}^2 \leq \w_{1,\mathrm{out}}\norm{\Xi_1\ket{\bar{w}_O^{{\sf L}}}}^2+\w_{0,\mathrm{out}}\left(1-\norm{\Xi_1\ket{\bar{w}_O^{{\sf L}}}}^2\right) \leq \w_{1,\mathrm{out}}\eps_O+\w_{0,\mathrm{out}}(1-\eps_O).
\end{equation*}
Combining this with \clm{alg-fwk-delta-prime}, which gives the upper bound $\norm{(I-\Pi_{\cal A})\ket{\tilde w_{\cal A}}}^2\leq 2{\sf Q}\epsilon_{\mathrm{avg}}$, 
in \eq{delta-prime} we have:
\begin{align*}
\norm{(I-\Pi_{\cal A})\ket{w_{\cal A}}}^2 &\leq \frac{4}{\w_0}\left(2{\sf Q}\epsilon_{\mathrm{avg}} +\w_{1,\mathrm{out}}\eps_O+\w_{0,\mathrm{out}}(1-\eps_O)\right)=\delta'.
\end{align*}
Similarly, since $\ket{\psi_{0,v_0}}\in {\cal B}$, applying \clm{alg-fwk-delta-prime}, we have:
$$\norm{(I-\Pi_{\cal B})\ket{w_{\cal B}}}^2\leq \frac{1}{\w_0}\norm{(I-\Pi_{\cal B})\ket{\tilde{w}_{\cal B}}}^2 \leq \frac{2{\sf Q}}{\w_0}\epsilon_{\mathrm{avg}} < \delta'.
$$

To complete the proof, we give an upper bound on $\norm{\ket{w_{\cal A}}}^2$. By \eq{alg-fwk-neg-witness}, we have:
\begin{equation}\label{eq:alg-fwk-neg-wit-compl}
\norm{\ket{w_{\cal A}}}^2 \leq \frac{4}{\w_0}\left( \norm{\ket{\tilde{w}_{\cal A}}}^2 + \norm{\ket{\bot}\ket{\bar{w}_O^{\sf L}}\ket{{\sf L}}}^2 \right)
=\frac{4}{\w_0}\left( \norm{\ket{\tilde{w}_{\cal A}}}^2 + 1 \right).
\end{equation}
Using the fact (see \eq{alg-fwk-connecting-states}):
\begin{align*}
&\ket{\psi_{i,b,y,\ell}^{\rightarrow}}+(-1)^{g_i}\ket{\psi_{i,b,y,\ell}^{\leftarrow}}+(\ket{\rightarrow}\ket{i}-(-1)^{g_i}\ket{\leftarrow}\ket{i})\ket{0,0}\ket{0}\ket{b,y}\ket{\ell}\\
={}& \ket{\bot}\ket{i}\ket{0,0}\ket{0}\ket{b,y}\ket{\ell} - (-1)^{g_i}\ket{\bot}\ket{i}\ket{0,0}\ket{0}\ket{b,y}\ket{\ell+1} 
\end{align*}
and referring to \eq{alg-fwk-neg-witness-tilde}, we have:
\begin{align*}
\norm{\ket{\tilde w_{\cal A}}}^2 &\leq 
\sum_{\substack{\ell\in\{0,\dots,{\sf L}-1\}\setminus{\cal Q}\\ \mathrm{even}}}\sum_{i,b,y}|\beta_{i,b,y}^{\ell}|^2\norm{\ket{\psi_{i,b,y,\ell}^O}}^2\\
& \quad+ \sum_{\substack{\ell\in{\cal Q}}}\sum_{i,b,y}|\beta_{i,b,y}^{\ell}|^2\norm{\ket{\bot}\ket{i}\ket{0,0}\ket{0}\ket{b,y}(\ket{\ell}-(-1)^{g_i}\ket{\ell+1})-\ket{w_-(i)}\ket{b,y}\ket{\ell}}^2 
\\
&= 
2\sum_{\substack{\ell\in\{0,\dots,{\sf L}-1\}\setminus{\cal Q}\\ \mathrm{even}}}\norm{\ket{w_O^{\ell}}}^2 
+\sum_{\substack{\ell\in{\cal Q}}}\sum_{i\in[n]}\q_{i,\ell}\left(2+\norm{\ket{w_-(i)}}^2\right) & \mbox{by \eq{alg-comp-q_i-ell}}\\
&=
2({\sf L}/2-{\sf Q}) 
+2{\sf Q}
+\sum_{i\in[n]}{\sf Q}\bar\q_i\norm{\ket{w_-(i)}}^2
\leq {\sf L}+\sum_{i\in[n]}{\sf Q}\bar\q_{i}2\mathbb{E}\left[T_i+1\right]
, 
\end{align*}
by \eq{alg-comp-bar-q_i} and \cor{pos-witness-complexity}, using $\alpha_t=1$. Plugging this into \eq{alg-fwk-neg-wit-compl}, we have:
\begin{align*}
\norm{\ket{w_{\cal A}}}^2 &\leq \frac{4}{\w_0}\left({\sf L}+ 2{\sf Q}\left(\sum_{i\in [n]}\bar\q_i\mathbb{E}[T_i]+1\right) + 1\right)
\leq\frac{4}{\w_0}\left({\sf L}+2{\sf Q}({\sf T}_{\mathrm{avg}}+1)+1\right),
\end{align*}
by the assumption of \thm{alg-composition}.
\end{proof}

\subsection{Conclusion of Proof of \thm{alg-composition}}

We set the parameters $\w_0$, $\w_{1,\mathrm{out}}$ and $\w_{0,\mathrm{out}}$ as follows:
$$\w_0:=\frac{1}{{\sf L}+1+2{\sf Q}({\sf T}_{\mathrm{avg}}+1)},\quad
\w_{1,\mathrm{out}}:=\frac{1-\eps_O}{8}\w_0 
,\quad\mbox{and}\quad
\w_{0,\mathrm{out}}:=\frac{\eps_O}{8}\w_0 
.$$
Note that $\w_{1,\mathrm{out}}>\w_{0,\mathrm{out}}$ is satisfied.

We apply \thm{phase-est-fwk} with $H$ as in \eq{alg-comp-H}, $\ket{\psi_0}=\ket{\circ}\ket{0}\ket{0,0}\ket{0}\ket{0,0}\ket{0}$ (as in \eq{alg-fwk-psi-init}), and $\Psi^{\cal A}$ and $\Psi^{\cal B}$ as in \eq{alg-fwk-Psi}. First note that we can easily generate $\ket{\psi_0}$ in cost ${\sf S}=O(1)$.
By \lem{alg-comp-unitary} we can implement $U_{\cal AB}$ in cost ${\sf A}=\log{\sf T}+\log{\sf L}$. Let
$$c_+=18
\quad\mbox{ and }\quad
{\cal C}_- = 4({\sf L}+1+2{\sf Q}({\sf T}_{\mathrm{avg}}+1))^2$$
and let 
\begin{align*}
\delta &= 2\w_0{\sf Q}\epsilon_{\mathrm{avg}} = \frac{2{\sf Q}\epsilon_{\mathrm{avg}}}{{\sf L}+1+2{\sf Q}({\sf T}_{\mathrm{avg}}+1)}\\
\mbox{and }
\delta' &=\frac{4}{\w_0}\left(2{\sf Q}\epsilon_{\mathrm{avg}}+\w_{1,\mathrm{out}}\eps_O+\w_{0,\mathrm{out}}(1-\eps_O) \right)=8{\sf Q}\epsilon_{\mathrm{avg}}({{\sf L}+1+2{\sf Q}({\sf T}_{\mathrm{avg}}+1)})+{\eps_O(1-\eps_O)}
\end{align*}
as in \lem{alg-comp-pos-witness} and \lem{alg-comp-neg-witness}.
Then, using the assumption in the theorem statement that $\epsilon_{\mathrm{avg}}\leq \eta/({\sf Q}({\sf L}+{\sf Q}{\sf T}_{\mathrm{avg}}))$,
 we can verify that:
\begin{align*}
\delta &= 2{\sf Q}\epsilon_{\mathrm{avg}}\frac{4({\sf L}+1+2{\sf Q}({\sf T}_{\mathrm{avg}}+1))^2}{{\sf L}+1+2{\sf Q}({\sf T}_{\mathrm{avg}}+1)}\frac{1}{{\cal C}_-}
{\leq} \frac{2{\sf Q}\eta}{{\sf Q}({\sf L}+{\sf Q}{\sf T}_{\mathrm{avg}})}4({\sf L}+1+2{\sf Q}({\sf T}_{\mathrm{avg}}+1))\frac{1}{{\cal C}_-}
\leq \frac{1}{(8c_+)^3\pi^8}\frac{1}{{\cal C}_-}
\end{align*}
when $\eta$ is a sufficiently small constant, since ${\sf T}_{\mathrm{avg}}\geq 1$; and
\begin{align*}
\delta' &{\leq} 8{\sf Q}\eta\frac{{\sf L}+1+2{\sf Q}({\sf T}_{\mathrm{avg}}+1)}{{\sf Q}({\sf L}+{\sf Q}{\sf T}_{\mathrm{avg}})}+\eps_O(1-\eps_O)\leq \frac{3}{4}\frac{1}{\pi^4c_+}
\end{align*}
for sufficiently small constants $\eps_O$ and $\eta$. 
Furthermore:
\begin{description}
\item[Positive Condition:] By \lem{alg-comp-pos-witness}, if $f(g)=1$, there is a $\delta$-positive witness with 
\begin{align*}
\frac{\norm{\ket{w}}^2}{|\braket{w}{\psi_0}|^2} & \leq 1+\w_0\left({\sf L}+1+2{\sf Q}({\sf T}_{\mathrm{avg}}+1)\right)+\frac{\w_0}{\w_{0,\mathrm{out}}}\eps_O+\frac{\w_0}{\w_{1,\mathrm{out}}}(1-\eps_O)\\
&=1+1+8+8 = 18=c_+.
\end{align*}
\item[Negative Condition:] By \lem{alg-comp-neg-witness}, if $f(g)=0$, there is a $\delta'$-negative witness with $\norm{\ket{w_{\cal A}}}^2 \leq {\cal C}_-$.
\end{description}
Thus, by \thm{phase-est-fwk}, there is a quantum algorithm that decides $f\circ g$ with bounded error in complexity:
$$O({\sf S}+\sqrt{{\cal C}_-}{\sf A})=O(({\sf L}+{\sf Q}{\sf T}_{\mathrm{avg}})(\log {\sf L}+\log{\sf T})).$$

\section{Acknowledgements}

We thank Fr\'ed\'eric Magniez for early discussions about the concept of variable-time quantum walks. We thank Simon Apers, Aleksandrs Belovs, Shelby Kimmel, Jevg\={e}nijs Vihrovs, Ronald de Wolf, Duyal Yolcu and Sebastian Zur for helpful discussions and feedback on these results. We thank the referees for their helpful suggestions and improvements, and Wouter van Bastelaere for pointing out a typo in equation \eq{alg-comp-H}.

\bibliographystyle{alpha}
\bibliography{refs}

\end{document}